\documentclass[11pt,oneside,reqno]{amsart}

\usepackage{ifthen}
\newboolean{showfigures}
\setboolean{showfigures}{true}

\usepackage{algorithm, algpseudocode, amsfonts, amssymb, amsthm, amsmath, array, authblk, bm, braket, caption, changepage, comment, dsfont, enumitem, etoolbox, fullpage, graphicx, mathrsfs, mathtools, microtype, multicol, pythonhighlight, relsize, subfig, subfloat, textcomp, tikz, tikz-cd, url, xpatch}

\usepackage{hyperref}
\hypersetup{
	colorlinks=true,
	linkcolor=blue,
	citecolor=black,
	urlcolor=blue}
\usepackage[noabbrev,capitalise]{cleveref}
\creflabelformat{equation}{#2\textup{#1}#3}

\numberwithin{equation}{section}

\usepackage{colortbl}
\usepackage{xcolor}
\usepackage{pifont}
\definecolor{darkgreen}{rgb}{0.0, 0.5, 0.0}
\newcommand{\cmark}{\textcolor{darkgreen}{\ding{51}}}  % Green checkmark
\newcommand{\xmark}{\textcolor{red}{\ding{55}}}    % Red X

\usepackage[round]{natbib}

\setlist[enumerate]{itemsep=1ex, topsep=-1ex}
\setlist[itemize]{itemsep=0pt, topsep=0pt}

\usepackage[foot]{amsaddr}

\usepackage{titlesec}

\titleformat{\section}
    {\titlerule\vspace{-2ex}}
    {\thesection.}{1ex}{\centering\scshape}
    [\vspace{.5ex}\titlerule]

\titleformat{\subsection}[runin]
  {\normalfont\normalsize\bfseries}{\thesubsection.}{1ex}{\protect\subsectiontitle}

\titleformat{\subsubsection}[runin]
  {\normalfont\normalsize\itshape}{\thesubsubsection.}{1ex}{\protect\subsubsectiontitle}

\titlespacing*{\subsection}{0pt}{0.0\baselineskip}{0.5ex}
\titlespacing*{\subsubsection}{0pt}{0.0\baselineskip}{0.5ex}

\newcommand{\subsectiontitle}[1]{#1.}
\newcommand{\subsubsectiontitle}[1]{#1.}

\newtheorem{theorem}{Theorem}[section]

\newtheoremstyle{style}
  {\baselineskip} % Space above
  {0em} % Space below
  {\itshape} % Body font
  {} % Indent amount
  {\bfseries} % Theorem head font
  {.} % Punctuation after theorem head
  {.5em} % Space after theorem head
  {} % Theorem head spec (can be left empty, meaning `normal')
\theoremstyle{style}

\newtheorem*{theorem*}{Theorem}

\numberwithin{equation}{section}

\newtheorem{condition}{Condition}

\newtheorem*{definition*}{Definition}

\newtheorem{lemma}[theorem]{Lemma}
\newtheorem*{lemma*}{Lemma}
\newtheorem{prop}[theorem]{Proposition}
\newtheorem*{prop*}{Proposition}

\crefname{condition}{Condition}{Conditions}

\algrenewcommand\algorithmicrequire{\textbf{Input:}}
\algrenewcommand\algorithmicensure{\textbf{Output:}}

\DeclareMathOperator{\1}{\mathds{1}}

\DeclareMathOperator*{\argmin}{arg\,min}
\newcommand{\efp}{\mathtt{efp}}
\newcommand{\nhalf}{\lfloor n/2 \rfloor}

\newcommand{\lmax}{\lambda_{\mathrm{max}}}
\newcommand{\lmin}{\lambda_{\mathrm{min}}}
\newcommand{\evipss}{\mathrm{E(FP)}}
\newcommand{\evmb}{\mathrm{E(FP)}}
\newcommand{\evum}{\mathrm{E(FP)}}
\newcommand{\evtarget}{\mathrm{E(FP)}_*}

\title{{\Large {\LARGE I}ntegrated path stability selection}}

\author{{\large O}mar {\large M}elikechi\textsuperscript{1}}
\author{{\large J}effrey {\large W}. {\large M}iller\textsuperscript{1}}
\address{\textsuperscript{1}Department of Biostatistics, Harvard T.H. Chan School of Public Health, Boston, MA}

\begin{document}

\maketitle

\textbf{Publication notice.} This version of the paper has been peer-reviewed and published in the \textit{Journal of the American Statistical Association}. See \url{https://doi.org/10.1080/01621459.2025.2525589}.

\begin{abstract}
Stability selection is a popular method for improving feature selection algorithms. One of its key attributes is that it provides theoretical upper bounds on the expected number of false positives, E(FP), enabling false positive control in practice. However, stability selection often selects few features because existing bounds on E(FP) are relatively loose. In this paper, we introduce a novel approach to stability selection based on integrating stability paths rather than maximizing over them. This yields upper bounds on E(FP) that are much stronger than previous bounds, leading to significantly more true positives in practice for the same target E(FP). Furthermore, our method requires no more computation than the original stability selection algorithm. We demonstrate the method on simulations and real data from two cancer studies.
\end{abstract}

%%%%%%%%%%%%%%%%%%%%%%%%%%%%%%%%%%%%%%%%%%%%%

\vspace{1em}
\section{Introduction}\label{sec:intro}

%%%%%%%%%%%%%%%%%%%%%%%%%%%%%%%%%%%%%%%%%%%%%

Stability selection is a widely used method that uses subsampling to improve feature selection algorithms \citep{mb}. It is attractive due to its generality, simplicity, and theoretical control on the expected number of false positives, E(FP), sometimes called the \textit{per-family error rate}. Despite these favorable qualities, existing theory for stability selection---which heavily informs its implementation---provides relatively weak bounds on E(FP), resulting in a diminished number of true positives \citep{gwas,hofner,binco,big_empirical}.
Stability selection also requires users to specify two of three parameters: the target E(FP), a selection threshold, and the expected number of selected features. Several works have shown that stability selection is sensitive to these choices, making it difficult to tune for good performance \citep{tigress,binco,big_empirical,topk}. 

The limitations of stability selection are illustrated in \cref{fig:limitations}. Here, we simulate data $y_i=\bm{x}_i^\mathtt{T}\bm{\beta}^*+\epsilon_i$ for $i\in\{1,\dots,n\}$, where $n = 200$ is the number of samples, $p=1000$ is the number of features, and $(\bm{x}_1,y_1),\ldots,(\bm{x}_n,y_n)$ are observations of $n$ independent random vectors $(\bm{X}_i,Y_i)$, where $\bm{X}_i\in\mathbb{R}^p$ and $Y_i\in\mathbb{R}$. The features and noise are generated as $X_{ij}\sim\mathcal{N}(0,1)$ independently, and $\epsilon_i\sim\mathcal{N}(0,\sigma^2)$ independently given $\bm{X}_1,\ldots,\bm{X}_n$, where $\sigma^2=\frac{1}{2 n}\sum_{i=1}^n (\bm{x}_i^\mathtt{T}\bm{\beta}^*)^2$ so that the empirical signal-to-noise ratio is $2$. The coefficient vector $\bm{\beta}^*$ has $s=20$ nonzero entries $\beta^*_j\sim\mathrm{Uniform}([-1,-0.5]\cup[0.5,1])$, located at randomly selected $j\in\{1,\ldots,p\}$, and all other entries are $0$. The results in \cref{fig:limitations} are obtained by simulating $100$ random data sets as above and running stability selection using lasso as the base estimator \citep{lasso}. A selected feature is a \textit{true positive} if its corresponding $\bm{\beta}^*$ entry is nonzero, and is a \textit{false positive} otherwise.

\begin{figure}
\includegraphics[width=\textwidth]{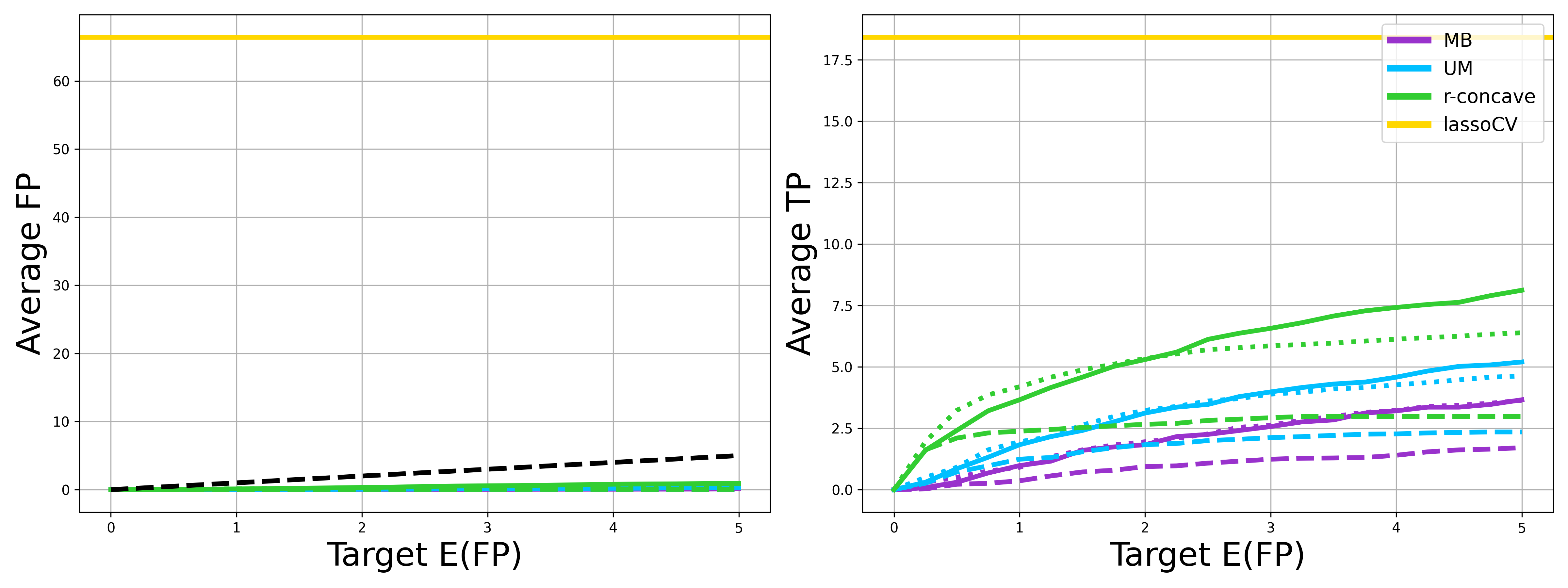}%
\caption{\textit{Tradeoff between FP and TP}. Stability selection is overly conservative, yielding a small number of false positives (FP) at the expense of a small number of true positives (TP).  Meanwhile, lasso has a high TP, but also a very high FP. (\textit{Left}) Average FP versus target E(FP) for the original stability selection method of \cite{mb}, denoted MB, and the unimodal (UM) and $r$-concave methods of \cite{shah} at thresholds $\tau=0.6$, $0.75$, and $0.9$ (solid, dotted, and dashed lines, respectively), as well as lasso with cross-validation, which does not depend on the horizontal axis. The black dashed line is the target E(FP). (\textit{Right}) Average TP versus target E(FP) for the same methods.}
\label{fig:limitations}
\end{figure}

The horizontal axis in \cref{fig:limitations} is the target E(FP). The vertical axes show the actual numbers of false positives (FP) and true positives (TP), averaged over the 100 data sets. The stability selection methods have relatively few true positives on average, and usually 0 false positives, undershooting the target E(FP). On the other hand, lasso with regularization parameter chosen by cross-validation (lassoCV) produces around $18$ true positives, but over $65$ false positives. This is a known trade-off: Typically, stability selection is too conservative, while lasso with cross-validation is not conservative enough \citep{gwas,leng,adaptive_lasso}. Furthermore, while UM and $r$-concave outperform MB, both rely on additional assumptions and have less transparent bounds on E(FP); in particular, $r$-concave requires an additional algorithm to implement and its bound does not admit a closed form, making it difficult to interpret \citep{shah}.

In this article, we introduce \textit{integrated path stability selection} (IPSS) to address these limitations. We prove that IPSS satisfies bounds on $\mathrm{E(FP)}$ that are orders of magnitude stronger than existing stability selection bounds, yielding more true positives for the same target E(FP). This is a key advantage since the actual number of false positives is unknown in practice, so stronger bounds enable one to increase the true positive rate while maintaining control over $\mathrm{E(FP)}$. IPSS is also simple to implement and requires no more computation than stability selection. 

The rest of the article is organized as follows. 
In \cref{sec:background}, we describe our setup, provide background on stability selection, and discuss related work. In \cref{sec:ipss}, we introduce IPSS, and in \cref{sec:theory}, we present our theoretical results.
\cref{sec:simulations} contains an extensive simulation study, and \cref{sec:applications} contains applications of IPSS to prostate and colon cancer. \cref{sec:discussion} concludes with a discussion.

%%%%%%%%%%%%%%%%%%%%%%%%%%%%%%%%%%%%%%%%%%%%%

\section{Stability selection and previous work}\label{sec:background}

%%%%%%%%%%%%%%%%%%%%%%%%%%%%%%%%%%%%%%%%%%%%%

In this section, we define the setup of the problem (\cref{sec:setup}) and describe existing stability selection methods (\cref{sec:ss}) and other related work (\cref{sec:related_work}).

%%%%%%%%%%%%%%%%%%%%%%%%%%%%%%%%%%%%%%%%%%%%%

\subsection{Setup}\label{sec:setup}

%%%%%%%%%%%%%%%%%%%%%%%%%%%%%%%%%%%%%%%%%%%%%

Suppose $S\subseteq \{1,\ldots,p\}$ is an unknown subset to be estimated from observations $\bm{z}_1,\dots,\bm{z}_n$ of $n$ independent and identically distributed (iid) random vectors $\bm{Z}_1,\ldots,\bm{Z}_n$. 
Let $\hat{S}_\lambda(\bm{Z}_{1:n})\subseteq\{1,\ldots,p\}$ be an estimator of $S$, where $\bm{Z}_{1:n} = (\bm{Z}_1,\ldots,\bm{Z}_n)$ and $\lambda > 0$ is a parameter, and let $\hat{S}_\lambda(\bm{z}_{1:n})$ denote the corresponding estimate obtained from the observed data. We refer to $\hat{S}_\lambda$ as the \textit{base estimator}, and allow it to be a random function, such as a stochastic optimization algorithm.
In regression, we have an observation $\bm{z}_i = (\bm{x}_i,y_i)$ of each $\bm{Z}_i=(\bm{X}_i,Y_i)$, where $\bm{x}_i \in \mathbb{R}^p$ is a vector of features and $y_i\in\mathbb{R}$ is a response variable. A canonical base estimator in this setting is the lasso algorithm \citep{lasso}, in which case $\hat{S}_\lambda(\bm{z}_{1:n}) = \{j : \hat{\beta}_j(\lambda) \neq 0\}$ where 
\begin{align}\label{eq:lasso}
	\hat{\bm{\beta}}(\lambda) &= \argmin_{\bm{\beta}\in\mathbb{R}^p}\, \frac{1}{2}\sum_{i = 1}^n (y_i-\bm{x}_i^\mathtt{T}\bm{\beta})^2 + \lambda \sum_{j=1}^p |\beta_j|.
\end{align}
We also consider the minimax concave penalty, MCP \citep{mcp}, the smoothly clipped absolute deviation penalty, SCAD \citep{scad}, the adaptive lasso \citep{adaptive_lasso}, and $\ell_1$-regularized logistic regression \citep{logistic_regression}; see \cref{sup_sec:other_estimators} for details. In an unsupervised learning setting such as graphical lasso, we have $\bm{Z}_i = \bm{X}_i \in\mathbb{R}^p$, without a response variable \citep{glasso}. Additional background on graphical lasso and other examples of base estimators amenable to stability selection can be found in Section 2 of \cite{mb}.

For a given $A \subseteq \{1,\ldots,n\}$, define $\bm{Z}_A = (\bm{Z}_i : i \in A)$ and let $\hat{S}_\lambda(\bm{Z}_A)$ denote the estimator computed using only the data in $\bm{Z}_A$.
A key quantity is the probability that feature $j$ is selected when using half of the data. We denote this \textit{selection probability} by
\begin{align}\label{eq:selection_probabilites}
   \pi_j(\lambda) &= \mathbb{P}\big(j\in\hat{S}_\lambda(\bm{Z}_{1:\nhalf})\big).
\end{align}
Stability selection and IPSS employ estimators of $\pi_j(\lambda)$, called the \textit{estimated selection probabilities} $\hat\pi_j(\lambda)$, computed by repeatedly applying a base estimator to random subsamples of the data (\cref{alg:ss}). The resulting \textit{stability paths} $\lambda\mapsto\hat\pi_j(\lambda)$, which will be important in what follows, are shown in \cref{fig:stability_paths}. Both $\pi_j(\lambda)$ and $\hat\pi_j(\lambda)$ depend on $n$, but $n$ is suppressed from notation since it is always fixed in this work. Notably, all our results are non-asymptotic, applying to any $n\geq 2$.

\begin{algorithm}
\caption{(Estimated selection probabilities)}\label{alg:ss}
\begin{algorithmic}[1]
\Require{Data $\bm{z}_1,\ldots,\bm{z}_n$, base estimator $\hat{S}_\lambda$, parameter grid $\Lambda$, number of iterations $B$.}
\For{$b=1,\ldots,B$}
\State Randomly select disjoint $A_{2b-1},A_{2b}\subseteq\{1,\ldots,n\}$ with $\lvert A_{2b-1}\rvert=\lvert A_{2b}\rvert=\lfloor n/2\rfloor$.
\For{$\lambda\in\Lambda$}
\State Evaluate $\hat{S}_\lambda(\bm{z}_{A_{2b-1}})$ and $\hat{S}_\lambda(\bm{z}_{A_{2b}})$.
\EndFor
\EndFor
\Ensure{Estimated selection probabilities $\hat\pi_j(\lambda)=\frac{1}{2 B}\sum_{b=1}^{2 B} \1\!\big(j \in \hat{S}_\lambda(\bm{z}_{A_b})\big)$.}
\end{algorithmic}
\end{algorithm}

In \cref{alg:ss} and throughout this work, $\1(\cdot)$ denotes the indicator function: $\1(E) = 1$ if $E$ is true and $\1(E) = 0$ otherwise.
Note that $\hat{S}_\lambda$ is evaluated on both disjoint subsets, $A_{2b-1}$ and $A_{2b}$, at each iteration of \cref{alg:ss}. This technique of using complementary pairs of subsets was introduced by \cite{shah}. By contrast, in the original stability selection algorithm of \cite{mb}, $\hat{S}_\lambda$ is applied to only one subset of size $\nhalf$ at each iteration. The difference in empirical performance between the two approaches is minimal, but this slight modification simplifies the assumptions needed for the theory \citep{shah}. 

The choice of $\nhalf$ samples is required for the theory of stability selection to hold, both in this paper and in previous works. An alternative approach is to randomly select $n$ samples with replacement in each of the $B$ subsampling steps \citep{bolasso}. However, this bootstrap approach only guarantees that $S$ is recovered asymptotically in the special case where lasso is the base estimator. By contrast, stability selection and IPSS provide finite sample control of E(FP) for arbitrary base estimators. \cite{shah} also observe that the selection probabilities computed by sampling with replacement are very similar to those computed using \cref{alg:ss}, suggesting that stability selection and IPSS depend little on whether subsampling is implemented with or without replacement.

%%%%%%%%%%%%%%%%%%%%%%%%%%%%%%%%%%%%%%%%%%%%%

\subsection{Stability selection}\label{sec:ss}

%%%%%%%%%%%%%%%%%%%%%%%%%%%%%%%%%%%%%%%%%%%%%

Once the $\hat{\pi}_j$ values are computed, the set of features selected by stability selection \citep{mb,shah} is
\begin{align}\label{eq:ss}
    \hat{S}_{\mathrm{SS}} &= \Big\{j : \max_{\lambda\in\Lambda} \hat{\pi}_j(\lambda) \geq \tau\Big\}
\end{align}
where $\Lambda = [\lmin,\lmax]\subseteq (0,\infty)$ is an interval defined below and $\tau \in (0,1)$ is a user-specified threshold. The upper endpoint $\lmax$ is inconsequential provided it is large enough that all features have small selection probability, which is easy to determine empirically.
Choosing $\lmin$ is considerably more subtle, since many or even all features satisfy $\hat{\pi}_j(\lambda)\geq\tau$ as $\lambda \to 0$. 
While there is no consensus on how to choose $\lmin$ \citep{cox,topk}, a standard approach is to use theoretical upper bounds on E(FP) as follows.

The MB, UM, and r-concave versions of stability selection all satisfy theoretical upper bounds of the form $\mathrm{E(FP)} \leq \mathcal{B}(q,\tau)$, where $\mathcal{B}(q,\tau)$ is an expression that depends on the method (MB, UM, or $r$-concave), the average number of features selected over $\Lambda$, $q = \mathrm{E}\big\lvert\textstyle{\bigcup}_{\lambda\in\Lambda}\hat{S}_{\lambda}(\bm{Z}_{1:\nhalf})\big\rvert$, and the threshold, $\tau$; see \cref{sec:compare}. To determine $\lmin$, two of the following three quantities must be specified: (i) the target $\mathrm{E(FP)}$, denoted $\evtarget$, (ii) the threshold, $\tau$, and (iii) the target number of features selected, $q_*$. The third quantity is then obtained by setting $\evtarget = \mathcal{B}(q_*,\tau)$ and solving. Once $q_*$ is determined, $\lmin = \sup\big\{\lambda\in(0,\lmax) : \mathrm{E}\big\lvert\textstyle{\bigcup}_{\lambda'\in[\lambda,\lmax]}\hat{S}_{\lambda'}(\bm{Z}_{1:\nhalf})\big\rvert\geq q_*\big\}$ is empirically estimated and $\Lambda = [\lmin,\lmax]$ is used in \cref{eq:ss}.

The above construction elucidates some of the shortcomings of stability selection and motivates our formulation of IPSS. First, the inequalities $\mathrm{E(FP)} \leq \mathcal{B}(q,\tau)$ are replaced by equalities in order to determine $\lmin$ in a way that controls FP. Thus, while the recommended procedure does typically keep the actual FP smaller than $\evtarget$, it may be much smaller, as shown in \cref{fig:limitations,fig:bounds}. This overconservative tendency leads to a lower TP than necessary. More precisely, weak bounds on E(FP) lead to large values of $\lmin$, which prevent true features from being selected because their stability paths have not yet distinguished themselves from the noise (\cref{fig:stability_paths}). Second, $\mathrm{E(FP)}_*$, $q_*$, and $\tau$ are interdependent, making it difficult to select these parameters in practice. \cite{mb} recommended taking $\tau\in[0.6,0.9]$, but stability selection is sensitive to $\tau$ even when restricted to this interval \citep{binco,big_empirical}. Nevertheless, $\tau$ must be specified in most cases because one usually has little \textit{a priori} knowledge to inform the choice of $q_*$. Finally, while one can in principle use a smaller $\lmin$, it is unclear what value to choose and doing so would invalidate the E(FP) control guarantee, making it hard to interpret the results.

\begin{figure*}
\includegraphics[width=\textwidth, height=0.30\textheight]{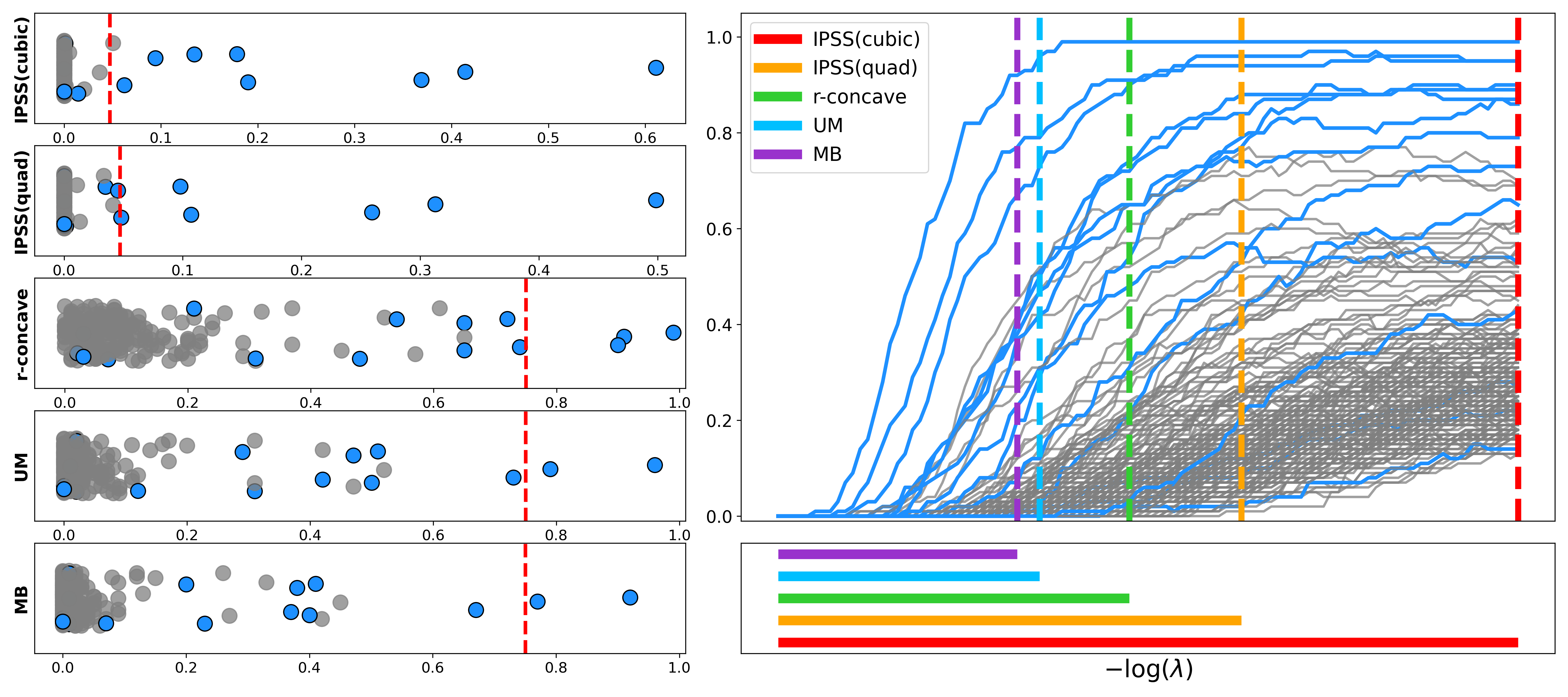}%
\caption{Linear regression with independent design as in \cref{sec:intro} with $n=150$, $p=200$, $\mathrm{SNR}=1$, and $s=15$ true features. (\textit{Left}) The horizontal axis is the score for each feature $j$, that is, $\int_\Lambda f(\hat{\pi}_j(\lambda))\mu(d\lambda)$ for IPSS and $\max_{\lambda\in\Lambda} \hat{\pi}_j(\lambda)$ for the others. The vertical axis is random jitter for visualization.  True features are shown in blue, and red vertical lines show the threshold $\tau$ separating selected and unselected features for each method. We use $\mathrm{E(FP)}_*=1$ for all methods, and for stability selection, we use $\tau = 0.75$. MB, UM, $r$-concave, IPSS(quad), and IPSS(cubic) identify 2, 2, 3, 6, and 8 true positives, respectively. All methods have 0 false positives except IPSS(cubic), which has 1, agreeing with the target E(FP). (\textit{Right}) Estimated stability paths $\hat{\pi}_j(\lambda)$. Vertical dashed lines show $-\log(\lmin)$ and horizontal lines below the plot show the intervals $[-\log(\lmax), -\log(\lmin)]$ for each method.
}
\label{fig:stability_paths}
\end{figure*}

%%%%%%%%%%%%%%%%%%%%%%%%%%%%%%%%%%%%%%%%%%%%%

\subsection{Related work}\label{sec:related_work}

%%%%%%%%%%%%%%%%%%%%%%%%%%%%%%%%%%%%%%%%%%%%%

Stability selection was introduced by \cite{mb} and refined by \cite{shah}. These remain the preeminent works on stability selection and are the most commonly implemented versions to date. \cite{topk} provide the only other work we are aware of that aims to improve the selection criterion in \cref{eq:ss} and the $\mathrm{E(FP)}$ bound. They propose \textit{top-$k$ stability selection}, which  averages over the $k$ largest selection probabilities for each feature. The special case of $k=1$ is stability selection. While \cite{topk} provide theory for top-$k$ stability selection, their improvement upon the $\mathrm{E(FP)}$ upper bound of \cite{mb} is considerably weaker than the improved bound provided by IPSS; compare \citet[Theorem 3.1]{topk} to our \cref{thrm:main}. Moreover, introducing $k$ increases the number of parameters, and results are sensitive to the choice of $k$ \citep{topk}.

It is far more common for stability selection to be modified on an \textit{ad hoc} basis to mitigate its sensitivity to parameters and overly conservative results. An example is the TIGRESS method of \cite{tigress}, which uses stability selection to infer gene regulatory networks. To reduce sensitivity to the stability selection parameters, they use a selection criterion that averages over selection probabilities. It turns out that this is the special case of IPSS with $f(x)=x$ (the function $w_1$ in \cref{sup_sec:other_functions}), whose analysis is relegated to the supplement because its bound on E(FP) is not nearly as strong as those in \cref{thrm:main,thrm:simplified}. Another example is in the work of \cite{diffeq1}, where stability selection is used to learn differential equations. There the authors use a selection criterion based only on selection probabilities at the smallest regularization parameter. Finally, a common approach is to combine stability selection with other methods. Examples include stability selection with boosting \citep{hofner} and grouping features prior to applying stability selection, which has been done in genome-wide association studies \citep{gwas}. IPSS can be used instead of stability selection in such methods at no additional cost. 

%%%%%%%%%%%%%%%%%%%%%%%%%%%%%%%%%%%%%%%%%%%%%

\section{Integrated path stability selection}\label{sec:ipss}

%%%%%%%%%%%%%%%%%%%%%%%%%%%%%%%%%%%%%%%%%%%%%

In this section, we introduce IPSS (\cref{sec:ipss_criterion}) and describe its parameters (\cref{sec:parameters}), computational details (\cref{sec:computation}), and how it can be used to control the false discovery rate (\cref{sec:fdr}).

%%%%%%%%%%%%%%%%%%%%%%%%%%%%%%%%%%%%%%%%%%%%%

\subsection{The IPSS criterion}\label{sec:ipss_criterion}

%%%%%%%%%%%%%%%%%%%%%%%%%%%%%%%%%%%%%%%%%%%%%

IPSS uses the same selection probabilities, $\hat{\pi}_j$, as stability selection. Once these are computed using \cref{alg:ss}, the set of features selected by IPSS is
\begin{align}\label{eq:ipss}
	\hat{S}_{\mathrm{IPSS},f} &= \Big\{j : \mathsmaller{\int}_\Lambda f(\hat{\pi}_j(\lambda))\mu(d\lambda) \geq \tau\Big\}
\end{align}
where the probability measure $\mu$, interval $\Lambda = [\lmin,\lmax]\subseteq(0,\infty)$, function $f:[0,1]\to\mathbb{R}$, and threshold $\tau$ are detailed in \cref{sec:parameters} below. Unlike the relatively coarse maximum criterion used by stability selection, the integral in \cref{eq:ipss} incorporates information about the stability paths over a wide range of $\lambda$ values. In \cref{fig:stability_paths}, for example, IPSS captures the fact that the stability paths of the true features rise at different rates, with some overtaking many of the false features more gradually than others, a point missed by MB, UM, and $r$-concave.

%%%%%%%%%%%%%%%%%%%%%%%%%%%%%%%%%%%%%%%%%%%%%

\subsection{Parameters}\label{sec:parameters}

%%%%%%%%%%%%%%%%%%%%%%%%%%%%%%%%%%%%%%%%%%%%%

The role of $f$ in \cref{eq:ipss} is to transform the selection probabilities $\hat{\pi}_j$ to improve performance. Importantly, not all functions $f$ can be used: To implement IPSS with a specific $f$, one must prove a corresponding bound on E(FP). Different functions yield different bounds, and tighter bounds often yield more true positives at the same target E(FP). Thus, functions with tighter bounds are generally more desirable. In \cref{thrm:main}, we establish upper bounds on E(FP) for the class of functions
\begin{align}\label{eq:half}
	h_m(x) &= (2x - 1)^m \1(x\geq 0.5),
\end{align}
where $m\in\mathbb{N}$ (the $h$ stands for ``half", since $h_m$ is nonzero on half the unit interval). We focus on this class---specifically IPSS with $f=h_2$ and $f=h_3$, denoted IPSS(quad) and IPSS(cubic), respectively---because they yield the tightest bounds on E(FP) among the many functions that we consider (\cref{sup_sec:other_functions}). Simulation and real-data results in \cref{sec:simulations}, \cref{sec:applications}, and the Supplement show that both methods routinely identify more true positives than all of the stability selection methods, consistent with the relative strengths of the theoretical bounds discussed in \cref{sec:compare}. 
Furthermore, IPSS(cubic) typically identifies at least as many true positives as IPSS(quad). 
The lone disadvantage of IPSS(cubic) relative to IPSS(quad) and stability selection is that its tighter E(FP) bound makes it more sensitive to violations of \cref{cond:1} (the only assumption besides iid samples that is required by \cref{thrm:main}), and thus more prone to overshooting the target E(FP). However, this primarily occurs in a subset of the logistic regression experiments (\cref{fig:compare_linear_class_200,fig:compare_linear_class_1000}), and even then the target E(FP) is only exceeded by $\leq 2$ false positives when $p=200$, or $\leq 4$ false positives when $p=1000$. For comparison, cross-validation produces $\geq 17$ false positives when $p=200$ and $\geq 50$ false positives when $p=1000$ on the same data.

For IPSS, the interval $\Lambda=[\lmin,\lmax]$ is defined as follows. The upper endpoint $\lmax$ is the same as in stability selection and equally inconsequential. The lower endpoint $\lmin$ is based on a bound of the form
$\evipss \leq \mathcal{I}(\lambda,\lmax) / \tau$ that depends on $f$,
where $\mathcal{I}(\lambda,\lmax)$ is an integral over $[\lambda,\lmax]$ such as in \cref{eq:quadratic,eq:cubic}. Specifically, we define
\begin{align}\label{eq:lipss}
    \lmin &= \inf\big\{\lambda\in(0,\lmax) : \mathcal{I}(\lambda,\lmax) \leq C\big\}
\end{align}
for a fixed cutoff, $C$. We always use $C=0.05$, but extensive sensitivity analyses in \cref{sup_sec:sensitivity_to_C} show that similar results are obtained for any choice of $C$ over a wide range of settings, and that both IPSS methods outperform all versions of stability selection regardless of this choice. Further details about the construction of $\Lambda$ and the evaluation of \cref{eq:lipss} are in \cref{sup_sec:lipss}.

The probability measure $\mu$ in \cref{eq:ipss} weights different values of $\lambda$. While \cref{thrm:main,thrm:simplified} hold for any choice $\mu$, we focus on the family of probability measures $\mu_\alpha(d\lambda)\propto \lambda^{-\alpha}d\lambda$ parametrized by $\alpha\in\mathbb{R}$, as detailed in \cref{sup_sec:riemann}. The values $\alpha=0$ and $\alpha=1$ correspond to averaging over $\Lambda$ on linear and log scales, respectively. Experiments in \cref{sup_sec:sensitivity} show that larger values of $\alpha$ generally yield more true positives, but are more likely to violate the E(FP) bound. This is because larger values of $\alpha$ place greater weight on smaller regularization values, where \cref{cond:1} is more likely to be violated; see \cref{sup_sec:condition}. Nevertheless, both IPSS(quad) and IPSS(cubic) identify more true positives than all of the stability selection methods while still controlling E(FP) across a wide range of $\alpha$ values (\cref{sup_sec:sensitivity_to_mu}). As a general default, $\alpha=1$ works well. Other choices of $\alpha$ based on known quantities---such as the base estimator or the number of features---can lead to even better performance; see \cref{sup_sec:sensitivity_to_mu}.

The threshold $\tau$ is determined by specifying a target E(FP), denoted $\mathrm{E(FP)}_*$, and replacing the inequality in \cref{eq:ipss_bound} with an equality, representing a worst-case scenario under the assumptions of \cref{thrm:main}. This gives $\tau = \mathcal{I}(\Lambda)/\mathrm{E(FP)}_*$ and \cref{eq:ipss} becomes
\begingroup
\small
\begin{align}\label{eq:ipss_efp}
	\hat{S}_{\mathrm{IPSS}} &= \left\{j : \mathsmaller{\int}_\Lambda f(\hat{\pi}_j(\lambda))\mu(d\lambda) \geq \frac{\mathcal{I}(\Lambda)}{\mathrm{E(FP)}_*}\right\}
		= \Big\{j : \mathrm{E(FP)}_* \geq \mathtt{efp}(j)\Big\}
\end{align}
\endgroup
where $\mathcal{I}(\Lambda) = \mathcal{I}(\lmin,\lmax)$ and, for each $j\in\{1,\dots,p\}$, the \textit{efp score} of $j$ is defined as
\begin{align}\label{eq:efp}
	\efp(j) &= \min\left\{\frac{\mathcal{I}(\Lambda)}{\mathsmaller{\int}_\Lambda f(\hat{\pi}_j(\lambda))\mu(d\lambda)},\ p\right\}.
\end{align}
Under the assumptions of \cref{thrm:main}, $\efp(j)$ is the smallest bound on E(FP) if $j$ is to be selected. The minimum in \cref{eq:efp} accounts for $\int_\Lambda f(\hat{\pi}_j(\lambda))\mu(d\lambda)$ being 0, in which case $j$ is never selected since E(FP) is at most the number of features, $p$.

\cref{eq:ipss_efp,eq:efp} help to understand why IPSS does not depend strongly on $\mu$ and $C$ (which determines $\Lambda$).  The key quantity is $\mathcal{I}(\Lambda)/\int_\Lambda f(\hat\pi_j(\lambda))\mu(d\lambda)$, in which the numerator and denominator are expectations with respect to $\mu$ over $\Lambda$. IPSS depends primarily on the integrands of these functions, which are determined by the function $f$, rather than the probability measure $\mu$ and its support $\Lambda$.

%%%%%%%%%%%%%%%%%%%%%%%%%%%%%%%%%%%%%%%%%%%%%

\subsection{Computation}\label{sec:computation}

%%%%%%%%%%%%%%%%%%%%%%%%%%%%%%%%%%%%%%%%%%%%%

\cref{alg:ipss} provides a step-by-step procedure for IPSS. Numerous works have noted that the number of subsampling iterations $B$ is inconsequential provided it is sufficiently large \citep{shah}; $B=50$ is a common choice, and the one used throughout this work. The grid in Step~\ref{item:grid} is used to accurately and efficiently approximate all integrals in the algorithm by simple Riemann sums (\cref{prop:riemann}). We use $r=25$ grid points. Like many of the other parameters, $r$ is inconsequential provided it is sufficiently large; in our experience, values greater than $15$ suffice. This is because the functions $h_m$, the stability paths, the integrand in the upper bound $\mathcal{I}(\Lambda)$, and the measures $\mu_\alpha$ are all very numerically stable. The bounds in Step~\ref{item:compute_I} are in \cref{thrm:simplified} for IPSS(quad) and IPSS(cubic). We find no discernible difference in computation time between IPSS and MB. This is because IPSS and MB both compute the $\hat{\pi}_j$ using \cref{alg:ss}, which is much more expensive than evaluation of either \cref{eq:ss} or \cref{eq:ipss}. For more on the computational requirements of \cref{alg:ss}, see \citet[Section 2.6]{mb}. 

\begin{algorithm}
\caption{(Integrated path stability selection)}\label{alg:ipss}
\begin{algorithmic}[1]
\Require{Data $\bm{z}_1,\ldots,\bm{z}_n$, selection algorithm $\hat{S}$, number of iterations $B$, function $f$, probability measure $\mu$, target $\evtarget$, integral cutoff value $C$, and number of grid points $r$.}

\item Compute $\lmin$ and $\lmax$ as described above and in \cref{sup_sec:lipss}.

\item\label{item:grid} Partition $\Lambda=[\lmin,\lmax]$ into $r$ grid points, typically on a log scale.

\item\label{item:compute_I} Compute $\mathcal{I}(\lmin,\lmax)$ using the relevant upper bound on E(FP) and \cref{prop:riemann}. 

\item Estimate selection probabilities via \cref{alg:ss} with $\bm{z}_1,\dots,\bm{z}_n$, $\hat{S}$, $\Lambda$, and $B$.

\item Approximate $\int_{\Lambda} f(\hat\pi_j(\lambda))\mu(d\lambda)$ using \cref{prop:riemann}, then compute $\efp(j)$ using \cref{eq:efp}.

\Ensure{Selected features $\hat{S}_{\mathrm{IPSS},f}=\{j : \efp(j)\leq\mathrm{E(FP)}_*\}$.}
\end{algorithmic}
\end{algorithm}

%%%%%%%%%%%%%%%%%%%%%%%%%%%%%%%%%%%%%%%%%%%%%

\subsection{False discovery rate}\label{sec:fdr}

%%%%%%%%%%%%%%%%%%%%%%%%%%%%%%%%%%%%%%%%%%%%%

IPSS and other forms of stability selection focus on E(FP) because it is interpretable and theoretically tractable. Another quantity of interest, the \textit{false discovery rate} (FDR), is the expected ratio between the number of false positives and the total number of features selected, $\mathrm{FDR} = \mathrm{E(FP/(TP + FP))}$. When $p$ is large, the FDR is well-approximated by $\mathrm{E(FP)}/\mathrm{E(TP + FP)}$ \citep{storey}. Thus, making the additional approximation $\lvert\hat{S}_{\mathrm{IPSS}}\rvert\approx \mathrm{E(TP + FP)}$, we have $\mathrm{FDR}\approx\mathrm{E(FP)}/\lvert\hat{S}_{\mathrm{IPSS}}\rvert$. Relabeling features by their efp scores so that $\efp(1)\leq\efp(2)\leq\dots\leq\efp(p)$, the set of features $\{1,\dots,j\}$ has an approximate FDR that is bounded above by $\efp(j)/j$ for each $j\in\{1,\dots,p\}$. Hence, instead of specifying $\mathrm{E(FP)}_*$, one can either (i) specify a target FDR, say $\mathrm{FDR}_*$, and choose the largest $j$ such that $\mathtt{efp}(j)/j \leq \mathrm{FDR}_*$, or (ii) choose $j$ to minimize $\mathtt{efp}(j)/j$. The resulting set of selected features is then $\{1,\dots,j\}$, that is, the features with the $j$ smallest efp scores.

%%%%%%%%%%%%%%%%%%%%%%%%%%%%%%%%%%%%%%%%%%%%%

\section{Theory}\label{sec:theory}

%%%%%%%%%%%%%%%%%%%%%%%%%%%%%%%%%%%%%%%%%%%%%

In this section, we present our theoretical results (\cref{sec:main_result}) and compare them to those of \cite{mb} and \cite{shah} (\cref{sec:compare}).
Our main result, \cref{thrm:main}, establishes a bound on E(FP) for IPSS with the functions $h_m$ defined in \cref{eq:half}. \cref{thrm:simplified} gives simplified formulas for this bound that we use in practice.
All proofs are in \cref{sup_sec:proofs}. Additional results for other choices of $f$ and their proofs are in \cref{sup_sec:other_functions}.

\subsection{Preliminaries}\label{sec:preliminaries}
It is assumed that the random vectors $\bm{Z}_1,\dots,\bm{Z}_n$, the random subsets $A_1,\dots,A_{2 B}$, and any randomness in the feature selection algorithm $\hat{S}$ are defined on a common probability space $(\Omega,\mathcal{F},\mathbb{P})$. Furthermore, $\mathrm{E}$ always denotes expectation with respect to $\mathbb{P}$. 
Let $\Lambda$ be a Borel measurable subset of $(0,\infty)$ equipped with the Borel sigma-algebra, let $\mu$ be a probability measure on $\Lambda$, and 
assume $\hat{S}_\lambda(\bm{Z}_A)$ is measurable as a function on $\Lambda\times\Omega$ for all $A\subseteq\{1,\dots,n\}$.

%%%%%%%%%%%%%%%%%%%%%%%%%%%%%%%%%%%%%%%%%%%%%

\subsection{Main results}\label{sec:main_result}

%%%%%%%%%%%%%%%%%%%%%%%%%%%%%%%%%%%%%%%%%%%%%

The following condition is used in \cref{thrm:main}. Recall that $S$ is the unknown subset of true features, and $S^c = \{1,\ldots,p\}\setminus S$ is its complement. Let $q(\lambda) = \mathrm{E}\lvert\hat{S}_\lambda(\bm{Z}_{1:\nhalf})\rvert$ denote the expected number of features selected by $\hat{S}_\lambda$ on half the data.

\begin{condition}\label{cond:1}
We say \textnormal{\cref{cond:1} holds for $m$} if for all $\lambda\in\Lambda$,
\begin{align}\label{eq:simultaneous}
    \max_{j\in S^c}\,\mathbb{P}\bigg(j\in\bigcap_{b=1}^m \big(\hat{S}_\lambda(\bm{Z}_{A_{2b-1}})\cap \hat{S}_\lambda(\bm{Z}_{A_{2b}})\big)\bigg) &\leq (q(\lambda)/p)^{2m}.
\end{align}
\end{condition}
\noindent\cref{eq:simultaneous}, discussed in greater detail below, says the probability that any non-true feature $j$ is selected by both $\hat{S}_\lambda(\bm{Z}_{A_{2 b-1}})$ and $\hat{S}_\lambda(\bm{Z}_{A_{2 b}})$ in $m$ resampling iterations is no greater than the $2m$-th power of the expected proportion of features selected by $\hat{S}_\lambda$ using half the data.

\begin{theorem}\label{thrm:main}
Let $\tau\in(0,1]$ and $m\in\mathbb{N}$. Define $\hat{S}_{\mathrm{IPSS},h_m}$ as in \cref{eq:ipss,eq:half}. If \cref{cond:1} holds for all $m'\in\{1,\ldots,m\}$, then
\begin{align}\label{eq:ipss_bound}
    \mathrm{E(FP)} &= \mathrm{E}\lvert\hat{S}_{\mathrm{IPSS},h_m}\cap S^c\rvert 
        \leq \frac{p}{\tau B^m}\sum_{k_1+\cdots+k_B=m}\frac{m!}{k_1!k_2!\cdots k_B!}\int_\Lambda (q(\lambda)/p)^{2\sum_b \1(k_b\neq 0)}\mu(d\lambda)
\end{align}
where $B$ is the number of subsampling steps in \cref{alg:ss} and the sum is over all nonnegative integers $k_1,\ldots,k_B$ such that $k_1 + \cdots + k_B = m$.
\end{theorem}

\cref{eq:ipss_bound} bounds the expected number of false positives when using IPSS with $h_m$. The following theorem shows that the bound simplifies considerably for certain choices of $m$.

\begin{theorem}\label{thrm:simplified}
Let $\tau\in(0,1]$. If \cref{cond:1} holds for $m=1$, then IPSS with $h_1$ satisfies
\begin{align}\label{eq:linear}
	\mathrm{E(FP)} &\leq \frac{1}{\tau }\int_\Lambda \frac{q(\lambda)^2}{p}\mu(d\lambda);
\end{align}
if \cref{cond:1} holds for $m\in\{1,2\}$, then IPSS with $h_2$ satisfies
\begin{align}\label{eq:quadratic}
	\mathrm{E(FP)} &\leq \frac{1}{\tau}\int_\Lambda \bigg(\frac{q(\lambda)^2}{Bp} + \frac{(B-1)q(\lambda)^4}{Bp^3}\bigg)\mu(d\lambda);
\end{align}
and if \cref{cond:1} holds for $m\in\{1,2,3\}$, then IPSS with $h_3$ satisfies
\begin{align}\label{eq:cubic}
	\mathrm{E(FP)} &\leq \frac{1}{\tau}\int_\Lambda \bigg(\frac{q(\lambda)^2}{B^2p} + \frac{3(B-1)q(\lambda)^4}{B^2p^3} + \frac{(B-1)(B-2)q(\lambda)^6}{B^2p^5}\bigg)\mu(d\lambda).
\end{align}
\end{theorem}

\noindent 
Taking the limit as $B\to\infty$, \cref{eq:quadratic,eq:cubic} become
\begin{align}
\label{eq:asymptotic}
	\limsup_{B\to\infty} \mathrm{E(FP)} &\leq \frac{1}{\tau p^3}\int_\Lambda q(\lambda)^4\mu(d\lambda)
		\quad\text{and}\quad
		\limsup_{B\to\infty} \mathrm{E(FP)}
		\leq \frac{1}{\tau p^5}\int_\Lambda q(\lambda)^6\mu(d\lambda).
\end{align}
Although we do not use these asymptotic bounds, the pattern from $h_1$ (\cref{eq:linear}) to $h_2$ to $h_3$ (\cref{eq:asymptotic}) provides insight into the relationships between $\mathrm{E(FP)}$, $p$, and $h_m$.

\cref{cond:1} holds for $m=1$ whenever $\max_{j\in S^c} \pi_j(\lambda) \leq q(\lambda)/p$ for all $\lambda\in\Lambda$ since $\bm{Z}_1,\ldots,\bm{Z}_n$ are i.i.d.\ and independent of $A_1,\ldots,A_{2 B}$, and thus for any $j\in S^c$,
\begin{align}\label{eq:pi_squared}
    \mathbb{P}\big(j\in \hat{S}_\lambda(\bm{Z}_{A_{2b-1}})\cap \hat{S}_\lambda(\bm{Z}_{A_{2b}})\big) &= 
    \mathrm{E}\Big(\mathrm{E}\big(\1(j\in \hat{S}_\lambda(\bm{Z}_{A_{2 b-1}}))\1(j\in\hat{S}_\lambda(\bm{Z}_{A_{2 b}}))\;\big\vert\; A_{2 b},A_{2 b-1} \big)\Big) \notag \\
        &= \mathrm{E}\big(\pi_j(\lambda)\pi_j(\lambda)\big) 
        = \pi_j(\lambda)^2
        \leq (q(\lambda)/p)^2.
\end{align}
In turn, the $\max_{j\in S^c} \pi_j(\lambda) \leq q(\lambda)/p$ condition is implied by the exchangeability and not-worse-than-random-guessing conditions used by \cite{mb} and \cite{shah} in the stability selection analogues of \cref{thrm:main}, detailed in \cref{sec:compare}.  To be precise, \cite{shah} do not require these conditions in their theory, but they are always assumed when implementing their versions of stability selection in practice. An empirical study and further details about \cref{cond:1} for the practically relevant cases of $m\in\{1,2,3\}$ are in \cref{sup_sec:condition}.

%%%%%%%%%%%%%%%%%%%%%%%%%%%%%%%%%%%%%%%%%%%%%

\subsection{Comparison to stability selection}\label{sec:compare}

%%%%%%%%%%%%%%%%%%%%%%%%%%%%%%%%%%%%%%%%%%%%%

The analogue of \cref{eq:ipss_bound} for stability selection (\cref{eq:ss}) under the exchangeability and not-worse-than-random-guessing conditions of \cite{mb} is
\begin{align}\label{eq:ss_bound_2}
	\evmb &\leq \frac{q_\Lambda^2}{(2\tau - 1)p},
\end{align}
where $q_\Lambda = \mathrm{E}\big\lvert\textstyle{\bigcup}_{\lambda\in\Lambda}\hat{S}_{\lambda}(\bm{Z}_{1:\nhalf})\big\rvert$.
Under the additional assumptions that (a) $q(\lambda)^2/p\leq 1/\sqrt{3}$ for all $\lambda\in\Lambda$ and (b) the distributions of the simultaneous selection probabilities (defined in \cref{sup_sec:proofs}) are unimodal, \cite{shah} establish the stronger bound
\begin{align}\label{eq:um}
    \evum &\leq \frac{C(\tau,B)\, q_\Lambda^2}{p}
\end{align} 
for stability selection where, for $\tau\in\big\{\frac{1}{2}+1/B,\, \frac{1}{2} + 3/(2 B),\, \frac{1}{2} + 2/B,\dots, 1\big\}$,
\begin{align*}
	C(\tau,B) &=
		\begin{cases}
		\displaystyle\frac{1}{2(2\tau-1-\frac{1}{2 B})} & \text{if }  \tau\in\Big(\min\!\Big\{\tfrac{1}{2}+\frac{q_\Lambda^2}{p^2},\;\frac{1}{2} + \frac{1}{2 B} + \frac{3 q_\Lambda^2}{4 p^2}\Big\},\;\; 3/4\Big], \\[14pt]
		\displaystyle\frac{4(1 - \tau + \frac{1}{2 B})}{1 + \frac{1}{B}} & \text{if } \tau\in(3/4,\, 1].
		\end{cases}
\end{align*}
There are several reasons the IPSS bounds in \cref{thrm:simplified} are significantly tighter than the Meinshausen and B\"uhlmann (MB) and unimodal (UM) bounds in \cref{eq:ss_bound_2,eq:um}. First, the IPSS bounds hold for all $\tau\in(0,1]$, whereas the MB and UM bounds are restricted to $\tau\in(0.5,1]$ since they go to $\infty$ as $\tau \to 0.5$. Second, all of the terms in the integrands of \cref{eq:quadratic,eq:cubic} are typically orders of magnitude smaller than $q^2/p$ in both \cref{eq:ss_bound_2,eq:um}. Indeed, since $\max\{q(\lambda) : \lambda\in\Lambda\}\leq q_\Lambda$ and $q(\lambda)$ is much smaller than $p$ over a wide range of $\lambda$ values in sparse or even moderately sparse settings, $q(\lambda)^4/p^3$ and $q(\lambda)^6/p^5$ are typically much smaller than $q_\Lambda^2/p$, and this difference becomes more pronounced as $p$ grows. 
Additionally, the lower order terms in \cref{eq:quadratic,eq:cubic} are $O(1/B)$ or $O(1/B^2)$, so with our typical choice of $B = 50$, the contribution of these terms is reduced even further, tending to 0 as $B\to\infty$ (\cref{eq:asymptotic}).
By contrast, the MB bound has no $B$ dependence, and the UM bound depends only weakly on $B$.

\cite{shah} also derive an upper bound based on assumptions of $r$-concavity. They argue that if the estimated selection probabilities are $-1/4$-concave and the simultaneous selection probabilities (\cref{sup_sec:proofs}) are $-1/2$-concave for every feature in $S^c$, then
\begin{align}\label{eq:r-concave}
	\mathrm{E(FP)} &\leq \min\Big\{D(q_\Lambda^2/p^2, 2\tau - 1, B, -1/2),\; D(q_\Lambda/p, \tau, 100, -1/4)\Big\}\,p
\end{align}
where $D(\eta, \tau, B, r)$ is the maximum of $\mathbb{P}(X\geq \tau)$ over all $r$-concave random variables $X$ that are supported on $\{0, 1/B, 2/B, \dots, 1\}$ and satisfy $\mathbb{E}(X)\leq \eta$ \citep{shah}. While tighter than the MB and UM bounds, the function $D$---and hence, the upper bound in \cref{eq:r-concave}---does not have a closed form and must be approximated with an additional algorithm. This lack of a closed-form expression and the constrained maximization over all $r$-concave random variables makes it difficult to compare the $r$-concave bound with \cref{eq:ipss_bound,eq:ss_bound_2,eq:um} analytically. However, the results in \cref{fig:bounds} indicate that our bounds for IPSS(quad) and IPSS(cubic) are tighter than all of the stability selection bounds---including $r$-concave---especially for target E(FP) values less than 5, which is the most practically relevant range. Furthermore, our empirical results show that IPSS(quad) and IPSS(cubic) consistently identify more true positives than $r$-concave while maintaining E(FP) control across many diverse settings. Thus, while $r$-concave is less conservative than other traditional approaches, it requires stronger assumptions, relies on an additional algorithm due to lack of a closed form, and routinely underperforms IPSS.

\begin{figure*}
\makebox[\textwidth]{\includegraphics[width=\textwidth, height=.32\textheight]{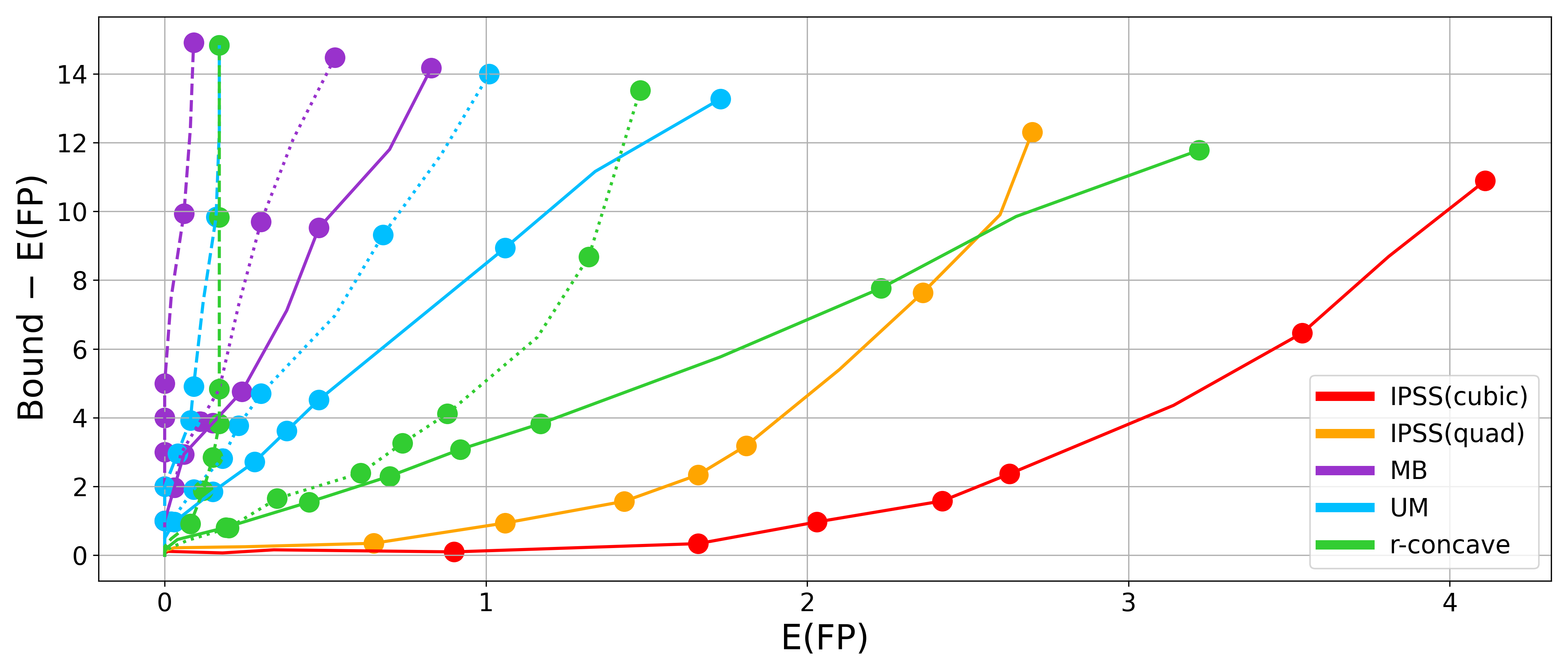}}%
\caption{\textit{Tightness of the E(FP) bounds}. The difference between the theoretical E(FP) bound for each method and its actual E(FP), averaged over $100$ data sets simulated from an independent linear regression model with normal residuals and $p=200$ features, as described in \cref{sec:simulations}, versus actual E(FP). Since all methods replace the inequalities in their E(FP) bounds with equalities, a perfectly calibrated method would generate a horizontal line at 0, that is, $\mathrm{E(FP)}=\mathrm{Bound}$. Dots show the results for each method when $\mathrm{E(FP)}_*$ equals, from left to right, 1, 2, 3, 4, 5, 10, and 15. IPSS(quad) and especially IPSS(cubic) are much closer to their theoretical bounds, particularly when $\mathrm{E(FP)}_*\leq 5$. Solid, dotted, and dashed lines for the stability selection methods correspond to $\tau=0.6$, 0.75, and 0.9, respectively.
}
\label{fig:bounds}
\end{figure*}

%%%%%%%%%%%%%%%%%%%%%%%%%%%%%%%%%%%%%%%%%%%%%

\section{Simulations}\label{sec:simulations}

%%%%%%%%%%%%%%%%%%%%%%%%%%%%%%%%%%%%%%%%%%%%%

We present results from linear and logistic regression simulations for a variety of feature distributions and base estimators. The performance of IPSS is compared to the stability selection methods of \cite{mb} and \cite{shah}, as well as cross-validation.

\textit{Setup.}
Data are simulated from a linear regression model with normal residuals:
\begin{align*}
    Y_i &= \bm{X}_i^\mathtt{T}\bm{\beta}^* + \epsilon_i,
        \quad
        \epsilon_i\sim \mathcal{N}(0,\sigma^2),
\end{align*}
from a linear model with residuals from a Student's $t$ distribution with $2$ degrees of freedom:
\begin{align*}
    Y_i &= \bm{X}_i^\mathtt{T}\bm{\beta}^* + \epsilon_i,
        \quad
        \epsilon_i\sim t(2),
\end{align*}
and from a binary logistic regression model:
\begin{align*}
    Y_i &\sim \mathrm{Bernoulli}(p_i),
        \quad
        p_i = \frac{\exp(\gamma \bm{X}_i^\mathtt{T}\bm{\beta}^*)}{1 + \exp(\gamma \bm{X}_i^\mathtt{T}\bm{\beta}^*)},
\end{align*}
for $i \in\{1,\ldots,n\}$.
For each simulated data set, the coefficient vector $\bm{\beta}^*$ has $s$ nonzero entries $\beta^*_j\sim\mathrm{Uniform}([-1,-0.5]\cup[0.5,1])$ located at randomly chosen coordinates $j\in\{1,\dots,p\}$, and the remaining $p-s$ entries are set to $0$. We simulate features from the following designs:
\begin{itemize}
\item \textit{Independent}: $\bm{X}_i\sim \mathcal{N}(0,I_p)$ for all $i$, where $I_p$ is the $p\times p$ identity matrix.

\item \textit{Toeplitz:} $\bm{X}_i\sim\mathcal{N}(0,\Sigma)$ where $\Sigma_{j k}=\rho^{\lvert j-k\rvert}$. We consider two cases: $\rho=0.5$ and $\rho=0.9$.

\item \textit{RNA-seq:} $n$ samples and $p$ features are drawn uniformly at random from RNA-sequencing measurements of $6426$ genes from $569$ ovarian cancer patients \citep{linkedomics}.
\end{itemize}
For each $j\in\{1,\dots,p\}$, we standardize the observed features $(x_{1 j},\ldots,x_{n j})$ to have sample mean $0$ and sample variance $1$ before applying $\hat{S}_\lambda$. For linear regression, the observed response $y=(y_1,\ldots,y_n)$ is centered to have sample mean $0$. For linear regression with normal residuals, $\sigma^2$ is chosen to satisfy a specified empirical signal-to-noise ratio (SNR), defined by $\mathrm{SNR} = \sum_{i=1}^n (\bm{x}_i^\mathtt{T}\bm{\beta}^*)^2/(n\sigma^2)$. For logistic regression, $\gamma>0$ determines the strength of the signal. 

The three models (linear regression with normal and Student's $t$ residuals, and logistic regression), four feature designs (independent, Toeplitz with $\rho=0.5$ and $0.9$, and RNA-seq), and two feature dimensions ($p=200$ and $1000$) yield a total of $24$ simulation settings. For linear regression, we perform experiments with lasso, MCP, SCAD, and the adaptive lasso as the base estimators, and for logistic regression the base estimator is $\ell_1$-regularized logistic regression. For IPSS, the parameter $C$ is always set to $0.05$ and $\alpha$ is set to the default values described in \cref{sup_sec:sensitivity_to_mu}. Each experiment consists of $100$ trials, where one trial consists of generating data as above with $n$, $s$, and signal strength chosen according to \cref{tab:simulation_parameters}, estimating the selection probabilities via \cref{alg:ss} with $B=50$ subsamples, and choosing features according to each criterion.

\begin{table}[h!]
\centering
\setlength{\tabcolsep}{4pt} % reduce column padding
\begin{tabular}{|c|c|c|c|c|}
\hline
$p$ & $n$ & $s$ & SNR (regression) & $\gamma$ (classification) \\
\hline
200 & $\mathrm{Uniform}\{50,\dots,200\}$ & $\mathrm{Uniform}\{5,\dots,20\}$ & $\mathrm{Uniform}(1/3,3)$ & $\mathrm{Uniform}(1/2,2)$ \\
\hline
1000 & $\mathrm{Uniform}\{100,\dots,500\}$ & $\mathrm{Uniform}\{10,\dots,40\}$ & $\mathrm{Uniform}(1/3,3)$ & $\mathrm{Uniform}(1/2,2)$ \\
\hline
\end{tabular}
\caption{\textit{Simulation parameters}. The number of samples $n$, number of true features $s$, and the signal strength parameters, SNR and $\gamma$, are randomly selected prior to each trial as above, ensuring that our experiments cover a wide range of settings.}
\label{tab:simulation_parameters}
\end{table}

\textit{Results}. We quantify performance in terms of true positives (TP) and false positives (FP), where a selected feature is a \textit{true positive} if its corresponding $\bm{\beta}^*$ entry is nonzero, and is a \textit{false positive} otherwise. The dashed black line in each FP plot shows the target value of $\mathrm{E(FP)}$. A tight bound on $\mathrm{E(FP)}$ should lead to curves lying close to this line, since all methods replace the inequalities in their respective $\mathrm{E(FP)}$ bounds with equalities to calibrate the parameters. For example, if the target $\mathrm{E(FP)}$ is $2$, a perfectly calibrated algorithm would produce an average FP of $2$.

\begin{figure*}
\includegraphics[width=\textwidth, height=.35\textheight]{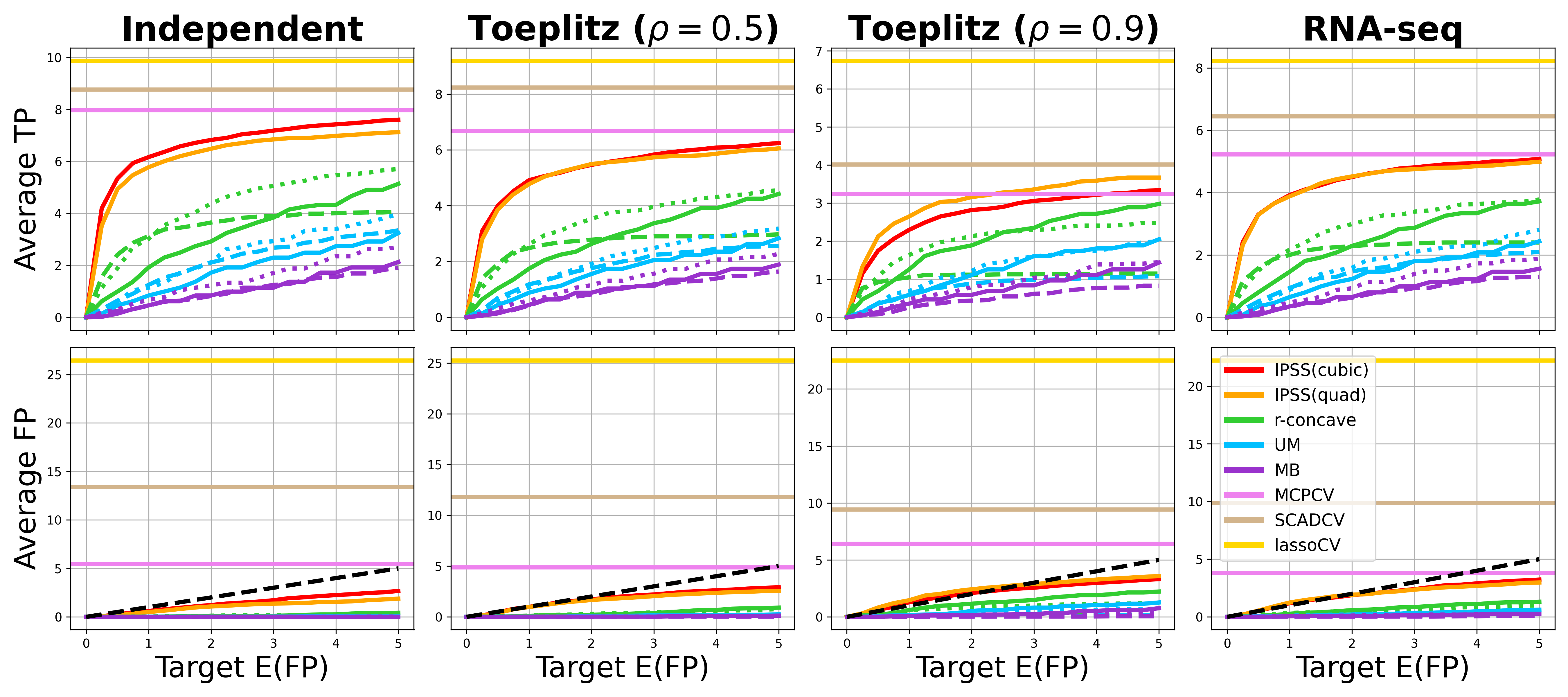}%
\caption{\textit{Linear regression, normal residuals $(p=200)$}. The solid, dotted, and dashed lines for the stability selection methods represent $\tau=0.6$, $0.75$, and $0.9$, respectively. Cross-validation results (horizontal lines) are independent of Target E(FP). The dashed black line represents perfect E(FP) control.}
\label{fig:compare_linear_reg_200}
\end{figure*}

\begin{figure*}
\includegraphics[width=\textwidth, height=.35\textheight]{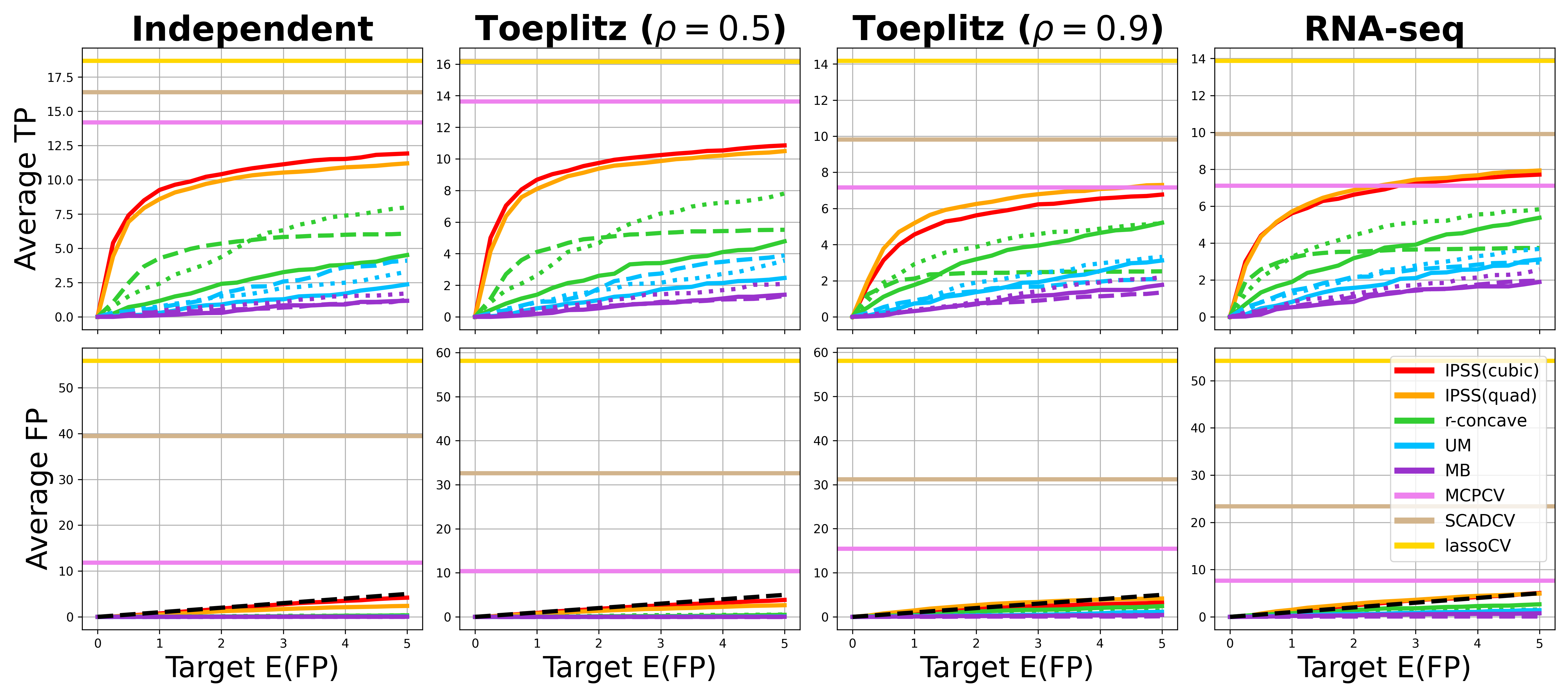}%
\caption{\textit{Linear regression, normal residuals $(p=1000)$}. See  \cref{fig:compare_linear_reg_200} for details.}
\label{fig:compare_linear_reg_1000}
\end{figure*}

\cref{fig:compare_linear_reg_200,fig:compare_linear_reg_1000} show results for linear regression with normal residuals and lasso as the base estimator when $p=200$ and $1000$. Results are shown for IPSS and the stability selection methods as well as lasso, MCP, and SCAD with their regularization parameters by cross-validation (lassoCV, MCPCV, and SCADCV). The remaining results are in \cref{sup_sec:simulations}. \cref{fig:compare_linear_reg_200_df2,fig:compare_linear_reg_1000_df2} show results for linear regression with Student's $t$ residuals and lasso as the base estimator, \cref{fig:compare_linear_reg_200_mcp,fig:compare_linear_reg_1000_mcp} show results for linear regression with normal residuals and MCP as the base estimator, \cref{fig:compare_linear_reg_200_scad,fig:compare_linear_reg_1000_scad} show results for linear regression with normal residuals and SCAD as the base estimator, \cref{fig:compare_linear_reg_200_adaptive,fig:compare_linear_reg_1000_adaptive} show results for linear regression with normal residuals and the adaptive lasso as the base estimator, and \cref{fig:compare_linear_class_200,fig:compare_linear_class_1000} show results for logistic regression with $\ell_1$-regularized logistic regression as the base estimator.

In all 48 experiments, IPSS(quad) achieves more accurate E(FP) control and identifies more true positives than all of the stability selection methods, often substantially so. IPSS(cubic) also achieves more accurate E(FP) control and identifies more true positives than the stability selection methods in almost every experiment, though---as discussed in \cref{sec:parameters}---its empirical E(FP) slightly exceeds the target E(FP) in a few cases, notably logistic regression with Toeplitz ($\rho=0.9$) and RNA-seq designs. This is because violations of \cref{cond:1} are more likely for $\ell_1$-regularized logistic regression than for lasso, MCP, SCAD, or the adaptive lasso. One possible explanation for this---which agrees with our empirical observations---is that $\ell_1$-regularized logistic regression can encounter computational difficulties at small regularization values when $p\gg n$, leading to suboptimal solutions and hence poorly estimated selection probabilities \citep{logistic_regression}. Among the cross-validation methods, lassoCV and SCADCV have a high TP, but also exceedingly high FP. MCPCV does better in terms of limiting the number of false positives, but not as well as IPSS and stability selection. In some experiments, IPSS even identifies more true positives than MCPCV while selecting fewer false positives.

\cref{fig:compare_linear_reg_200,fig:compare_linear_reg_1000} and the results in \cref{sup_sec:simulations} indicate that IPSS provides a better balance between true and false positives across a wide range of settings and base estimators than stability selection and cross-validation. That is, while stability selection has low FP at the expense of low TP, and cross-validation has high TP at the expense of high FP, IPSS stays at or below the target E(FP) while achieving a TP that can approach or even surpass the TP of cross-validation.

%%%%%%%%%%%%%%%%%%%%%%%%%%%%%%%%%%%%%%%%%%%%%

\section{Applications}\label{sec:applications}

%%%%%%%%%%%%%%%%%%%%%%%%%%%%%%%%%%%%%%%%%%%%%

%%%%%%%%%%%%%%%%%%%%%%%%%%%%%%%%%%%%%%%%%%%%%

\subsection{Prostate cancer}\label{sec:prostate}

%%%%%%%%%%%%%%%%%%%%%%%%%%%%%%%%%%%%%%%%%%%%%

We applied IPSS and the stability selection methods to reverse-phase protein array (RPPA) measurements of $p=125$ proteins in $n=351$ prostate cancer patients \citep{linkedomics}. The response is \textit{tumor purity}---the proportion of cancerous cells in a tissue sample---and the goal is to identify the genes that are most related to it. Since tumor purity takes values in $[0,1]$, we use lasso as our base selection algorithm. The target $\mathrm{E(FP)}$ is $1$ for all methods, and for MB, UM, and $r$-concave we set $\tau=0.75$. \cref{fig:prostate_rppa} shows the results. 
MB, UM, $r$-concave, IPSS(quad), and IPSS(cubic) select $4$, $4$, $5$, $8$, and $10$ proteins, respectively.
Although it is difficult to know which features should be selected on real data such as this, a literature search presented in \cref{sup_sec:applications} indicates that all $10$ proteins identified by IPSS play a nontrivial role in prostate cancer. 
As in \cref{fig:stability_paths}, the stability selection methods miss important information about the stability paths that IPSS successfully captures. For example, the selection probabilities for PKC, BAK1, and PTEN are essentially $0$ for many $\lambda$ values before abruptly rising above many of the other paths.

\begin{figure*}[!ht]
\makebox[\textwidth]{\includegraphics[width=\textwidth, height=.30\textheight]{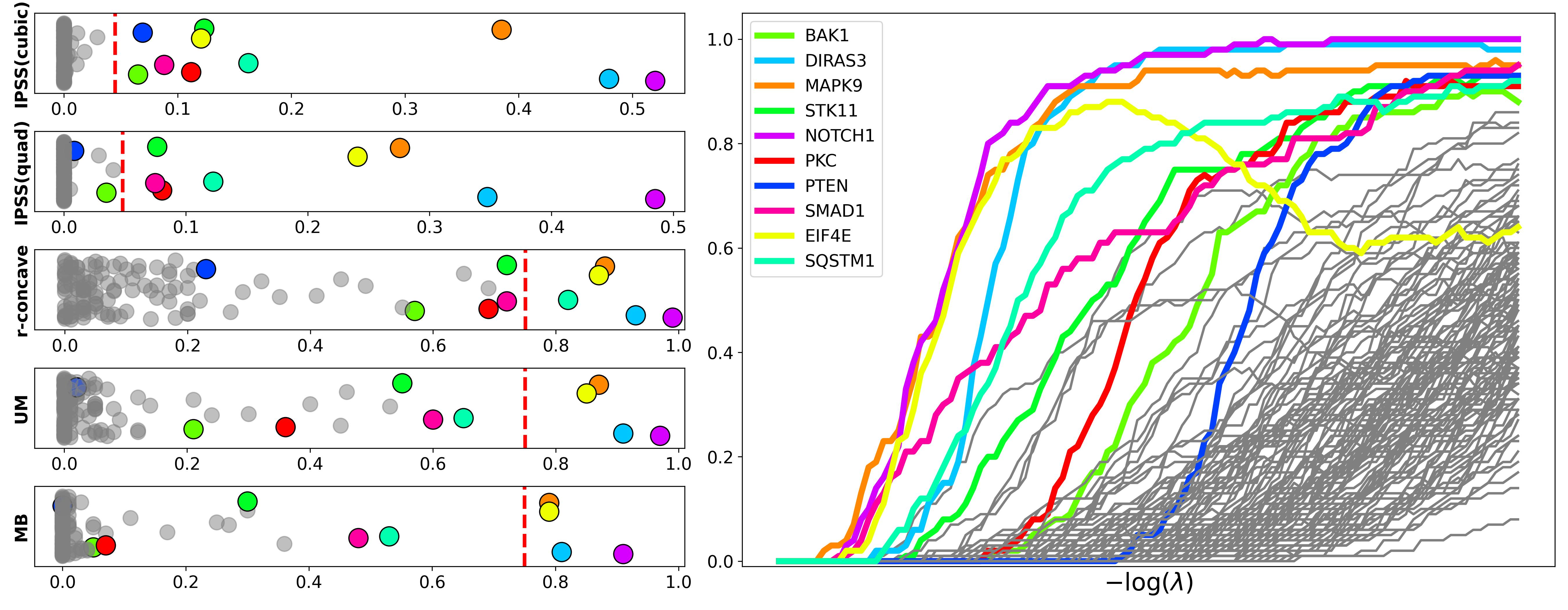}}%
\caption{\textit{Prostate cancer results}. (\textit{Left}) Feature scores and thresholds (vertical red lines) separating selected and unselected genes for each method. Scores are one-dimensional and correspond to the horizontal axes. Every method's set of selected proteins is a subset of the proteins selected by IPSS(cubic), shown in color in all plots; the remaining proteins are in gray.  (\textit{Right}) Estimated stability paths for each protein. The horizontal axis is on a log scale. 
}
\label{fig:prostate_rppa}
\end{figure*}

%%%%%%%%%%%%%%%%%%%%%%%%%%%%%%%%%%%%%%%%%%%%%

\subsection{Colon cancer}\label{sec:colon}

%%%%%%%%%%%%%%%%%%%%%%%%%%%%%%%%%%%%%%%%%%%%%

We applied IPSS and stability selection to the expression levels of $p=1908$ genes in $n=62$ tissue samples, $40$ cancerous and $22$ normal \citep{colon}. The goal is to identify genes whose expression levels differ between the cancerous and normal samples. Since the response is binary, we use $\ell_1$-regularized logistic regression as the base estimator. The target E(FP) is $1/2$ for all methods, and the threshold for MB, UM, and $r$-concave is $\tau=0.75$. Expression levels are log-transformed and standardized as in \citet{shah}, who also study these data. 
\cref{fig:colon,fig:colon_heatmap} show the results. MB, UM, $r$-concave, IPSS(quad), and IPSS(cubic) select $1$, $2$, $7$, $11$, and $16$ genes, respectively. In addition to \cref{fig:colon_heatmap}, a literature search reported in \cref{sup_sec:applications} supports the claim that the genes identified by IPSS are related to colon cancer.  

\begin{figure*}[!ht]
\makebox[\textwidth]{\includegraphics[width=\textwidth, height=.27\textheight]{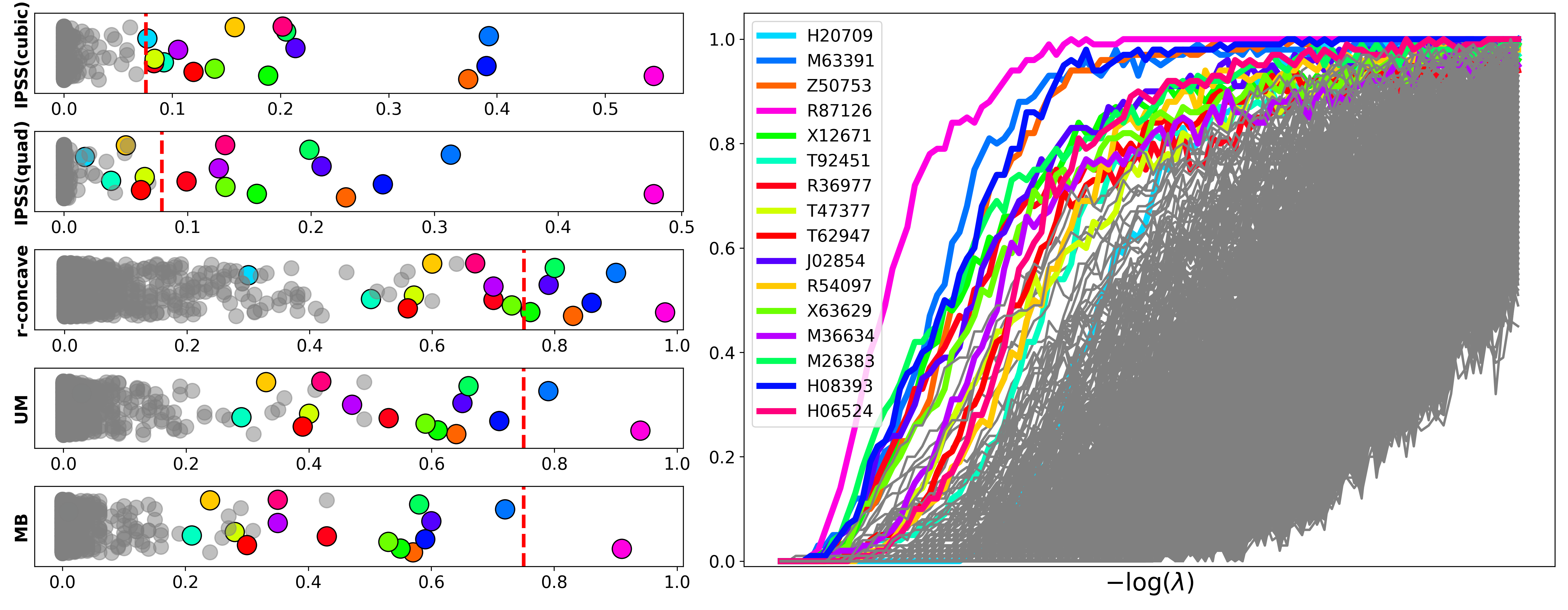}}%
\caption{\textit{Colon cancer results}. (\textit{Left}) Feature scores and thresholds (vertical red lines) separating selected and unselected genes. Every method's set of selected genes is a subset of those selected by IPSS(cubic), shown in color in all plots; the remaining genes are in gray. (\textit{Right}) Estimated stability paths for each gene.}
\label{fig:colon}
\end{figure*}

\begin{figure*}[!ht]
\makebox[\textwidth]{\includegraphics[width=0.7\textwidth]{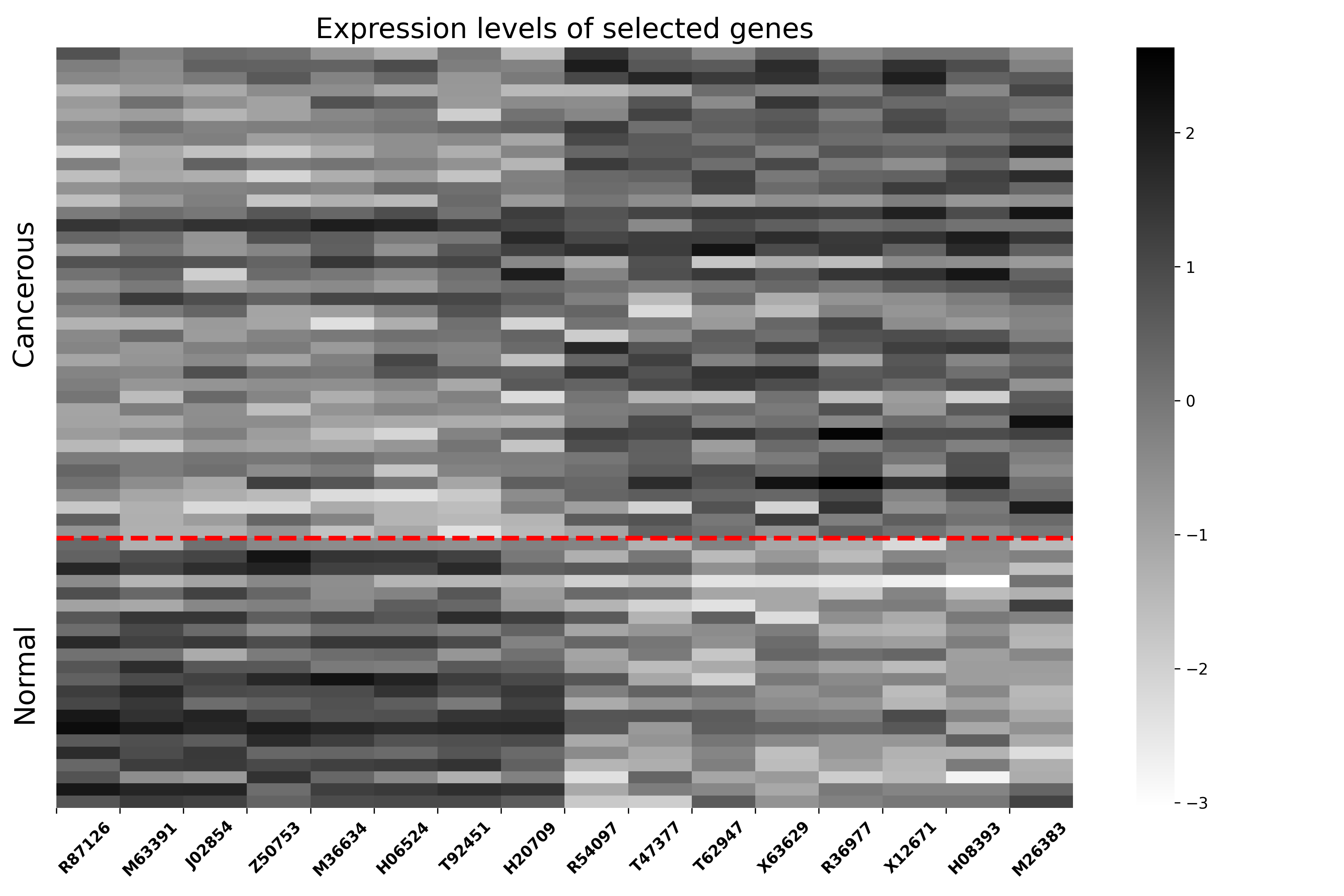}}%
\caption{\textit{Expression level heatmap for the 16 genes selected by IPSS(cubic)}. Each of the 62 rows corresponds to one tissue sample. The first 40 rows are cancerous, and the latter 22 are normal; the dashed red line separates the two classes. Each column corresponds to a gene selected by IPSS(cubic), or equivalently, the union of genes selected by each method since each gene selected by the other methods was also selected by IPSS(cubic). For each gene, there is a clear distinction between expression levels for cancerous versus normal samples.}
\label{fig:colon_heatmap}
\end{figure*}

%%%%%%%%%%%%%%%%%%%%%%%%%%%%%%%%%%%%%%%%%%%%%

\section{Discussion}\label{sec:discussion}

%%%%%%%%%%%%%%%%%%%%%%%%%%%%%%%%%%%%%%%%%%%%%

IPSS has several attractive properties. It has stronger theoretical guarantees than stability selection and significantly better performance on a wide range of simulated and real data. It also uses the same selection probabilities as stability selection and therefore has the same computational cost. In this work, we focused on the functions $f = h_2$ and $f = h_3$ due to their favorable theoretical properties and empirical performance. However, it is possible that other functions will lead to better results; investigating this point is an interesting line of future work. We also focused on the family of probability measures $\mu_\alpha\propto \lambda^{-\alpha}d\lambda$ and provided guidance on choosing $\alpha$. Like $f$, it is possible that other means of selecting $\alpha$ or other choices of $\mu$ might lead to further improvements in performance. Notably, \cref{thrm:main,thrm:simplified} and the theorems in \cref{sup_sec:other_functions} hold for arbitrary probability measures $\mu$, providing considerable flexibility in this choice.

Another interesting direction is to apply IPSS with other base estimators, such as graphical lasso and elastic net \citep{glasso, elastic_net}. In the case of elastic net, there are two regularization parameters, and while our methodology appears to carry over to this setting (now with $\Lambda\subseteq (0,\infty)^2$ and $(\lambda_1,\lambda_2)\mapsto\hat{S}_{\lambda_1,\lambda_2}$), it remains to work out the details and investigate the performance of IPSS in the context of multiple regularization parameters. Finally, we noted in \cref{sec:related_work} that stability selection is often used in conjunction with other statistical methods. Given that IPSS yields better results than stability selection at no additional cost and with less tuning, it would be interesting to study these joint methods with IPSS in place of stability selection.

%%%%%%%%%%%%%%%%%%%%%%%%%%%%%%%%%%%%%%%%%%%%%

\section*{Data and code availability}\label{sec:data}

%%%%%%%%%%%%%%%%%%%%%%%%%%%%%%%%%%%%%%%%%%%%%

All datasets and code from this work are available at \url{https://github.com/omelikechi/ipss_jasa}. Original datasets can be accessed at the following links: ovarian cancer (\url{https://www.linkedomics.org/data_download/TCGA-OV/}); prostate cancer (\url{https://www.linkedomics.org/data_download/TCGA-PRAD}); and colon cancer (\url{http://genomics-pubs.princeton.edu/oncology/affydata/index.html}).

%%%%%%%%%%%%%%%%%%%%%%%%%%%%%%%%%%%%%%%%%%%%%

\section*{Acknowledgments}\label{sec:acks}

%%%%%%%%%%%%%%%%%%%%%%%%%%%%%%%%%%%%%%%%%%%%%

O.M. thanks David Dunson and Steven Winter for initial discussions that took place under funding from Merck \& Co.\ and the National Institutes of Health (NIH) grant R01ES035625.
J.W.M.\ was supported in part by the National Institutes of Health (NIH) grant R01CA240299.

%%%%%%%%%%%%%%%%%%%%%%%%%%%%%%%%%%%%%%%%%%%%%

\section*{Disclosure statement}\label{sec:disclosure}

%%%%%%%%%%%%%%%%%%%%%%%%%%%%%%%%%%%%%%%%%%%%%

The authors report there are no competing interests to declare.

%%%%%%%%%%%%%%%%%%%%%%%%%%%%%%%%%%%%%%%%%%%%

\bibliographystyle{plainnat}
\bibliography{refs}

%%%%%%%%%%%%%%%%%%%%%%%%%%%%%%%%%%%%%%%%%%%%%

\clearpage

\setcounter{page}{1}
\setcounter{section}{0}
\setcounter{table}{0}
\setcounter{figure}{0}
\renewcommand{\theHsection}{SIsection.\arabic{section}}
\renewcommand{\theHtable}{SItable.\arabic{table}}
\renewcommand{\theHfigure}{SIfigure.\arabic{figure}}
\renewcommand{\thepage}{S\arabic{page}}  
\renewcommand{\thesection}{S\arabic{section}}   
\renewcommand{\thetable}{S\arabic{table}}   
\renewcommand{\thefigure}{S\arabic{figure}}
\renewcommand{\thealgorithm}{S\arabic{algorithm}}

\begin{center}
{\Large\textbf{{\LARGE S}UPPLEMENTARY MATERIAL}}
\end{center}

We provide further details on the implementation of IPSS (\cref{sup_sec:implementation}), prove \cref{thrm:main,thrm:simplified} (\cref{sup_sec:proofs}), state and prove results for IPSS with functions that are not considered in the main text (\cref{sup_sec:other_functions}), describe using IPSS with other base estimators (\cref{sup_sec:other_estimators}), and provide an empirical study of \cref{cond:1} (\cref{sup_sec:condition}), additional results from the prostate and colon cancer applications (\cref{sup_sec:applications}) and simulation experiments (\cref{sup_sec:simulations}), and a sensitivity analysis of the IPSS parameters, $C$ and $\alpha$ (\cref{sup_sec:sensitivity}).

%%%%%%%%%%%%%%%%%%%%%%%%%%%%%%%%%%%%%%%%%%%%%%%%%%%%%%%%%%%%%%%%%

\section{Algorithmic details}\label{sup_sec:implementation}

%%%%%%%%%%%%%%%%%%%%%%%%%%%%%%%%%%%%%%%%%%%%%%%%%%%%%%%%%%%%%%%%

We elaborate on the construction of the interval $\Lambda$ for IPSS (\cref{sup_sec:lipss}) and describe a simple way to approximate the integrals in the IPSS criterion (\cref{eq:ipss}) for a wide range of functions $f$ and measures $\mu$ (\cref{sup_sec:riemann}).

%%%%%%%%%%%%%%%%%%%%%%%%%%%%%%%%%%%%%%%%%%%%%%%%%%%%%%%%%%%%%%%%

\subsection{Constructing $\bf{\Lambda}$}\label{sup_sec:lipss}

%%%%%%%%%%%%%%%%%%%%%%%%%%%%%%%%%%%%%%%%%%%%%%%%%%%%%%%%%%%%%%%%

Features are always standardized to have mean $0$ and standard deviation $1$ in both linear and logistic regression, and the response is always centered to have mean $0$ in linear regression. For lasso, the upper endpoint $\lmax$ of $\Lambda=[\lmin,\lmax]$ is set to twice the maximum correlation between the features and response; that is, $\lmax = 2\max_j \lvert \frac{1}{n}\sum_{i=1}^n x_{ij}y_i\rvert$. For $\ell_1$-regularized logistic regression, $\lmax=10/\max_j \lvert \frac{1}{n}\sum_{i=1}^n x_{ij}\widetilde{y}_i\rvert$ where $\widetilde{y}_i=y_i - \bar{y}(1-\bar{y})$ and $\bar{y}=\tfrac{1}{n}\sum_{i=1}^n y_i$. Both constructions ensure all features have selection probability close to zero at $\lmax$. For other regularized feature selection algorithms, $\lmax$ can be chosen via grid search.

Next, we define $\lmin$ by starting at $\lambda=\lmax$ and decreasing $\lambda$ over a grid of $100$ points that are evenly spaced on a log scale between $\lmax/10^{10}$ and $\lmax$, stopping once $\hat{S}_\lambda(\bm{Z}_{1:n})$ selects more than $p/2$ features. 
We then set $\lambda_0$ to be the smallest value of $\lambda$ in this grid such that fewer than $p/2$ features were selected. The idea is that we wish to quickly identify a reasonable lower endpoint that is not so close to $\lmax$ that the selection paths are cut off prematurely,  
but not so small that nearly all of the $p$ features are selected, since exceedingly small regularization values can cause unnecessary computational difficulties. Our experience suggests that $3p/4$ is an effective value for the maximum number of features selected and that results are insensitive to this choice. 

Finally, $\lmin$ is chosen as follows. Given the integral cutoff value $C>0$, recall that
\begin{align*}
    \lmin &= \inf\big\{\lambda\in (0,\lmax) : \mathcal{I}(\lambda,\lmax) \leq C\big\}
\end{align*}
according to \cref{eq:lipss} in the main text, where $\mathcal{I}$ is an integral over $[\lambda,\lmax]$. In practice, to approximate the infimum, we partition $\Lambda=[\lambda_0,\lmax]$ into a grid of $r$ points evenly spaced on a log scale.
We successively add the integrand of $\mathcal{I}$ evaluated at each value of $\lambda$ in the grid, starting from $\lmax$ and stopping either when $\lambda=\lambda_0$ or when the sum surpasses $C$, whichever comes first. In the latter case, $\lmin$ is the smallest $\lambda$ in the grid such that the corresponding Riemann sum in \cref{prop:riemann} is at most $C$.

%%%%%%%%%%%%%%%%%%%%%%%%%%%%%%%%%%%%%%%%%%%%%%%%%%%%%%%%%%%%%%%%

\subsection{Numerical approximation of the IPSS integral}\label{sup_sec:riemann}

%%%%%%%%%%%%%%%%%%%%%%%%%%%%%%%%%%%%%%%%%%%%%%%%%%%%%%%%%%%%%%%%

To implement IPSS, one must approximate the integral in \cref{eq:ipss}. We use a simple numerical approximation based on a Riemann sum, described in \cref{prop:riemann}. Consider the class of probability measures $\mu_\alpha(d\lambda)=z^{-1}_\alpha\lambda^{-\alpha}d\lambda$ on $\Lambda=[\lmin,\lmax]\subseteq (0,\infty)$, where $\alpha\in\mathbb{R}$ and the normalizing constant $z_\alpha$ is
\begin{align*}
	z_\alpha &= \int_{\Lambda}\lambda^{-\alpha}d\lambda
		=
		\begin{cases}
			\log\big(\lmax/\lmin\big) & \text{if}\ \alpha = 1 \\[8pt]
			\frac{1}{1-\alpha}\big(\lmax^{1-\alpha}-\lmin^{1-\alpha}\big) & \text{if}\ \alpha\neq 1.
		\end{cases}
\end{align*}
The parameter $\alpha$ controls the scale on which the $\lambda$ values are weighted. For instance, $\alpha=0$ and $\alpha=1$ correspond to linear and log scales, respectively.

\begin{prop}\label{prop:riemann}
Fix $\alpha\in\mathbb{R}$, let $\Lambda=[\lmin,\lmax]\subseteq(0,\infty)$, and define $\lambda_{k,r}=\lmin^{1-k/r}\lmax^{k/r}$ for all $r\in\mathbb{N}$, $k\in\{0,1,\dots,r\}$. For any Riemann integrable $g:\Lambda\to\mathbb{R}$,
\begin{align}\label{eq:riemann}
	\lim_{r\to\infty} \frac{1 - (\lmin/\lmax)^{1/r}}{z_\alpha}\sum_{k=1}^r \lambda_{k,r}^{1-\alpha}g(\lambda_{k,r}) &= \int_\Lambda g(\lambda)\mu_\alpha(d\lambda).
\end{align}
\end{prop}

\begin{proof}
Fix $m\in\mathbb{N}$. The sequence $\Lambda_r=(\lambda_{0,r},\lambda_{1,r},\dots,\lambda_{r,r})$ partitions $\Lambda$. Furthermore,
\begin{align}\label{eq:riemann-proof-1}
	\lambda_{k,r}-\lambda_{k-1,r} &= \Big(1 - \big(\lmin/\lmax\big)^{1/r}\Big)\lambda_{k,r}
\end{align}
for all $r$ and $k\in\{1,\dots,r\}$, and hence,
\begin{align}\label{eq:riemann-proof-2}
\max_{k}\big\lvert\lambda_{k,r}-\lambda_{k-1,r}\big\rvert \leq \Big\lvert 1 - \big(\lmin/\lmax\big)^{1/r}\Big\rvert |\lmax| \longrightarrow 0
\end{align}
as $r\to\infty$. Therefore, \cref{eq:riemann-proof-1,eq:riemann-proof-2} imply that
\begin{align*}
	 \frac{1 - (\lmin/\lmax)^{1/r}}{z_\alpha}\sum_{k=1}^r \lambda_{k,r}^{1-\alpha}g(\lambda_{k,r}) &= \frac{1}{z_\alpha}\sum_{k=1}^r \lambda_{k,r}^{-\alpha}g(\lambda_{k,r})(\lambda_{k,r}-\lambda_{k-1,r})
		\longrightarrow \int_\Lambda g(\lambda)\mu_\alpha(d\lambda)
\end{align*}
as $r\to\infty$ since $\lambda^{-\alpha} g(\lambda)$ is Riemann integrable on $\Lambda$.
\end{proof}

If $\alpha=1$ then $\lambda^{1-\alpha}_{k,r}=1$ for all $k$ and $r$. Thus, when using $f = h_m$ and $\mu = \mu_1$ as in the main text, computation of the IPSS criterion (\cref{eq:ipss}) amounts to 
\begin{align}\label{eq:compute}
\int_\Lambda f(\hat{\pi}_j(\lambda))\mu(d\lambda) =
	\int_\Lambda h_m(\hat{\pi}_j(\lambda))\mu_1(d\lambda) &\approx \frac{1 - (\lmin/\lmax)^{1/r}}{z_1}\sum_{k=1}^r h_m(\hat{\pi}_j(\lambda_{k,r}))
\end{align}
by applying \cref{prop:riemann} with $\alpha=1$ and $g(\lambda) = h_m(\hat{\pi}_j(\lambda))$. 
Furthermore, $h_m$ is a simple function of the estimated selection probabilities, making \cref{eq:compute} very inexpensive to compute once the estimated selection probabilities are obtained.

The use of \cref{prop:riemann} for IPSS is justified by \cref{prop:riemann-integrable}, which gives general conditions for functions of the form $\lambda\mapsto f(\hat{\pi}_j(\lambda))$ to be Riemann integrable when $\hat{S}_\lambda = \{j : \hat{\beta}_j(\lambda)\neq 0\}$ for some $\hat{\beta}_j:\Lambda\to\mathbb{R}$. 
In particular,  $h_m\circ\hat{\pi}_j$ is Riemann integrable when using lasso since $h_m$ is continuous and lasso regularization paths $\lambda\mapsto\hat{\beta}_j(\lambda)$ are continuous whenever no two features are perfectly correlated \citep{lasso_paths}. More generally, it follows from the proof of \cref{prop:riemann-integrable} that Riemann integrability of $h_m\circ\hat{\pi}_j$ is implied by Riemann integrability of the indicators $\1(j\in\hat{S}_\lambda)$.

\begin{prop}\label{prop:riemann-integrable}
Let $\Lambda$ be a closed interval in $\mathbb{R}$. Suppose $\hat{S}_\lambda$ is given by $\hat{S}_\lambda = \{j : \hat{\beta}_j(\lambda)\neq 0\}$ and that $\lambda\mapsto\hat{\beta}_j(\lambda)$ is continuous on $\Lambda$ for all $j\in\{1,\dots,p\}$. Then for every continuous function $f:[0,1]\to\mathbb{R}$, the composition $\lambda\mapsto f(\hat{\pi}_j(\lambda))$ is Riemann integrable.
\end{prop}

\begin{proof}
Fix $j\in\{1,\dots,p\}$. We first prove $\lambda\mapsto\1(j\in\hat{S}_\lambda)$ is Riemann integrable. Since a bounded function on a closed interval is Riemann integrable if and only if it is continuous almost everywhere, it suffices to show $\lambda\mapsto\1(j\in\hat{S}_\lambda)$ is continuous almost everywhere. To this end, we have
\begin{align*}
    \1(j\in\hat{S}_\lambda) &= \1(\hat{\beta}_j(\lambda)\neq 0)
        = \1\!\big(\lambda\in\hat{\beta}_j^{-1}(\mathbb{R}\setminus\{0\})\big).
\end{align*}
Since $\hat{\beta}_j$ is continuous, $U=\hat{\beta}_j^{-1}(\mathbb{R}\setminus\{0\})$ is an open subset of $\mathbb{R}$. A classical result from analysis states that every open subset of $\mathbb{R}$ is a countable union of disjoint open intervals. Letting $U=\bigcup_{k=1}^\infty (a_k,b_k)$ be such a union, we have
\begin{align*}
    \1\!\big(\lambda\in\hat{\beta}_j^{-1}(\mathbb{R}\setminus\{0\})\big) &= \1(\lambda\in U)
        = \1\!\big(\lambda\in{\textstyle\bigcup_{k=1}^\infty} (a_k,b_k)\big).
\end{align*}
Hence, $\1(j\in\hat{S}_\lambda)=\1(\lambda\in\bigcup_{k=1}^\infty (a_k,b_k))$, and it is clear from the latter expression that the set of discontinuities of $\1(j\in\hat{S}_\lambda)$ is contained in $\bigcup_{k=1}^\infty\{a_k,b_k\}$, which is countable and therefore has Lebesgue measure zero. Thus $\lambda\mapsto\1(j\in\hat{S}_\lambda)$ is Riemann integrable, and it immediately follows that $\hat{\pi}_j$ is Riemann integrable since it is a linear combination of indicators of the form $\1(j\in\hat{S}_\lambda)$. 

To conclude the proof, we show that $f\circ\hat{\pi}_j:\Lambda\to\mathbb{R}$ is bounded and continuous almost everywhere (and hence Riemann integrable). Boundedness holds because $f$ is continuous on the compact set $[0,1]$. For the continuous almost everywhere claim, let $A$ and $B$ be the sets of discontinuities of $\hat{\pi}_j$ and $f\circ\hat{\pi}_j$, respectively. We know from the first part of the proof that $A$ has measure zero. Thus, it suffices to show $B\subseteq A$ or, equivalently, $A^c\subseteq B^c$. Fix $\lambda\in A^c$ and let $\lambda_k$ be any sequence in $\Lambda$ converging to $\lambda$. Then $y_k=\hat{\pi}_j(\lambda_k)\to\hat{\pi}_j(\lambda)=y$ since $\hat{\pi}_j$ is continuous at $\lambda$. And since $f$ is continuous, $f(\hat{\pi}_j(\lambda_k))=f(y_k)\to f(y)=f(\hat{\pi}_j(\lambda))$. Thus $\lambda\in B^c$ and hence $A^c\subseteq B^c$.
\end{proof}

%%%%%%%%%%%%%%%%%%%%%%%%%%%%%%%%%%%%%%%%%%%%%%%%%%%%%%%%%%%%%%%%

\section{Proofs}\label{sup_sec:proofs}

%%%%%%%%%%%%%%%%%%%%%%%%%%%%%%%%%%%%%%%%%%%%%%%%%%%%%%%%%%%%%%%%

We begin by introducing some notation and preliminary results. Fix $B\geq 2$ throughout.  For $\lambda>0$ and $j\in\{1,\dots,p\}$, define the \textit{simultaneous selection probability}
\begin{align*}
	\widetilde{\pi}_j(\lambda) &= \frac{1}{B}\sum_{b=1}^B \1(j\in \hat{S}_\lambda(\bm{Z}_{A_{2 b-1}}))\1(j\in\hat{S}_\lambda(\bm{Z}_{A_{2 b}}))
		= \frac{1}{B}\sum_{b=1}^B U_{j b}(\lambda)
\end{align*}
where $U_{j b}(\lambda)=\1(j\in \hat{S}_\lambda(\bm{Z}_{A_{2 b-1}}))\1(j\in\hat{S}_\lambda(\bm{Z}_{A_{2 b}}))$. Here, $\hat{S}_\lambda$ and $A_{2 b - 1},A_{2 b}$ are defined as in \cref{sec:setup}. We use the following result (\cref{eq:shah1}), established in the proof of Lemma 1a of \citet{shah}: For all $\lambda>0$ and $j\in\{1,\dots,p\}$,
\begin{align}
& \hat{\pi}_j(\lambda) \leq \tfrac{1}{2}(1+\widetilde{\pi}_j(\lambda)). \label{eq:shah1}
\end{align}
This holds since
\begin{align*}
    0 \leq \frac{1}{B}\sum_{b=1}^B \Big(1 - \1(j\in \hat{S}_\lambda(\bm{Z}_{A_{2 b-1}}))\Big)\Big(1 - \1(j\in\hat{S}_\lambda(\bm{Z}_{A_{2 b}}))\Big) = 1 - 2\hat{\pi}_j(\lambda) + \widetilde{\pi}_j(\lambda).
\end{align*}

The following lemma is used in the proof of \cref{thrm:main}; it is essentially an application of the multinomial theorem. 
For readability, let us define $\Delta_m = \big\{k\in\mathbb{Z}^B : k_1,\ldots,k_B\geq 0,\, \sum_{b=1}^B k_b = m\big\}$,
denote the multinomial coefficients by
\begin{align*}
	\binom{m}{k_1,\dots,k_B} &= \frac{m!}{k_1!k_2!\cdots k_B!},
\end{align*}
and define $N_k = \sum_{b=1}^B \1(k_b \neq 0)$ for $k\in\Delta_m$. Note that the assumed measurability of 
$\hat{S}_\lambda(\bm{Z}_A):\Lambda\times\Omega\to 2^{\{1,\dots,p\}}$ implies measurability of $\widetilde{\pi}_j$, $\hat{\pi}_j$, and any continuous functions thereof.

\begin{lemma}\label{lem:second_moment}
Fix $m\in\mathbb{N}$. If \cref{cond:1} holds for all $m'\in\{1,\dots,m\}$, then for all $\lambda>0$,
\begin{align*}
	\max_{j\in S^c}\,\mathrm{E}\big(\widetilde{\pi}_j(\lambda)^m\big) &\leq \frac{1}{B^m}\sum_{k\in\Delta_m}\binom{m}{k_1,\dots,k_B}(q(\lambda)/p)^{2 N_k}.
\end{align*}
\end{lemma}

\begin{proof}
Fix $j\in S^c$ and define $0^0=1$. By the multinomial theorem,
\begin{align*}
	\widetilde{\pi}_j(\lambda)^m &= \frac{1}{B^m}\bigg(\sum_{b=1}^B U_{j b}(\lambda)\bigg)^m
		= \frac{1}{B^m}\sum_{k\in\Delta_m}\binom{m}{k_1,\dots,k_B}\prod_{b=1}^B U_{j b}(\lambda)^{k_b} \\
		&= \frac{1}{B^m}\sum_{k\in\Delta_m}\binom{m}{k_1,\dots,k_B}\prod_{b=1}^B U_{j b}(\lambda)^{\1(k_b\neq 0)}.
\end{align*}
The last equality holds because $U_{j b}(\lambda)\in\{0,1\}$, and hence, $U_{j b}(\lambda)^{k_b} = U_{j b}(\lambda)$ whenever $k_b>0$. Next, since \cref{cond:1} holds for all $m'\leq m$ and since every $k\in\Delta_m$ has at most $m$ nonzero integers,
\begin{align*}
    \mathbb{E}\prod_{b=1}^B U_{j b}(\lambda)^{\1(k_b\neq 0)} &= \mathbb{P}\bigg(j\in\bigcap_{b : k_b\neq 0} \big(\hat{S}_\lambda(\bm{Z}_{A_{2b-1}})\cap \hat{S}_\lambda(\bm{Z}_{A_{2b}})\big)\bigg)
        \leq (q(\lambda)/p)^{2N_k}.
\end{align*}
Therefore,
\begin{align*}
	\mathrm{E}\big(\widetilde{\pi}_j(\lambda)^m\big) &= \frac{1}{B^m}\sum_{k\in\Delta_m}\binom{m}{k_1,\dots,k_B}\mathrm{E}\prod_{b=1}^B U_{j b}(\lambda)^{\1(k_b\neq 0)} \\
		&\leq \frac{1}{B^m}\sum_{k\in\Delta_m}\binom{m}{k_1,\dots,k_B}(q(\lambda)/p)^{2 N_k}. \qedhere
\end{align*}
\end{proof}

\begin{proof}[\bf Proof of \texorpdfstring{\cref{thrm:main}}]

Fix $\tau\in(0,1]$ and $m\in\mathbb{N}$. Suppressing $\lambda$ and $\Lambda$ to declutter the notation,
\begin{align}
	\mathbb{P}\big(j\in\hat{S}_{\mathrm{IPSS},h_m}\big) &= \mathbb{P}\left(\int h_m(\hat{\pi}_j) d\mu\geq\tau\right)
		\leq \mathbb{P}\left(\int h_m\big(\tfrac{1}{2}(1+\widetilde{\pi}_j)\big) d\mu\geq\tau\right)
		= \mathbb{P}\left(\int\widetilde{\pi}_j^m d\mu\geq\tau\right) \notag\\
		&\leq \frac{1}{\tau}\int\mathrm{E}(\widetilde{\pi}_j^m) d\mu
		\leq \frac{1}{\tau B^m}\sum_{k\in\Delta_m}\binom{m}{k_1,\dots,k_B}\int(q/p)^{2 N_k}d\mu \label{eq:p_bound}
\end{align}
for all $j\in S^c$.
The first inequality holds by \cref{eq:shah1}, namely $\hat\pi_j\leq \frac{1}{2}(1+\widetilde{\pi}_j)$, and the fact that $h_m$ is monotonically increasing. The second equality holds since by the definition of $h_m$ in \cref{eq:half}, $h_m\big(\frac{1}{2}(1+x)\big) = x^m$ for $x\in[0,1]$. The second inequality is Markov's inequality followed by exchanging the order of integration and expectation, which is justified by Tonelli's theorem since $\widetilde{\pi}_j(\lambda)^m$ is nonnegative and measurable. The last inequality holds by \cref{lem:second_moment}. Thus,
\begin{align*}
	\mathrm{E}\lvert\hat{S}_{\mathrm{IPSS},h_m}\cap S^c\rvert &= \mathrm{E}\sum_{j=1}^p\1(j\in \hat{S}_{\mathrm{IPSS},h_m})\1(j\in S^c)
		= \sum_{j=1}^p\mathbb{P}\big(j\in\hat{S}_{\mathrm{IPSS},h_m}\big)\1(j\in S^c) \\
		&\leq \frac{1}{\tau B^m}\sum_{k\in\Delta_m}\binom{m}{k_1,\dots,k_B}\int(q/p)^{2 N_k}d\mu \sum_{j=1}^p\1(j\in S^c) \\
		&\leq \frac{p}{\tau B^m}\sum_{k\in\Delta_m}\binom{m}{k_1,\dots,k_B}\int(q/p)^{2 N_k}d\mu. \qedhere
\end{align*}
\end{proof}

\begin{proof}[\bf Proof of \texorpdfstring{\cref{thrm:simplified}}]
In the notation of this section, \cref{eq:ipss_bound} becomes
\begin{align*}
	\mathrm{E(FP)} \leq \frac{p}{\tau B^m}\sum_{k\in\Delta_m}\binom{m}{k_1,\ldots,k_B}\int (q/p)^{2 N_k}d\mu
\end{align*}
where $q = q(\lambda)$.
When $m=1$, each $k\in\Delta_1$ has a single nonzero entry $k_b = 1$; thus, $|\Delta_1| = B$ and $N_k = 1$ for all $k\in\Delta_1$. Therefore,
\begin{align*}
\mathrm{E(FP)} \leq
    \frac{p}{\tau B}\sum_{k\in\Delta_1}\binom{1}{k_1,\dots,k_B}\int (q/p)^{2 N_k} d\mu &= \frac{p}{\tau B}\sum_{b=1}^B\int (q/p)^2 d\mu
		= \frac{1}{\tau}\int\frac{q^2}{p} d\mu.
\end{align*}
When $m=2$, there are $B$ elements $k\in\Delta_2$ with $N_k = 1$ (each of which has $k_b = 2$ for exactly one $b$), and $\binom{B}{2}$ elements $k\in\Delta_2$ with $N_k = 2$ (each of which has $k_b = k_{b'} = 1$ for some $b\neq b'$). Therefore,
\begin{align*}
    \mathrm{E(FP)}\leq \frac{p}{\tau B^2}\sum_{k\in\Delta_2}\binom{2}{k_1,\dots,k_B}\int (q/p)^{2 N_k} d\mu &= \frac{1}{\tau B^2}\int\bigg(\frac{Bq^2}{p} + \frac{B(B-1)q^4}{p^3}\bigg) d\mu.
\end{align*}
When $m=3$, there are $B$ elements $k\in\Delta_3$ with $N_k = 1$ (each of which has $k_b = 3$ for one $b$), there are $2\binom{B}{2}$ elements $k\in\Delta_3$ with $N_k = 2$ (having $k_b = 2$ and $k_{b'}=1$ for some $b\neq b'$), and $\binom{B}{3}$ elements $k\in\Delta_3$ with $N_k = 3$ (having $k_b = k_{b'} = k_{b''} = 1$ for some distinct $b,b',b''$). Therefore,
\begin{align*}
\mathrm{E(FP)} &\leq
	\frac{p}{\tau B^3}\sum_{k\in\Delta_3}\binom{3}{k_1,\dots,k_B}\int (q/p)^{2 N_k} d\mu \\
 &= \frac{1}{\tau B^3}\int\bigg(\frac{Bq^2}{p} + \frac{3B(B-1)q^4}{p^3} + \frac{B(B-1)(B-2)q^6}{p^5}\bigg) d\mu. \qedhere
\end{align*}
\end{proof}

%%%%%%%%%%%%%%%%%%%%%%%%%%%%%%%%%%%%%%%%%%%%%%%%%%%%%%%%%%%%%%%%

\section{IPSS with other functions}\label{sup_sec:other_functions}

%%%%%%%%%%%%%%%%%%%%%%%%%%%%%%%%%%%%%%%%%%%%%%%%%%%%%%%%%%%%%%%%

In this section, we provide results about IPSS for other choices of $f$. Throughout, we fix $\Lambda\subseteq(0,\infty)$ and a probability measure $\mu$ on $\Lambda$. 

Our first result says that if a function $g$ dominates $f$, then the expected number of false positives selected by IPSS with $f$ is at most the expected number of false positives selected by IPSS with $g$. Similarly, the expected number of false negatives for IPSS with $g$ is at most the expected number of false negatives IPSS with $f$. We do not use \cref{lem:compare} in this work, but it gives intuition for how IPSS depends on the choice of $f$, and could lead to useful bounds on functions not considered here.

\begin{lemma}\label{lem:compare}
Let $f$ and $g$ be functions from $[0,1]$ to $\mathbb{R}$. If $f\leq g$ then for all $\tau\in(0,1]$,
\begin{align*}
	\mathrm{E}\lvert\hat{S}_{\mathrm{IPSS},f}\cap S^c\rvert \leq \mathrm{E}\lvert\hat{S}_{\mathrm{IPSS},g}\cap S^c\rvert \qquad\text{and}\qquad
	\mathrm{E}\lvert\hat{S}_{\mathrm{IPSS},g}^c\cap S\rvert \leq \mathrm{E}\lvert\hat{S}_{\mathrm{IPSS},f}^c\cap S\rvert.
\end{align*}
\end{lemma}

\begin{proof}
Fix $j\in\{1,\dots,p\}$ and $\tau\in(0,1]$. Since $f\leq g$,
\begin{align*}
    \mathbb{P}(j\in\hat{S}_{\mathrm{IPSS},f}) &= \mathbb{P}\left(\int f(\hat\pi_j)d\mu\geq\tau\right)
        \leq \mathbb{P}\left(\int g(\hat\pi_j)d\mu\geq\tau\right)
        = \mathbb{P}(j\in\hat{S}_{\mathrm{IPSS},g}).
\end{align*}
Therefore,
\begin{align*}
    \mathrm{E}\lvert\hat{S}_{\mathrm{IPSS},f}\cap S^c\rvert &= \sum_{j=1}^p \mathbb{P}(j\in\hat{S}_{\mathrm{IPSS},f})\1(j\in S^c)
        \leq \sum_{j=1}^p \mathbb{P}(j\in\hat{S}_{\mathrm{IPSS},g})\1(j\in S^c)
        = \mathbb{E}\lvert\hat{S}_{\mathrm{IPSS},g}\cap S^c\rvert.
\end{align*}
Similarly,
\begin{align*}
    \mathrm{E}\lvert\hat{S}_{\mathrm{IPSS},g}^c\cap S\rvert &= \sum_{j=1}^p \mathbb{P}(j\notin\hat{S}_{\mathrm{IPSS},g})\1(j\in S)
        = \sum_{j=1}^p (1 - \mathbb{P}(j\in\hat{S}_{\mathrm{IPSS},g}))\1(j\in S) \\
        &\leq \sum_{j=1}^p (1 - \mathbb{P}(j\in\hat{S}_{\mathrm{IPSS},f}))\1(j\in S)
        = \mathrm{E}\lvert\hat{S}_{\mathrm{IPSS},f}^c\cap S\rvert. \qedhere
\end{align*}
\end{proof}

%%%%%%%%%%%%%%%%%%%%%%%%%%%%%%%%%%%%%%%%%%%%%%%%%%%%%%%%%%%%%%%%

\subsection{Functions defined by monomials}\label{sec:whole}

%%%%%%%%%%%%%%%%%%%%%%%%%%%%%%%%%%%%%%%%%%%%%%%%%%%%%%%%%%%%%%%%

The functions $h_m$ defined in \cref{eq:half} are zero on $[0,0.5]$. 
Another natural class of functions $\{w_m:m\in\mathbb{N}\}$ on $[0,1]$ is defined by
\begin{align*}
    w_m(x) &= x^m.
\end{align*}
The $w$ stands for ``whole," indicating that these functions are positive on the whole unit interval. For example, TIGRESS \citep{tigress} uses a special case of IPSS with $f = w_1$; see \cref{sec:related_work}. \cref{thrm:whole} uses the following slight variation of \cref{cond:1}.

\begin{condition}\label{cond:2}
We say \textnormal{\cref{cond:2} holds for $m$} if for all $\lambda\in\Lambda$ and $\mathcal{B}\subseteq\{1,\dots,2B\}$ with $\lvert\mathcal{B}\rvert=m$,
\begin{align*}
    \max_{j\in S^c}\,\mathbb{P}\bigg(j\in\bigcap_{b\in\mathcal{B}} \hat{S}_\lambda(\bm{Z}_{A_b})\bigg) &\leq (q(\lambda)/p)^m.
\end{align*}
\end{condition}

\begin{theorem}\label{thrm:whole}
Let $\tau\in(0,1]$ and $m\in\mathbb{N}$. If \cref{cond:2} holds for all $m'\in\{1,\dots,m\}$, then
\begin{align*}
	\mathrm{E}\lvert\hat{S}_{\mathrm{IPSS},w_m}\cap S^c\rvert &\leq \frac{p}{\tau (2 B)^m}\sum_{k_1+\cdots+k_{2 B}=m}\binom{m}{k_1,\dots,k_{2 B}}\int_\Lambda (q(\lambda)/p)^{\sum_{b=1}^{2 B} \1(k_b \neq 0)} \mu(d\lambda)
\end{align*}
where $B$ is the number of subsampling steps in \cref{alg:ss} and the sum is over all nonnegative integers $k_1,\ldots,k_{2 B}$ such that $k_1 + \cdots + k_{2 B} = m$.
\end{theorem}

While \cref{thrm:main,thrm:whole} appear similar, there are two key differences. First, \cref{thrm:whole} has an additional factor of $2^m$ in the denominator, which is favorable in terms of tightness of the bound.  On the other hand, the integrand in \cref{thrm:whole} is $(q(\lambda)/p)^{N_k}$ rather than $(q(\lambda)/p)^{2 N_k}$ as in \cref{thrm:main}. Doubling the exponent makes the integral in \cref{thrm:main} significantly smaller than the one in \cref{thrm:whole}, since one typically has $q(\lambda)\ll p$ over much of $\Lambda$. This advantage outweighs the countervailing factor of $2^m$, which is why we recommend using $h_m$ in the main text.

\begin{proof}[Proof of \cref{thrm:whole}.]
Fix $\tau\in(0,1]$ and $m\in\mathbb{N}$. The proof is almost exactly the same as that of \cref{thrm:main}. Specifically, suppressing $\lambda$ and $\Lambda$ from the notation, the multinomial theorem gives
\begin{align*}
    \hat\pi_j^m &= \bigg(\frac{1}{2 B}\sum_{b=1}^{2 B}\1(j\in \hat{S}_\lambda(\bm{Z}_{A_b}))\bigg)^m
        = \frac{1}{(2 B)^m}\sum_{k_1+\cdots+k_{2 B}=m}\binom{m}{k_1,\dots,k_{2 B}}\prod_{b\,:\,k_b\neq 0} \1(j\in \hat{S}_\lambda(\bm{Z}_{A_b})).
\end{align*}
Since \cref{cond:2} holds for all $m'\leq m$ and since at most $m$ of $k_1,\ldots,k_{2 B}$ are nonzero,
\begin{align*}
    \mathrm{E}(\hat\pi_j^m) &= \frac{1}{(2 B)^m}\sum_{k_1+\cdots+k_{2 B}=m}\binom{m}{k_1,\dots,k_{2 B}}\mathrm{E}\prod_{b\,:\,k_b\neq 0} \1(j\in \hat{S}_\lambda(\bm{Z}_{A_b})) \\
        &\leq \frac{1}{(2 B)^m}\sum_{k_1+\cdots+k_{2 B}=m}\binom{m}{k_1,\dots,k_{2 B}}(q/p)^{N_k}
\end{align*}
where $N_k = \sum_{b=1}^{2 B} \1(k_b \neq 0)$. So for any $j\in\{1,\dots,p\}$, Markov's inequality gives
\begin{align*}
    \mathbb{P}(j\in\hat{S}_{\mathrm{IPSS},w_m}) &= \mathbb{P}\left(\int w_m(\hat\pi_j)d\mu\geq\tau\right)
        = \mathbb{P}\left(\int\hat\pi_j^md\mu\geq\tau\right)
        \leq\frac{1}{\tau}\int\mathrm{E}(\hat\pi_j^m)d\mu \\
        &\leq \frac{1}{\tau(2 B)^m}\sum_{k_1+\cdots+k_{2 B}=m}\binom{m}{k_1,\dots,k_{2 B}}\int(q/p)^{N_k}d\mu,
\end{align*}
where Tonelli's theorem justifies interchanging the order of expectation and integration.
Therefore,
\begin{align*}
    \mathbb{E}\lvert\hat{S}_{\mathrm{IPSS},w_m}\cap S^c\rvert &= \sum_{j=1}^p \mathbb{P}(j\in\hat{S}_{\mathrm{IPSS},w_m})\1(j\in S^c) \\
        &\leq \frac{1}{\tau(2 B)^m}\sum_{k_1+\cdots+k_{2 B}=m}\binom{m}{k_1,\dots,k_{2 B}}\int(q/p)^{N_k}d\mu \sum_{j=1}^p\1(j\in S^c) \\
        &\leq \frac{p}{\tau(2 B)^m}\sum_{k_1+\cdots+k_{2 B}=m}\binom{m}{k_1,\dots,k_{2 B}}\int(q/p)^{N_k}d\mu. \qedhere
\end{align*}
\end{proof}

%%%%%%%%%%%%%%%%%%%%%%%%%%%%%%%%%%%%%%%%%%%%%%%%%%%%%%%%%%%%%%%%

\subsection{Some piecewise linear functions}\label{sec:piecewise}

%%%%%%%%%%%%%%%%%%%%%%%%%%%%%%%%%%%%%%%%%%%%%%%%%%%%%%%%%%%%%%%%

For $\delta\in[0,1)$ define $f_\delta:[0,1]\to[0,1]$ by
\begin{align*}
	f_\delta(x) &= \left(\frac{x-\delta}{1-\delta}\right)\1(x\geq\delta).
\end{align*}
Note that $f_\delta$ is identically $0$ on $[0,\delta]$ and then rises linearly on $[\delta,1]$. \cref{thrm:linear} is a continuous analogue of Theorem 1 from \citet{mb} when $\delta=0$.

\begin{theorem}\label{thrm:linear}
Let $\delta\in[0,1)$, 
$\tau\in(0,1]$ such that $\tau > \tfrac{1-2\delta}{2(1-\delta)}$.
If \cref{cond:1} holds for $m=1$, then
\begin{align*}
	\mathrm{E}\lvert\hat{S}_{\mathrm{IPSS},f_\delta}\cap S^c\rvert &\leq \frac{1}{2(\tau-\tau\delta+\delta)-1}\int_\Lambda \frac{q(\lambda)^2}{p} \mu(d\lambda).
\end{align*}
\end{theorem}

\begin{proof}
Let $\delta\in[0,1)$ and $\tau\in(0,1]$ such that $\tau > (1-2\delta)/(2-2\delta)$. Fix $j\in S^c$. We have
\begin{align*}
    \mathbb{P}(j\in\hat{S}_{\mathrm{IPSS},f_\delta}) &= \mathbb{P}\left(\int f_\delta(\hat{\pi}_j)d\mu\geq\tau\right) 
	\leq \mathbb{P}\left(\frac{1}{1-\delta}\int \big(\tfrac{1}{2}(1+\widetilde{\pi}_j)-\delta\big)\!\ d\mu\geq\tau\right) \\
	&= \mathbb{P}\left(\int \widetilde{\pi}_j \,d\mu\geq 2(\tau-\delta\tau+\delta)-1\right)
	\leq \frac{1}{2(\tau-\delta\tau+\delta)-1}\int\mathrm{E}(\widetilde{\pi}_j)d\mu \\
	&\leq \frac{1}{2(\tau-\delta\tau+\delta)-1}\int (q/p)^2d\mu.
\end{align*}
The first inequality holds since 
$f_\delta(\hat{\pi}_j) \leq f_\delta\big(\frac{1}{2}(1 + \widetilde{\pi}_j)\big) = (\frac{1}{2}(1 + \widetilde{\pi}_j) - \delta) / (1 - \delta)$ by \cref{eq:shah1} and because $f_\delta$ is monotonically increasing with $f_\delta(x) = (x-\delta)/(1-\delta)$ for $x\geq 1/2$. The second inequality is by Markov's inequality and Tonelli's theorem. The final inequality follows from the definition of $\widetilde{\pi}_j$ and the assumption that \cref{cond:1} holds for $m=1$.
The assumption $\tau>(1-2\delta)/(2-2\delta)$ implies $2(\tau-\delta\tau+\delta)-1>0$, so the leading fraction is well-defined. Therefore,
\begin{align*}
	\mathrm{E}\lvert\hat{S}_{\mathrm{IPSS},f_\delta}\cap S^c\rvert &= \mathrm{E}\bigg(\sum_{j=1}^p\1(j\in \hat{S}_{\mathrm{IPSS},f_\delta})\1(j\in S^c)\bigg)
	= \sum_{j=1}^p \mathbb{P}(j\in\hat{S}_{\mathrm{IPSS},f_\delta})\1(j\in S^c) \\
	&\leq \frac{1}{2(\tau-\delta\tau+\delta)-1}\int \frac{q^2}{p^2}d\mu \sum_{j=1}^p\1(j\in S^c)d\mu
	\leq \frac{1}{2(\tau-\delta\tau+\delta)-1}\int \frac{q^2}{p}d\mu. \qedhere
\end{align*}
\end{proof}

%%%%%%%%%%%%%%%%%%%%%%%%%%%%%%%%%%%%%%%%%%%%%%%%%%%%%%%%%%%%%%%%

\section{IPSS with other base estimators}\label{sup_sec:other_estimators}

%%%%%%%%%%%%%%%%%%%%%%%%%%%%%%%%%%%%%%%%%%%%%%%%%%%%%%%%%%%%%%%%

IPSS can be used with any base estimator $\hat{S}_\lambda$ that performs feature selection and depends upon a parameter $\lambda\in\mathbb{R}$. Not all methods can be used, however. One such example is the debiased lasso, a popular approach for producing an unbiased estimator $\hat{\bm{\beta}}_u$ of the true coefficient vector $\bm{\beta}^*$ in the linear model $Y = X^{\mathtt{T}}\bm{\beta}^* + \epsilon$ \citep{debiased_lasso}. Given observations $(\bm{x}_i,y_i)$ of $n$ random vectors $(\bm{X}_i,Y_i)$ and a regularization value $\lambda>0$, this estimator takes the form
\begin{align*}
	\hat{\bm{\beta}}_u(\lambda) &= \hat{\bm{\beta}}(\lambda) + \frac{1}{n}M{\bf X}^{\mathtt{T}}({\bf Y} - {\bf X}\hat{\bm{\beta}}(\lambda))
\end{align*} 
where ${\bf X}\in\mathbb{R}^{n\times p}$ is the design matrix, ${\bf Y}=(y_1,\dots,y_n)^{\mathtt{T}}$, $\hat{\bm{\beta}}(\lambda)$ is the lasso solution (\cref{eq:lasso}), and $M\in\mathbb{R}^{p\times p}$ is the solution of a certain convex program \citep{debiased_lasso}. The problem with respect to IPSS is that every entry of the correction term $M\bf{X}^{\mathtt{T}}({\bf Y} - {\bf X}\hat{\bm{\beta}}(\lambda))$ is usually nonzero, even for features $X_j$ with $\beta^*_j=0$. Consequently, all entries of $\hat{\bm{\beta}}_u(\lambda)$ are nonzero and hence $\hat{S}_\lambda = \{j : \hat{\bm{\beta}}_u(\lambda) \neq 0\} = \{1,\dots,p\}$. Thus, unlike lasso, whose solutions typically include zeros, the debiased lasso does not inherently perform feature selection in the finite sample setting, making it incompatible with IPSS and stability selection.

A general class of estimators in the linear regression setting that perform feature selection and are therefore compatible with IPSS is
\begin{align}\label{eq:general_penalty}
	\hat{\bm{\beta}}_\varphi(\lambda) &= \argmin_{\bm{\beta}\in\mathbb{R}^p}\, \frac{1}{2}\sum_{i = 1}^n (y_i-\bm{x}_i^\mathtt{T}\bm{\beta})^2 + \sum_{j=1}^p \varphi(|\beta_j|, \lambda),
\end{align}
where $\varphi$ is a penalty function. We consider three estimators based on \cref{eq:general_penalty}. First, lasso is the case of $\varphi(x,\lambda)=\lambda x$ \citep{lasso}, which we implement using the \texttt{scikit-learn} Python package \citep{sklearn}. Second, the smoothly clipped absolute deviation penalty (SCAD), introduced by \cite{scad}, is
\begin{align*}
	\varphi_{\text{SCAD}}(x,\lambda) &=
	\begin{cases}
		\lambda |x| & \text{if}\ |x| \leq \lambda, \\
		\displaystyle\frac{2\gamma\lambda|x| - x^2 - \lambda^2}{2(\gamma - 1)} & \text{if}\ \lambda < |x| \leq \gamma\lambda, \\
		\displaystyle\frac{\lambda^2(\gamma + 1)}{2} & \text{if}\ |x| > \gamma\lambda.
	\end{cases}
\end{align*}
Following \citet{scad}, we use $\gamma = 3.7$ in all our experiments. Given $\gamma$, the solution $\hat{\bm{\beta}}_{\varphi_{\text{SCAD}}}(\lambda)$ to \cref{eq:general_penalty} depends only on $\lambda$ and can be used as a base estimator for IPSS.

Third, the minimax concave penalty, MCP \citep{mcp}, is another case of \cref{eq:general_penalty} where
\begin{align*}
	\varphi_{\text{MCP}}(x, \lambda) &=
	\begin{cases}
		\displaystyle\lambda|x| - \frac{x^2}{2\gamma} & \text{if}\ |x|\leq\gamma\lambda, \\
		\displaystyle\frac{\gamma\lambda^2}{2} & \text{if}\ |x| > \gamma\lambda.
	\end{cases}
\end{align*}
We use $\gamma=3$ in all our experiments, which is the default setting in the Python package \texttt{skglm} that we use to implement both MCP and SCAD \citep{skglm}. As with SCAD, once $\gamma$ is fixed, the resulting MCP estimator can be used with IPSS.

In addition to the three estimators based on \cref{eq:general_penalty}, we consider IPSS with the adaptive lasso as the base estimator \citep{adaptive_lasso}. In this case, the estimated coefficient vector is
\begin{align*}
	\hat{\bm{\beta}}_{\text{adap}} &= \argmin_{\bm{\beta}\in\mathbb{R}^p}\, \frac{1}{2}\sum_{i = 1}^n (y_i-\bm{x}_i^\mathtt{T}\bm{\beta})^2 + \lambda\sum_{j=1}^p \hat{w}_j|\beta_j|,
\end{align*}
where $\hat{w}_j$ are data-dependent weights that modulate the regularization applied to each coefficient. Following \cite{adaptive_lasso}, we set $\hat{w}_j = 1/\lvert\tilde{\beta}_j\rvert$ where $\tilde{\bm{\beta}}$ is the solution to the ridge regression problem \citep{ridge_regression}, computed using the default settings in \texttt{scikit-learn} \citep{sklearn}. We use ridge regression to compute $\tilde{\bm{\beta}}$ since ordinary least squares is not applicable when $p>n$, as is the case in our simulations.

Finally, the optimization problem for $\ell_1$-regularized logistic regression is
\begin{align}\label{eq:logistic_regression}
	\hat{\bm{\beta}}(\lambda) &= \argmin_{\bm{\beta}\in\mathbb{R}^p} -\sum_{i=1}^n \big[y_i\log(p_i) + (1-y_i)\log(1 - p_i)\big] + \lambda\sum_{j=1}^p|\beta_j|,
\end{align} 
where $p_i = 1/(1 + \exp(-\bm{x}_i^{\mathtt{T}}\bm{\beta}))$ \citep{logistic_regression}. Here, the response variables $Y_i$ take values in $\{0,1\}$. Like lasso, the $\ell_1$-penalty in \cref{eq:logistic_regression} typically shrinks some entries of $\hat{\bm{\beta}}(\lambda)$ to zero, making $\ell_1$-regularized logistic regression compatible with IPSS.

%%%%%%%%%%%%%%%%%%%%%%%%%%%%%%%%%%%%%%%%%%%%%%%%%%%%%%%%%%%%%%%%

\section{Empirical study of the theoretical conditions}\label{sup_sec:condition}

%%%%%%%%%%%%%%%%%%%%%%%%%%%%%%%%%%%%%%%%%%%%%%%%%%%%%%%%%%%%%%%%

Our theory is based on \cref{cond:1}, which is less stringent than the exchangeability and not-worse-than-random-guessing conditions of \cite{mb}; we refer to these as the MB conditions for short.
However, \cref{cond:1} does not hold in general. For example, these conditions tend to be violated when true and non-true features are highly correlated.
Nonetheless, even then, empirically we find that only a small proportion of non-true features violate \cref{cond:1} within the interval $[\lmin,\lmax]$; see \cref{fig:condition}.
This also illustrates how our construction of $\lmin$ prevents IPSS from integrating over small regularization values where \cref{cond:1} tends to fail.

Additionally, in \cref{fig:condition} we see that for any given value of $\lambda$, \cref{cond:1} tends to be more likely to hold as $m$ increases. This subtle point contributes to the improved performance of IPSS relative to stability selection. Recall that the MB conditions imply \cref{cond:1} when $m=1$. Hence, for any given $\lambda$, the more frequent violations of \cref{cond:1} when $m=1$ imply that the MB conditions are more frequently violated, compared to \cref{cond:1} for $m \in\{2,3\}$.
Furthermore, while \cref{eq:quadratic,eq:cubic} require \cref{cond:1} to hold for $m\in\{1,2\}$ and $m\in\{1,2,3\}$, respectively, the $1/B$ and $1/B^2$ terms in their integrands mitigate the higher proportion of failures of \cref{cond:1}. This partly explains why the target E(FP) is relatively well-approximated in all of our simulation results, even when \cref{cond:1} does not necessarily hold for all features.

\ifthenelse{\boolean{showfigures}}{
\begin{figure*}[!ht]
\makebox[\textwidth]{\includegraphics[width=\textwidth, height=.2\textheight]{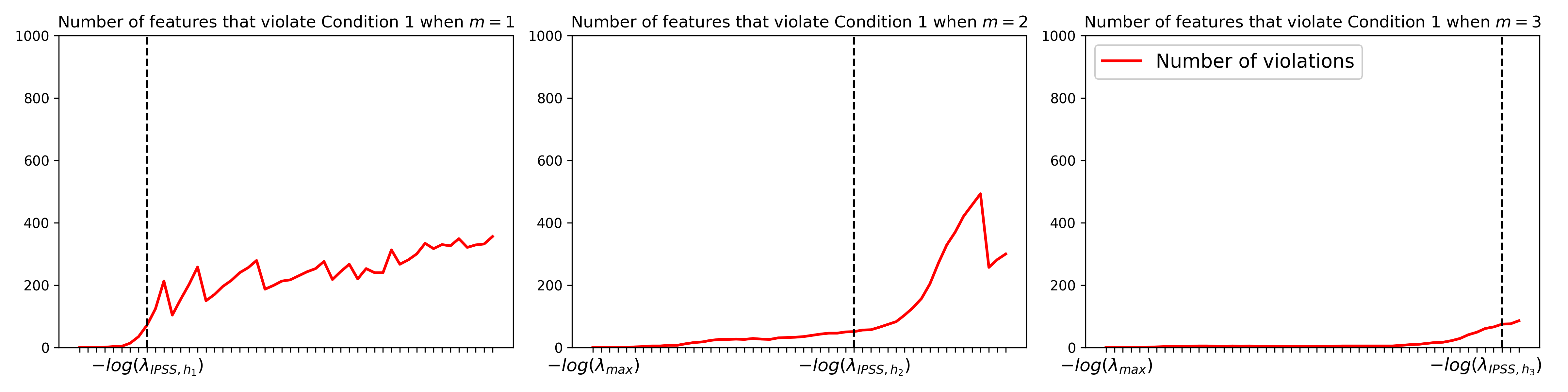}}%
\caption{\textit{Violations of \cref{cond:1} for highly correlated data}. (\textit{Left}) Number of features that violate \cref{cond:1} when $m=1$ for each value of $\lambda$. (\textit{Middle, Right}) Same as the left plot but for $m=2$ and $3$. The dashed lines show $-\log(\lambda_{\mathrm{IPSS},h_m})$, where $\lambda_{\mathrm{IPSS},h_m}$ is the $\lmin$ for IPSS with $h_m$. The horizontal axis, shown on a log scale, is the same in all three plots.
Data are simulated from a linear regression model with normal residuals and Toeplitz design with $\rho=0.9$ (\cref{sec:simulations}), with $n=200$, $p=1000$, $s=20$, and $\mathrm{SNR}=1$.
}
\label{fig:condition}
\end{figure*}
}{}

%%%%%%%%%%%%%%%%%%%%%%%%%%%%%%%%%%%%%%%%%%%%%%%%%%%%%%%%%%%%%%%%

\section{Additional application results}\label{sup_sec:applications}

%%%%%%%%%%%%%%%%%%%%%%%%%%%%%%%%%%%%%%%%%%%%%%%%%%%%%%%%%%%%%%%%

\cref{tab:prostate} lists all of the proteins that are identified by the IPSS and stability selection methods in our prostate cancer RPPA study in \cref{sec:prostate}. A \cmark\, indicates that the protein was selected by the corresponding method, and an \xmark\, indicates that it was not selected. The references next to each protein provide further evidence that these proteins are related to prostate cancer. \cref{tab:colon} is similar to \cref{tab:prostate}, but for the colon cancer study in \cref{sec:colon}. 

\begin{table}[!ht]
\centering
\resizebox{\textwidth}{!}{
\begin{tabular}{|l|c|c|c|c|c|}
\hline
\textbf{Protein} & \textbf{IPSS(cubic)} & \textbf{IPSS(quad)} & \textbf{r-concave} & \textbf{UM} & \textbf{MB} \\
\hline
BAK1 \citep{bak1}     &     \cmark       &     \xmark       &      \xmark      &     \xmark       &    \xmark        \\
\hline
DIRAS3 \citep{diras3} &     \cmark      &     \cmark       &     \cmark       &     \cmark       &    \cmark        \\
\hline
MAPK9 \citep{mapk9}   &     \cmark       &     \cmark       &     \cmark       &    \cmark        &     \cmark       \\
\hline
STK11 \citep{stk11}   &     \cmark       &     \cmark       &      \xmark      &      \xmark      &      \xmark      \\
\hline
NOTCH1 \citep{notch1} &     \cmark       &     \cmark       &     \cmark       &     \cmark       &    \cmark        \\
\hline
PKC \citep{pkc}       &     \cmark       &      \cmark      &      \xmark      &      \xmark      &    \xmark        \\
\hline
PTEN \citep{pten}     &     \cmark       &     \xmark       &      \xmark      &     \xmark       &     \xmark       \\
\hline
SMAD1 \citep{smad1}   &     \cmark       &     \cmark       &      \xmark      &     \xmark       &     \xmark       \\
\hline
EIF4E \citep{eif4e}   &     \cmark       &     \cmark      &     \cmark       &      \cmark      &     \cmark       \\
\hline
SQSTM1 \citep{sqstm1} &     \cmark       &     \cmark       &       \cmark     &     \xmark       &     \xmark       \\
\hline
\end{tabular}}
\caption{\textit{Selected proteins from the prostate cancer RPPA data set}.}
\label{tab:prostate}
\end{table}

% Colon
\begin{table}[!ht]
\centering
\resizebox{\textwidth}{!}{
\begin{tabular}{|l|c|c|c|c|c|}
\hline
\textbf{Gene} & \textbf{IPSS(cubic)} & \textbf{IPSS(quad)} & \textbf{r-concave} & \textbf{UM} & \textbf{None} \\
\hline
R87126 \citep{h08393_j02854_MYL9_m26383_IL1_r87126_t47377_S100} & \cmark & \cmark & \cmark & \cmark & \cmark \\
\hline
M63391 \citep{m63391_desmin} & \cmark & \cmark & \cmark & \cmark & \xmark \\
\hline
H08393 \citep{h08393_j02854_MYL9_m26383_IL1_r87126_t47377_S100} & \cmark & \cmark & \cmark & \xmark & \xmark \\
\hline
Z50753 \citep{z50753_GUCA2B} & \cmark & \cmark & \cmark & \xmark & \xmark \\
\hline
J02854 \citep{h08393_j02854_MYL9_m26383_IL1_r87126_t47377_S100} & \cmark & \cmark & \cmark & \xmark & \xmark \\
\hline
M26383 \citep{h08393_j02854_MYL9_m26383_IL1_r87126_t47377_S100} & \cmark & \cmark & \cmark & \xmark & \xmark \\
\hline
H06524 \citep{h06524} & \cmark & \cmark & \xmark & \xmark & \xmark \\
\hline
X12671 \citep{x12671_HNRNPA1} & \cmark & \cmark & \cmark & \xmark & \xmark \\
\hline
R54097 \citep{r54097_EIF2S2} & \cmark & \xmark & \xmark & \xmark & \xmark \\
\hline
X63629 \citep{x63629_pcadherin} & \cmark & \cmark & \xmark & \xmark & \xmark \\
\hline
T62947 \citep{t62947} & \cmark & \xmark & \xmark & \xmark & \xmark \\
\hline
M36634 \citep{m36634_VIP} & \cmark & \cmark & \xmark & \xmark & \xmark \\
\hline
T92451 \citep{h20709_t92451} & \cmark & \xmark & \xmark & \xmark & \xmark \\
\hline
T47377 \citep{h08393_j02854_MYL9_m26383_IL1_r87126_t47377_S100} & \cmark & \xmark & \xmark & \xmark & \xmark \\
\hline
R36977 \citep{r36977_GTF3A} & \cmark & \cmark & \xmark & \xmark & \xmark \\
\hline
H20709 \citep{h20709_t92451} & \cmark & \xmark & \xmark & \xmark & \xmark \\
\hline
\end{tabular}}
\caption{\textit{Selected genes from the colon cancer genomics data set}. Genes are identified by their GenBank accession numbers; here is a list of some of their official symbols (in parentheses):
J02854 (MYL9), 
M26383 (IL-1), 
T47377 (S-100),
J02854 (MYL9),
%J05032 (DARS1) \citep{j05032_DARS1}, 
M26383 (CXCL8), 
M36634 (VIP), 
M63391 (desmin),
% M82919 (GABAA) \citep{m82919_gabaa},
R36977 (GTF3A),
R54097 (EIF2S2),
X12671 (hnRNP A1), 
X63629 (P-cadherin),
Z50753 (GUCA2B).}
\label{tab:colon}
\end{table}

%H06524 \citep{h06524},
%% H08393 and R87126 \citep{h08393_r87126},
%H08393, J02854 (MYL9), M26383 (IL-1), R87126, and T47377 (S-100) \citep{h08393_j02854_MYL9_m26383_IL1_r87126_t47377_S100},
%% J02854 (MYL9) \citep{j02854_MYL9},
%J05032 (DARS1) \citep{j05032_DARS1}, 
%% M26383 (CXCL8) \citep{m26383_CXCL8}, 
%M36634 (VIP) \citep{m36634_VIP}, 
%M63391 (desmin) \citep{m63391_desmin},
%% M82919 (GABAA) \citep{m82919_gabaa},
%R36977 (GTF3A) \citep{r36977_GTF3A},
%R54097 (EIF2S2) \citep{r54097_EIF2S2},
%X12671 (hnRNP A1) \citep{x12671_HNRNPA1}, 
%X63629 (P-cadherin) \citep{x63629_pcadherin},
%Z50753 (GUCA2B) \citep{z50753_GUCA2B}.

\clearpage

%%%%%%%%%%%%%%%%%%%%%%%%%%%%%%%%%%%%%%%%%%%%%%%%%%%%%%%%%%%%%%%%

\section{Additional simulation results}\label{sup_sec:simulations}

%%%%%%%%%%%%%%%%%%%%%%%%%%%%%%%%%%%%%%%%%%%%%%%%%%%%%%%%%%%%%%%%

All results in this section are from the simulation experiments described in \cref{sec:simulations} of the main text. \cref{fig:compare_linear_reg_200_df2,fig:compare_linear_reg_1000_df2} show results for linear regression for Student's $t$ residuals with $2$ degrees of freedom, \cref{fig:compare_linear_reg_200_mcp,fig:compare_linear_reg_1000_mcp} show results for linear regression with normal residuals when the base estimator is MCP, \cref{fig:compare_linear_reg_200_scad,fig:compare_linear_reg_1000_scad} show results for linear regression with normal residuals when the base estimator is SCAD, \cref{fig:compare_linear_reg_200_adaptive,fig:compare_linear_reg_1000_adaptive} show results for linear regression with normal residuals when the base estimator is adaptive lasso, and \cref{fig:compare_linear_class_200,fig:compare_linear_class_1000} show results for logistic regression when the base estimator is $\ell_1$-regularized logistic regression. In all cases, the same base estimator and stability paths are used for each of the IPSS and stability selection methods, and the IPSS parameters are always set to their default values, as described in \cref{sup_sec:sensitivity_to_mu}.

\ifthenelse{\boolean{showfigures}}{
\begin{figure*}[!ht]
\includegraphics[width=0.85\textwidth, height=.275\textheight]{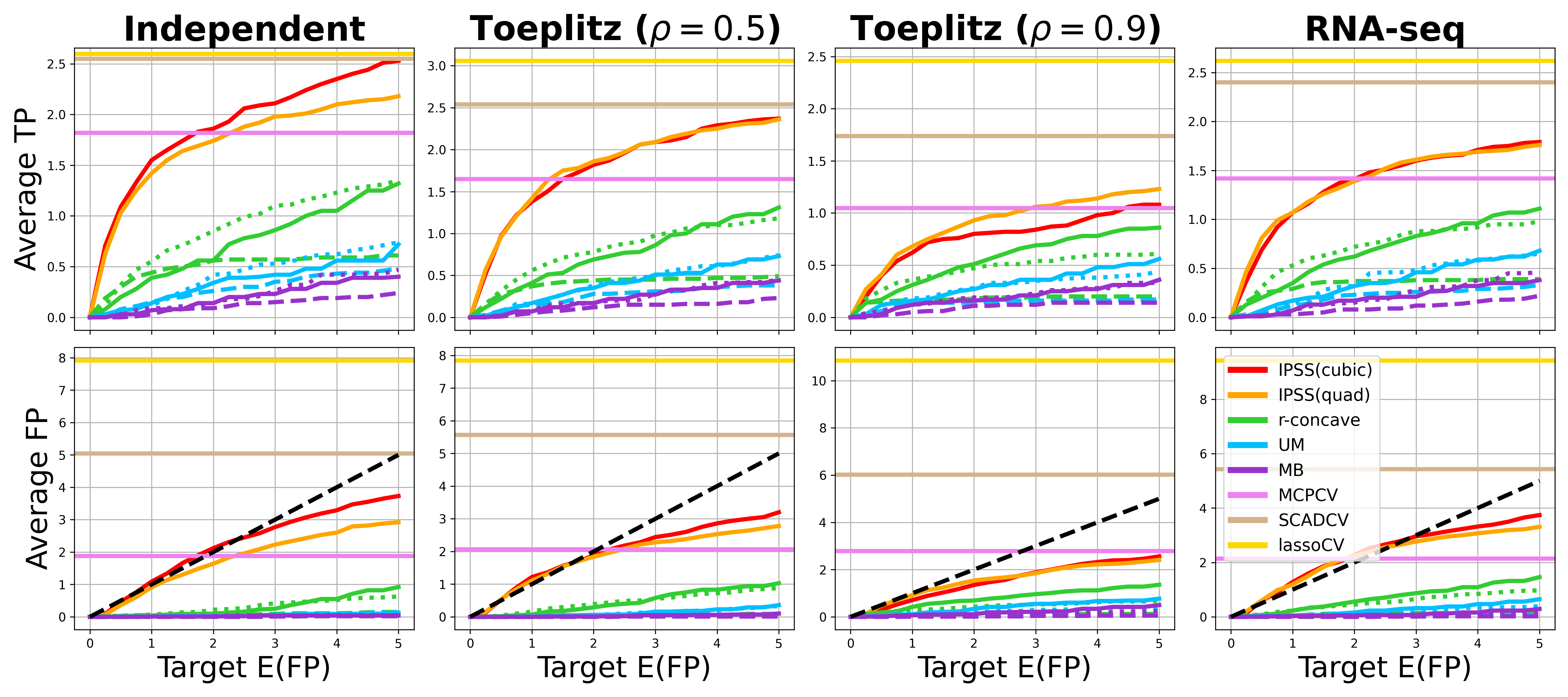}%
\caption{\textit{Linear regression, Student's $t$ residuals: Lasso $(p=200)$}. The solid, dotted, and dashed lines for the stability selection methods represent $\tau=0.6$, $0.75$, and $0.9$, respectively. Cross-validation results (horizontal lines) are independent of Target E(FP). The dashed black line represents perfect E(FP) control.}
\label{fig:compare_linear_reg_200_df2}
\end{figure*}

\begin{figure*}[!ht]
\includegraphics[width=0.85\textwidth, height=.275\textheight]{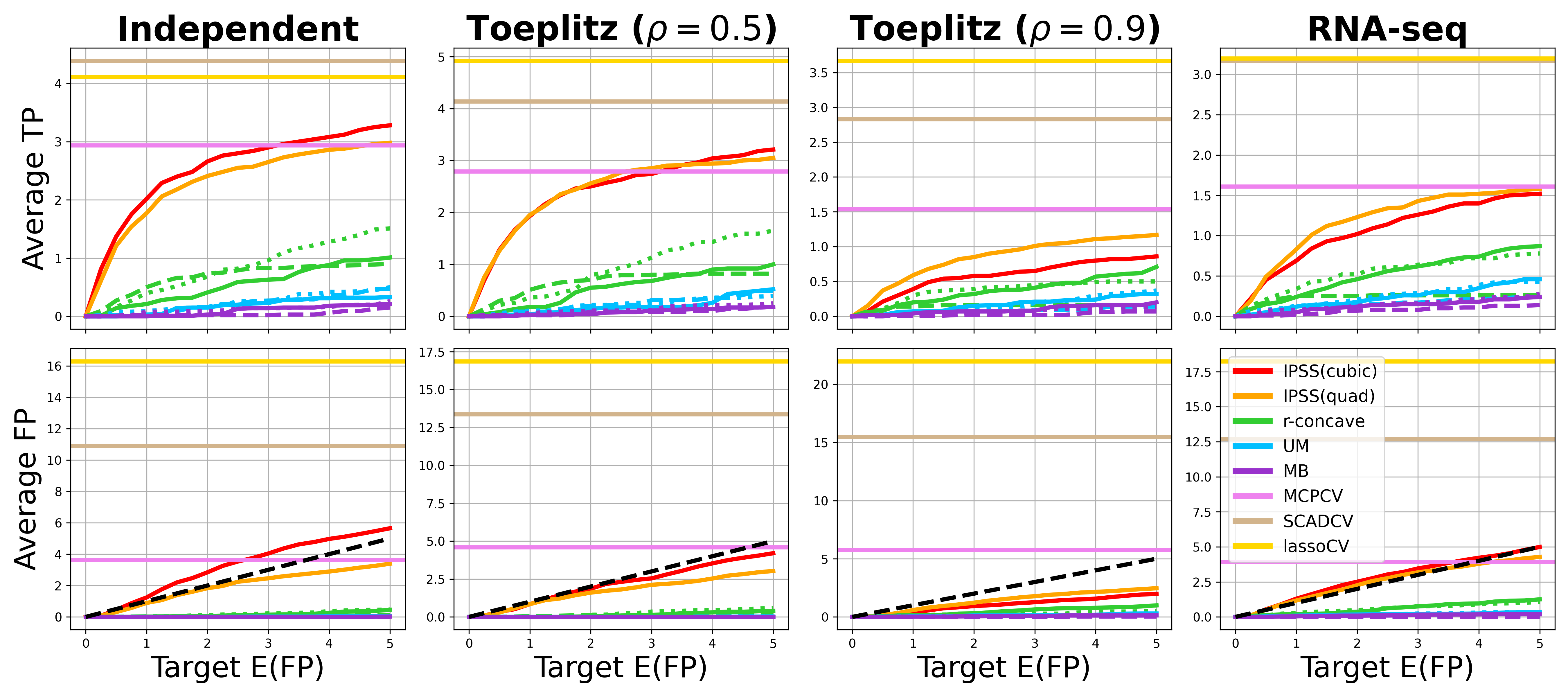}%
\caption{\textit{Linear regression, Student's $t$ residuals: Lasso $(p=1000)$}. See \cref{fig:compare_linear_reg_200_df2}.}
\label{fig:compare_linear_reg_1000_df2}
\end{figure*}
}{}

\ifthenelse{\boolean{showfigures}}{
\begin{figure*}[!ht]
\includegraphics[width=0.85\textwidth, height=.275\textheight]{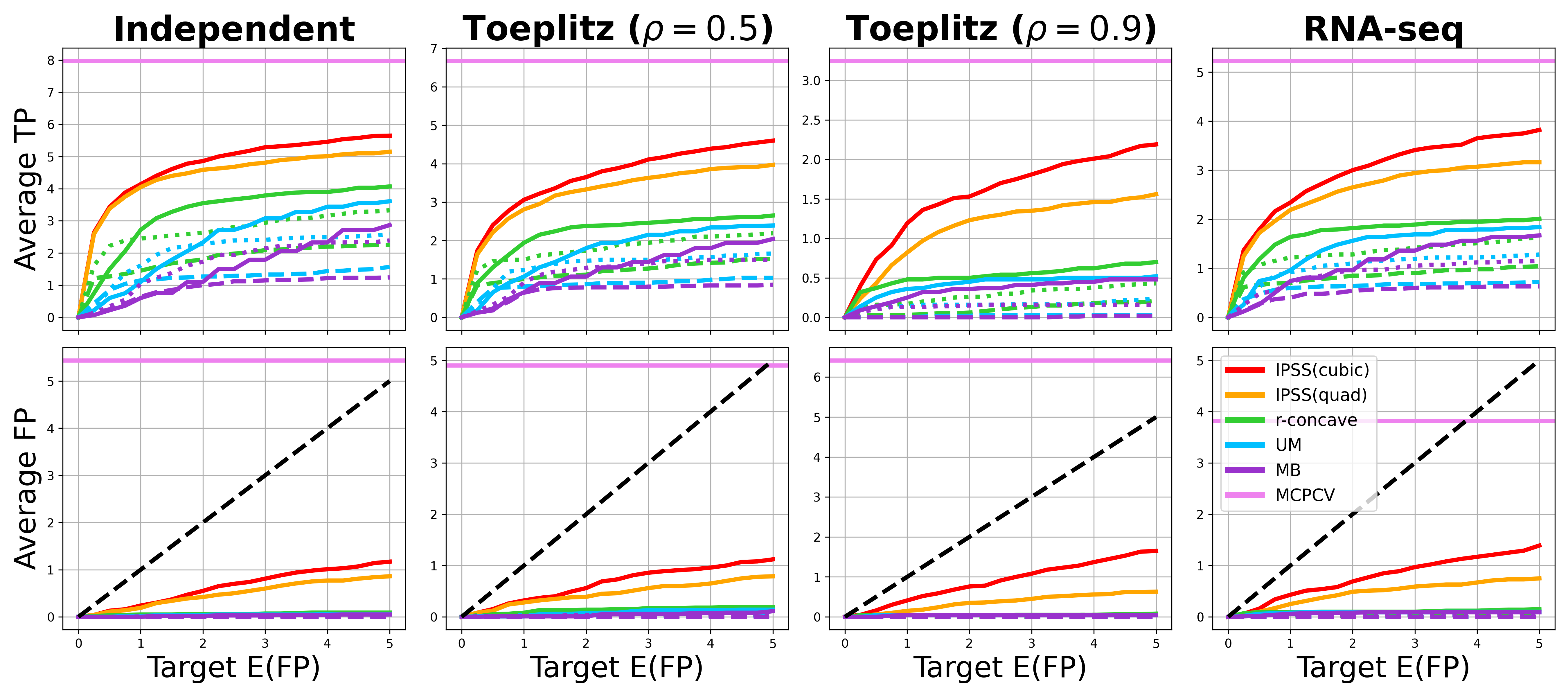}%
\caption{\textit{Linear regression, normal residuals: MCP $(p=200)$}. See \cref{fig:compare_linear_reg_200_df2}.}
\label{fig:compare_linear_reg_200_mcp}
\end{figure*}

\begin{figure*}[!ht]
\includegraphics[width=0.85\textwidth, height=.275\textheight]{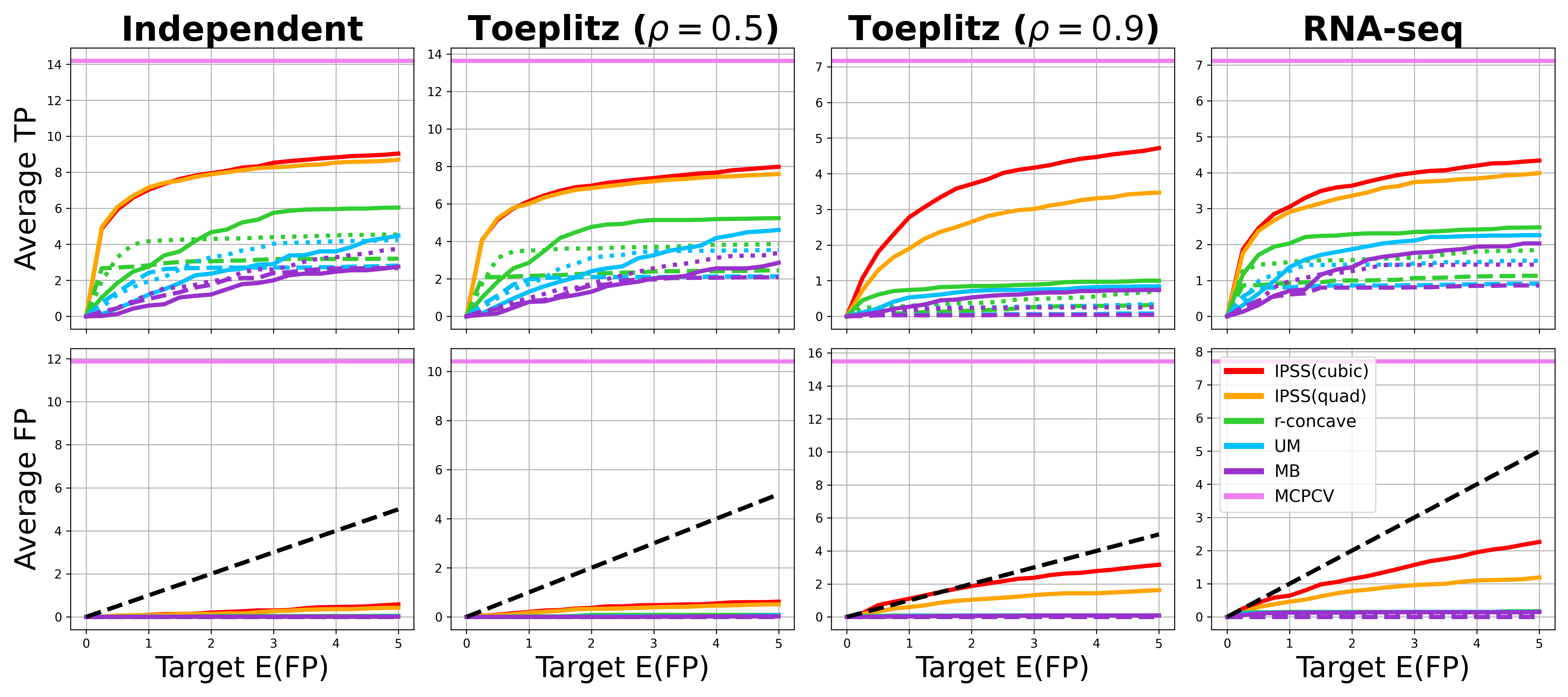}%
\caption{\textit{Linear regression, normal residuals: MCP $(p=1000)$}. See \cref{fig:compare_linear_reg_200_df2}.}
\label{fig:compare_linear_reg_1000_mcp}
\end{figure*}
}{}

\ifthenelse{\boolean{showfigures}}{
\begin{figure*}[!ht]
\includegraphics[width=0.85\textwidth, height=.275\textheight]{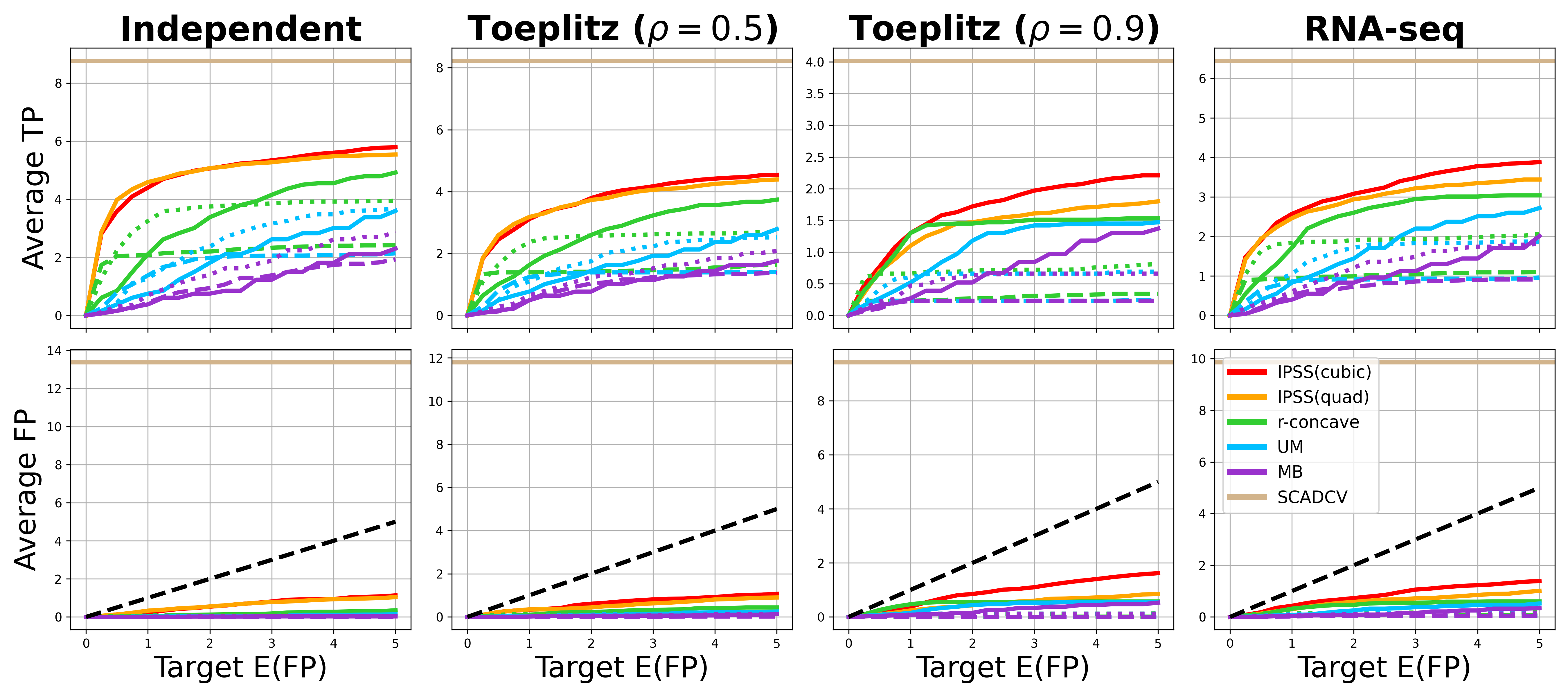}%
\caption{\textit{Linear regression, normal residuals: SCAD $(p=200)$}. See \cref{fig:compare_linear_reg_200_df2}.}
\label{fig:compare_linear_reg_200_scad}
\end{figure*}

\begin{figure*}[!ht]
\includegraphics[width=0.85\textwidth, height=.275\textheight]{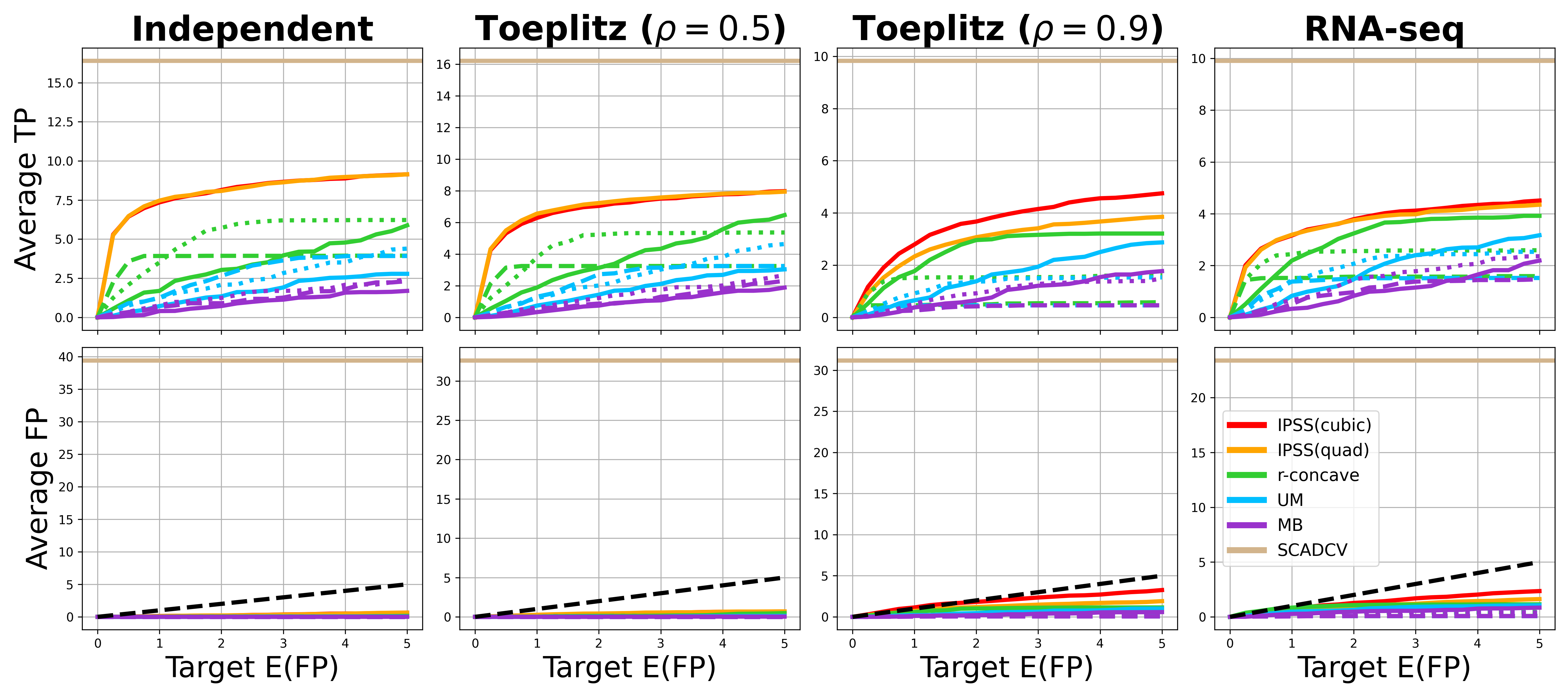}%
\caption{\textit{Linear regression, normal residuals: SCAD $(p=1000)$}. See \cref{fig:compare_linear_reg_200_df2}.}
\label{fig:compare_linear_reg_1000_scad}
\end{figure*}
}{}

\ifthenelse{\boolean{showfigures}}{
\begin{figure*}[!ht]
\includegraphics[width=0.85\textwidth, height=.275\textheight]{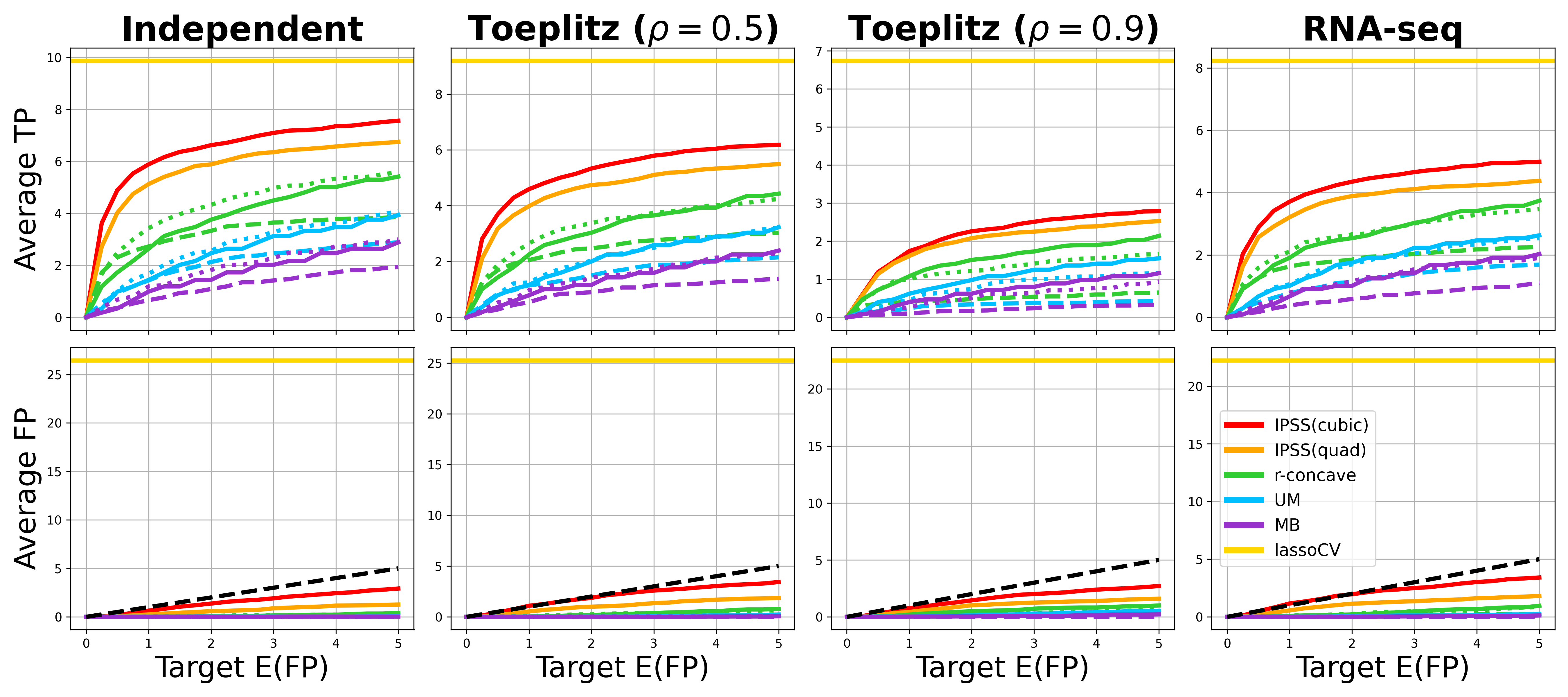}%
\caption{\textit{Linear regression, normal residuals: Adaptive lasso $(p=200)$}. See \cref{fig:compare_linear_reg_200_df2}.}
\label{fig:compare_linear_reg_200_adaptive}
\end{figure*}

\begin{figure*}[!ht]
\includegraphics[width=0.85\textwidth, height=.275\textheight]{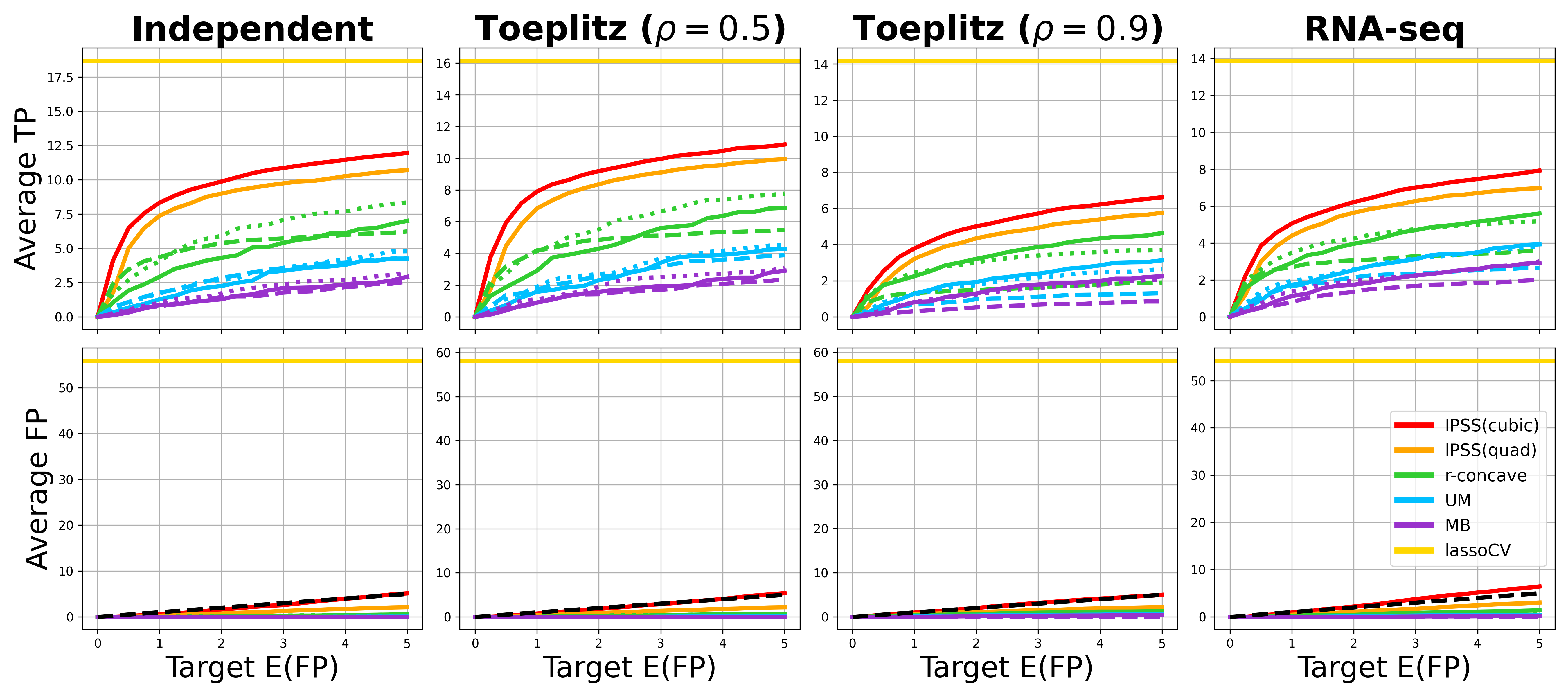}%
\caption{\textit{Linear regression, normal residuals: Adaptive lasso $(p=1000)$}. See \cref{fig:compare_linear_reg_200_df2}.}
\label{fig:compare_linear_reg_1000_adaptive}
\end{figure*}
}{}

\ifthenelse{\boolean{showfigures}}{
\begin{figure*}[!ht]
\includegraphics[width=0.85\textwidth, height=.275\textheight]{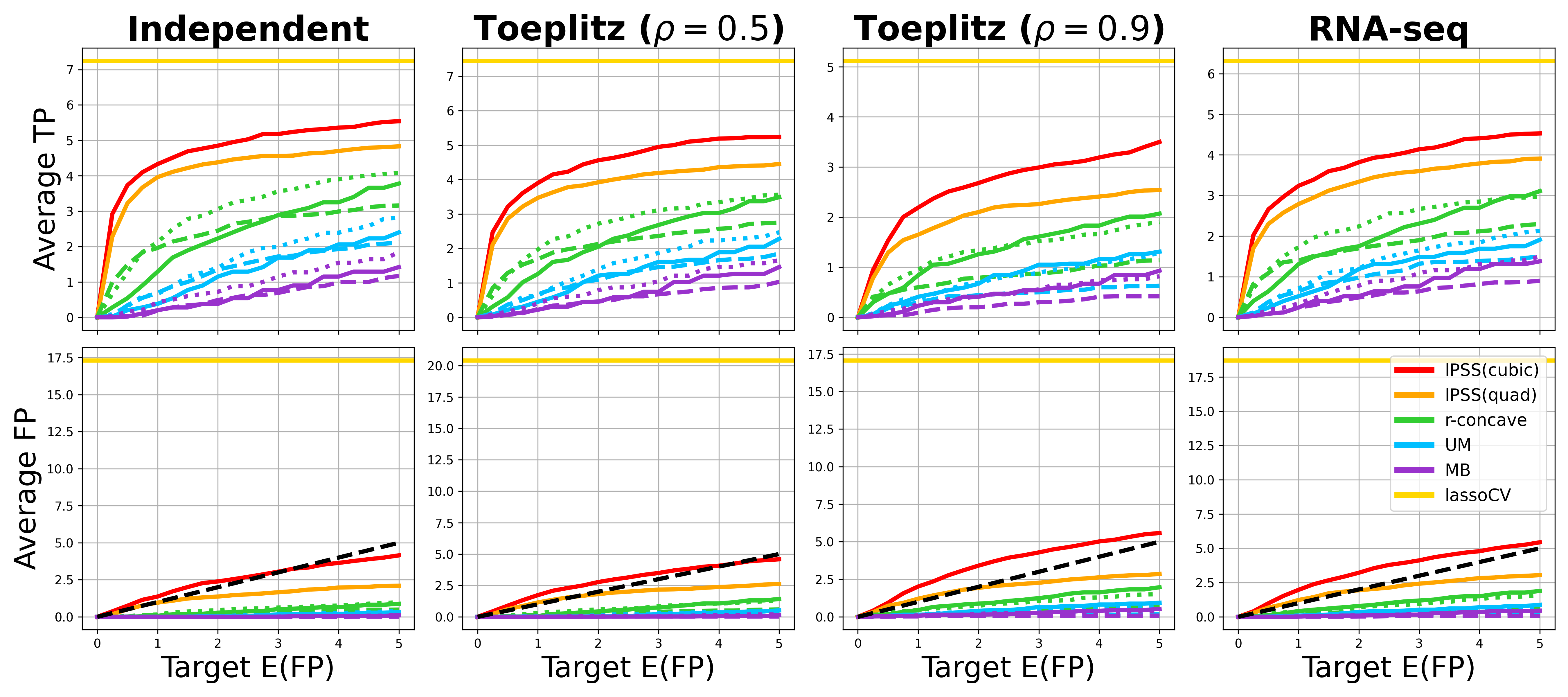}%
\caption{\textit{$\ell_1$-regularized logistic regression $(p=200)$}. See  \cref{fig:compare_linear_reg_200_df2}.}
\label{fig:compare_linear_class_200}
\end{figure*}

\begin{figure*}[!ht]
\includegraphics[width=0.85\textwidth, height=.275\textheight]{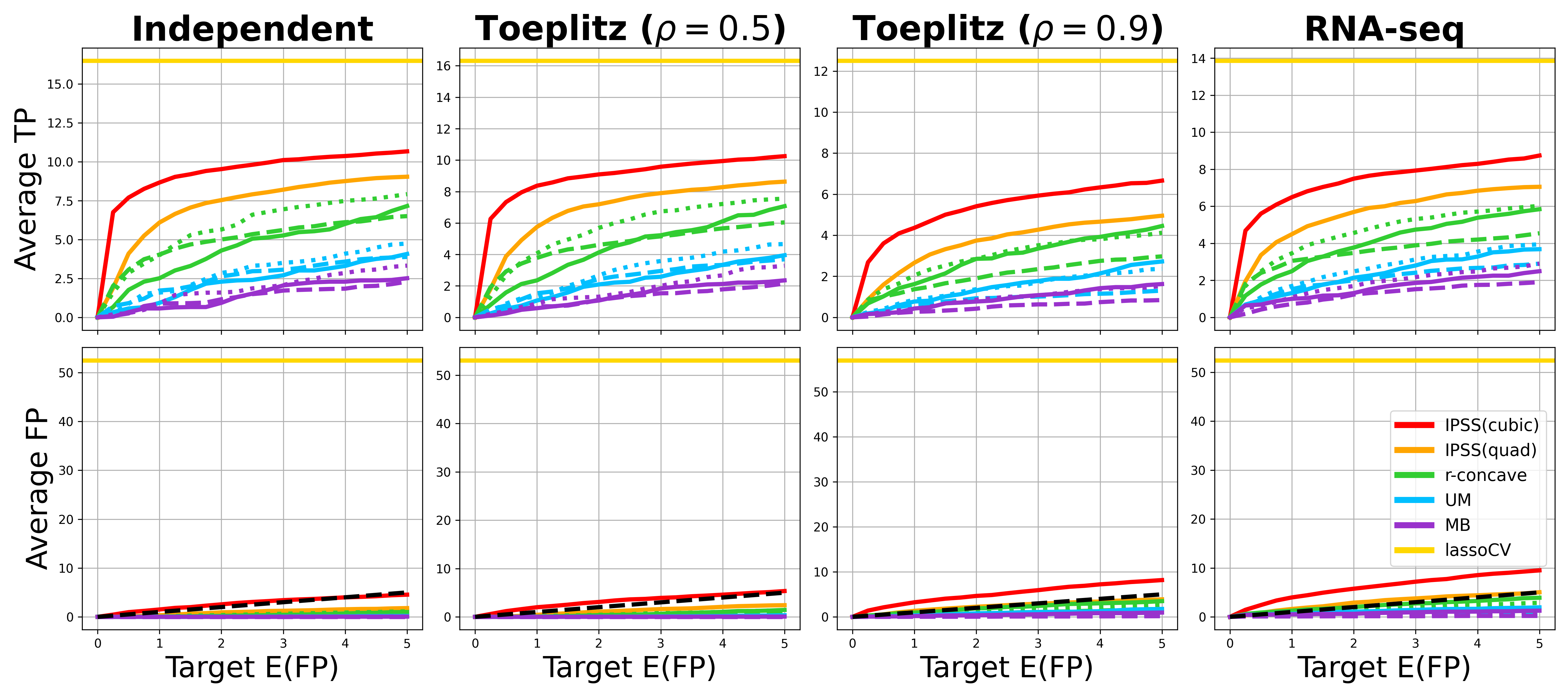}%
\caption{\textit{$\ell_1$-regularized logistic regression $(p=1000)$}. See  \cref{fig:compare_linear_reg_200_df2}.}
\label{fig:compare_linear_class_1000}
\end{figure*}
}{}

\clearpage

%%%%%%%%%%%%%%%%%%%%%%%%%%%%%%%%%%%%%%%%%%%%%%%%%%%%%%%%%%%%%%%%

\section{Sensitivity to parameters}\label{sup_sec:sensitivity}

%%%%%%%%%%%%%%%%%%%%%%%%%%%%%%%%%%%%%%%%%%%%%%%%%%%%%%%%%%%%%%%%

We assess the sensitivity of IPSS to the integral cutoff $C$ (\cref{sup_sec:sensitivity_to_C}) and to the parameter $\alpha$ that defines the probability measure $\mu_\alpha\propto\lambda^{-\alpha}d\lambda$ (\cref{sup_sec:sensitivity_to_mu}).

%%%%%%%%%%%%%%%%%%%%%%%%%%%%%%%%%%%%%%%%%%%%%%%%%%%%%%%%%%%%%%%%

\subsection{Sensitivity to {\boldmath$C$}}\label{sup_sec:sensitivity_to_C}

%%%%%%%%%%%%%%%%%%%%%%%%%%%%%%%%%%%%%%%%%%%%%%%%%%%%%%%%%%%%%%%%

We ran IPSS(quad) and IPSS(cubic) with $C\in\{0.025, 0.05, 0.75, 0.1\}$ for linear regression with normal residuals and logistic regression and for each feature design described in \cref{sec:simulations}. Results are reported for lasso, MCP, SCAD, and $\ell_1$-regularized logistic regression. The $\alpha$ parameter in $\mu_\alpha$ is always set to the default values described in \cref{sec:simulations}. In each plot, different colors correspond to different choices of $C$, and the solid and dashed lines correspond to IPSS(quad) and IPSS(cubic), respectively. The solid, dashed, and dotted gray curves---included for reference---correspond to MB, UM, and $r$-concave, respectively, each with $\tau=0.75$.

The results in Figures \ref{fig:cutoff_linear_reg_200} through \ref{fig:cutoff_linear_class_1000} show that IPSS is highly robust to $C$, with different choices of $C$ often yielding nearly indistinguishable solutions. This is because $C$ directly determines $\lmin$ and, for the range of cutoffs we consider, the corresponding $\lmin$ values are sufficiently small that the stability paths are no longer changing in a meaningful way at these points, especially relative to the changes that occur at much larger regularization values (\cref{fig:stability_paths}). Thus, perturbations to $C$ on the scale we consider do not meaningfully alter the stability paths and therefore the IPSS results.

%%%%%%%%%%%%%%%%%%%%%%%%%%%%%%%%%%%%%%%%%%%%%%%%%%%%%%%%%%%%%%%%

\subsection{Sensitivity to {\boldmath$\alpha$}}\label{sup_sec:sensitivity_to_mu}

%%%%%%%%%%%%%%%%%%%%%%%%%%%%%%%%%%%%%%%%%%%%%%%%%%%%%%%%%%%%%%%%

We ran IPSS(quad) and IPSS(cubic) for $\alpha\in\{0, 1/4, 1/2, 3/4, 1, 5/4\}$ for linear regression with normal residuals and logistic regression and for each feature design described in \cref{sec:simulations}. Results are reported for lasso, MCP, SCAD, and $\ell_1$-regularized logistic regression. The cutoff $C$ is always set to $0.05$. In each plot, different colors correspond to different choices of $\alpha$, and the solid and dashed lines correspond to IPSS(quad) and IPSS(cubic), respectively. The solid, dashed, and dotted gray curves---included for reference---correspond to MB, UM, and $r$-concave, respectively, each with $\tau=0.75$.

The results in Figures \ref{fig:mu_linear_reg_200} through \ref{fig:mu_linear_class_1000} show that IPSS is robust to $\alpha$ in general, with all choices of $\alpha$ leading to more true positives than the stability selection methods, and most choices of $\alpha$ producing accurate E(FP) control in most settings. There are two notable exceptions. First, both IPSS methods significantly exceed the target E(FP) when $\alpha=5/4$ in the $p=1000$ linear regression experiments with lasso as the base estimator. This issue---included here for illustrative purposes and easily avoided by simply not using $\alpha=5/4$ when the base estimator is lasso---occurs because $\alpha=5/4$ puts too much weight on small regularization values, leading to violations of \cref{cond:1}. The second notable exception is the $p=1000$ logistic regression experiments, where IPSS(cubic) with $\alpha\in\{3/4,1,5/4\}$ and, to a lesser extent, IPSS(quad) with $\alpha=5/4$ significantly violate of the E(FP) bound. However, the remaining IPSS and $\alpha$ combinations typically keep the actual E(FP) close to or below target levels.

While $\alpha=1$ works well in general, further improvements in performance are possible with other choices of $\alpha$. Based on our experience, we recommend $\alpha=5/4$ when MCP or SCAD are the base estimators. For IPSS with lasso or adaptive lasso as a base estimator, we recommend $\alpha=1$ if $p\leq 200$, $\alpha=3/4$ if $p\geq 1000$, and $\alpha = -\tfrac{p}{3200} + \tfrac{17}{16}$ for $p\in (200,1000)$. The latter linearly interpolates between $(p,\alpha) = (200, 1)$ and $(1000, 3/4)$. For IPSS with $\ell_1$-regularized logistic regression, we recommend $\alpha=1$ if $p\leq 200$, $\alpha=0$ if $p\geq 1000$, and $\alpha = -\tfrac{p}{800} + \tfrac{5}{4}$  for $p\in (200,1000)$. The latter linearly interpolates between $(p,\alpha) = (200, 1)$ and $(1000, 0)$.

\ifthenelse{\boolean{showfigures}}{
\begin{figure*}
\includegraphics[width=0.8\textwidth, height=.25\textheight]{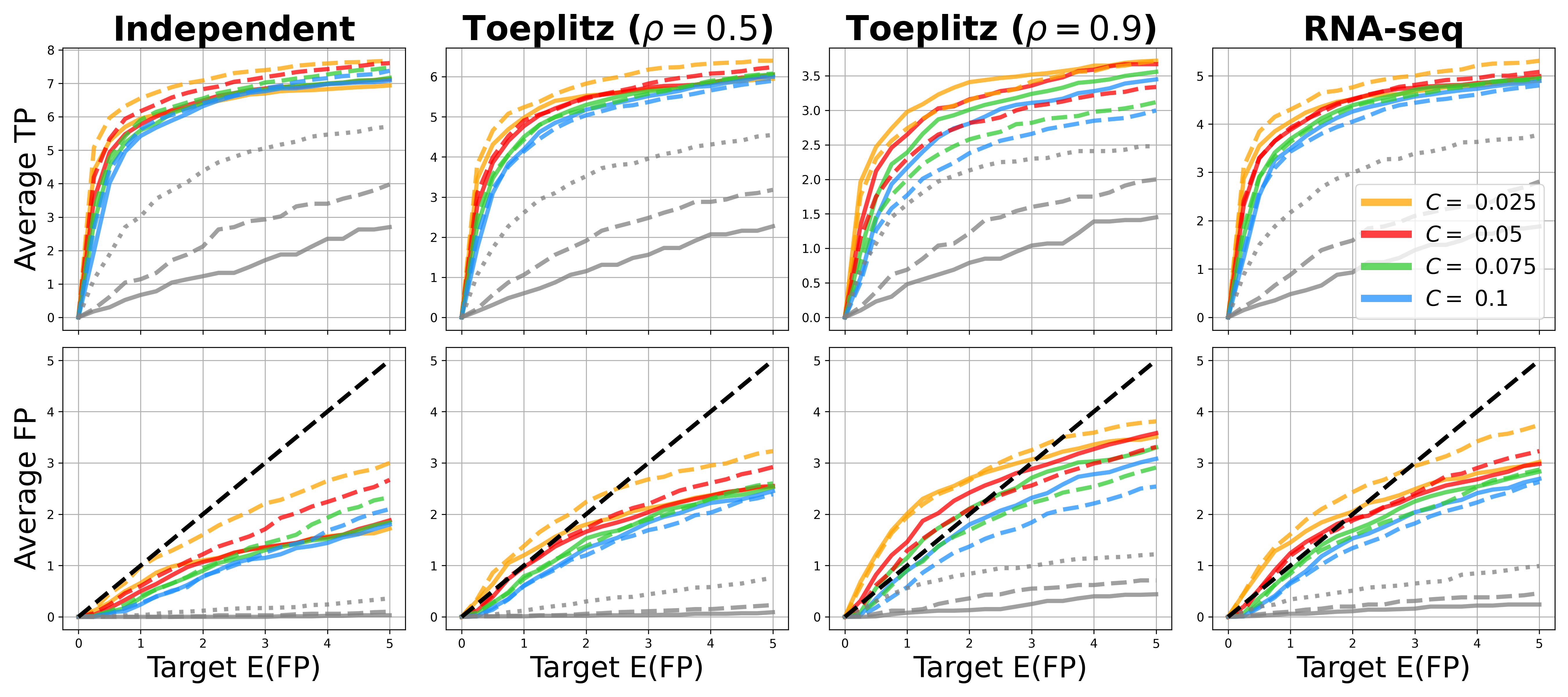}%
\caption{\textit{Sensitivity to $C$: Lasso $(p=200)$}. Colored solid and dashed curves show IPSS(quad) and IPSS(cubic) for different choices of $C$, respectively. Gray solid, dashed, and dotted gray curves show MB, UM, and $r$-concave stability selection with $\tau=0.75$. The dashed black line represents perfect E(FP) control.}
\label{fig:cutoff_linear_reg_200}
\end{figure*}

\begin{figure*}
\includegraphics[width=0.8\textwidth, height=.25\textheight]{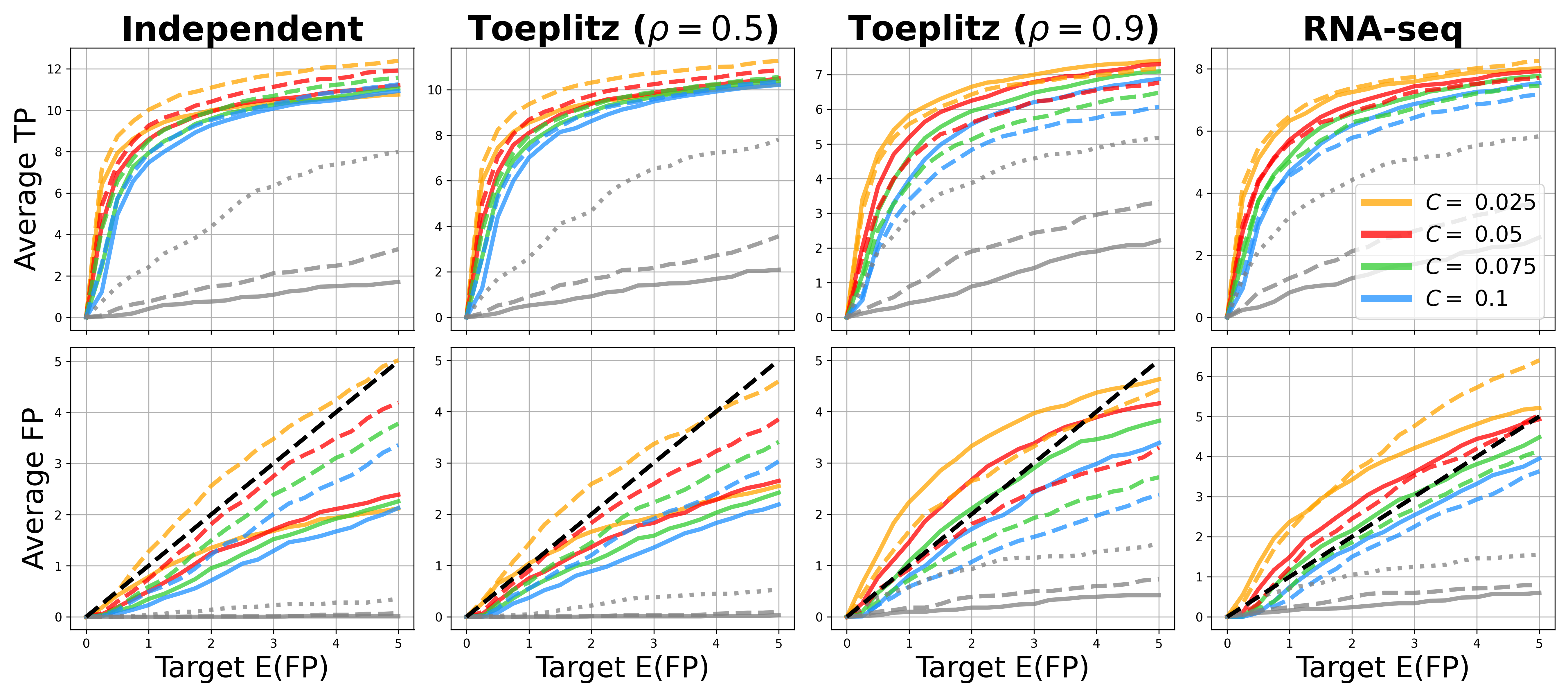}%
\caption{\textit{Sensitivity to $C$: Lasso $(p=1000)$}. See \cref{fig:cutoff_linear_reg_200}.}
\label{fig:cutoff_linear_reg_1000}
\end{figure*} 

\begin{figure*}
\includegraphics[width=0.8\textwidth, height=.25\textheight]{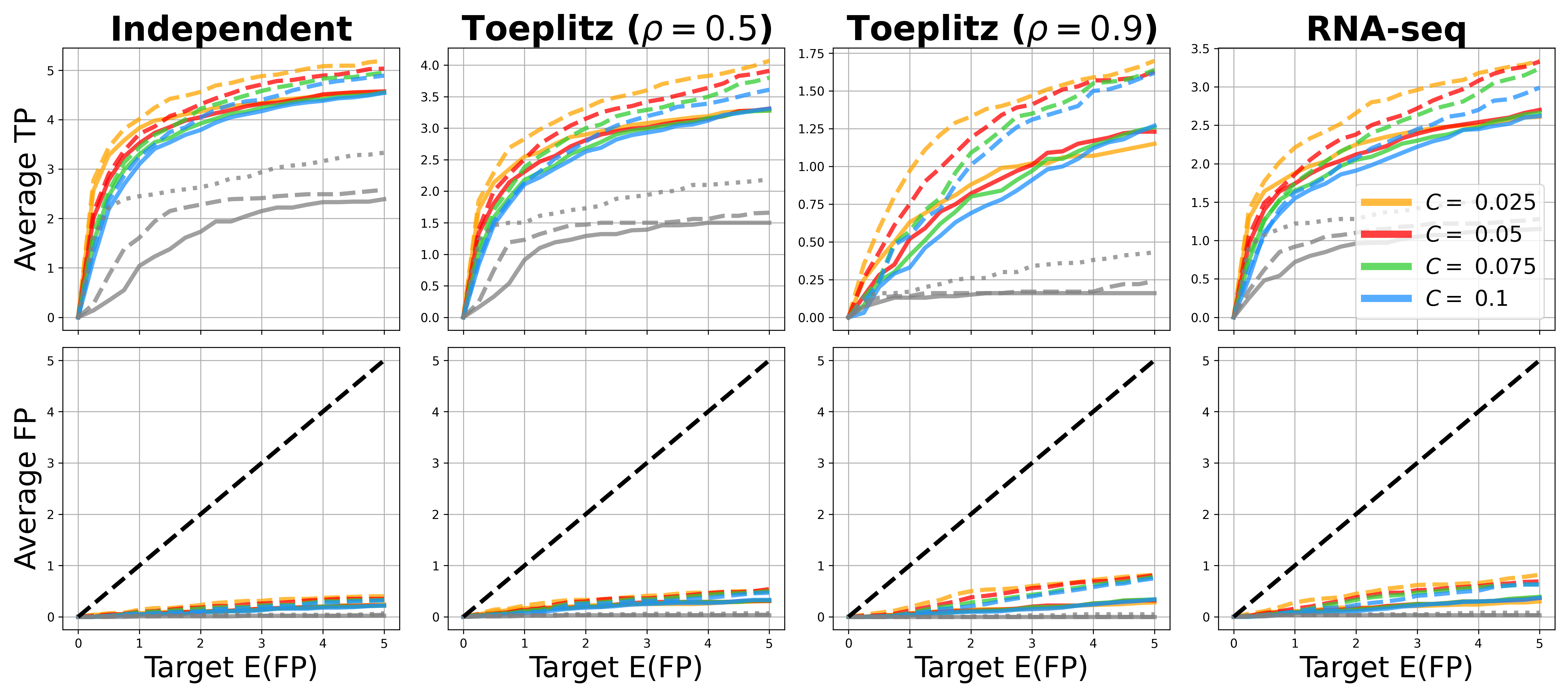}%
\caption{\textit{Sensitivity to $C$: MCP $(p=200)$}. See \cref{fig:cutoff_linear_reg_200}.}
\label{fig:cutoff_linear_reg_200_mcp}
\end{figure*}

\begin{figure*}
\includegraphics[width=0.8\textwidth, height=.25\textheight]{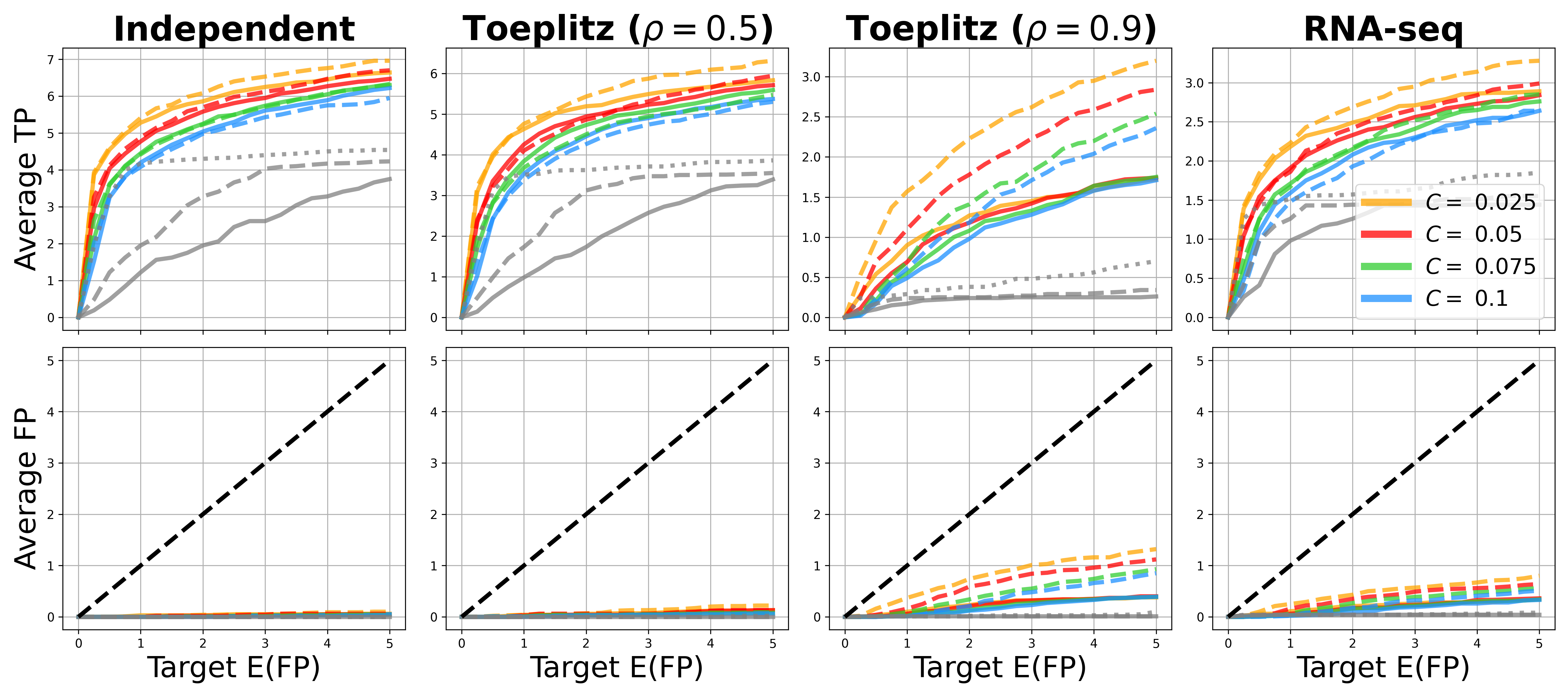}%
\caption{\textit{Sensitivity to $C$: MCP $(p=1000)$}. See \cref{fig:cutoff_linear_reg_200}.}
\label{fig:cutoff_linear_reg_1000_mcp}
\end{figure*} 

\begin{figure*}
\includegraphics[width=0.8\textwidth, height=.25\textheight]{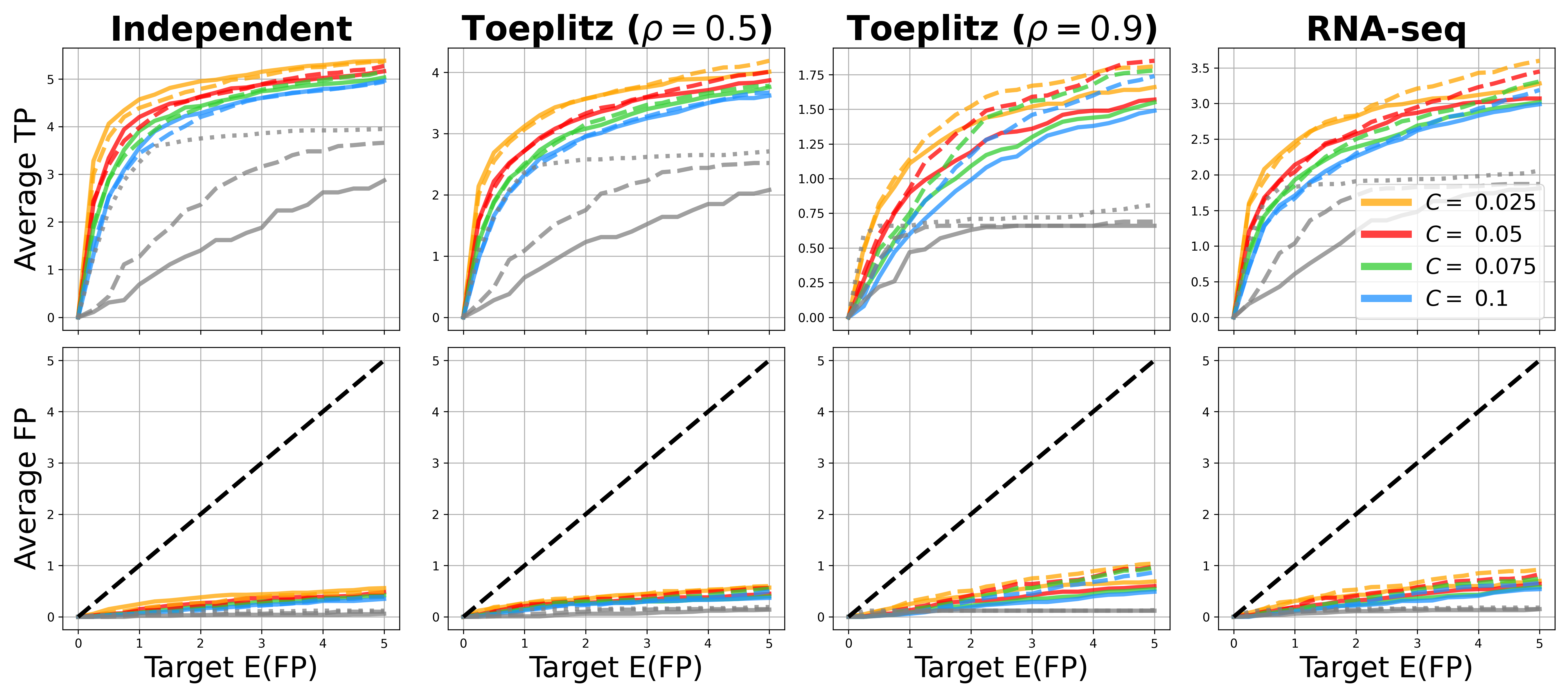}%
\caption{\textit{Sensitivity to $C$: SCAD $(p=200)$}. See \cref{fig:cutoff_linear_reg_200}.}
\label{fig:cutoff_linear_reg_200_scad}
\end{figure*}

\begin{figure*}
\includegraphics[width=0.8\textwidth, height=.25\textheight]{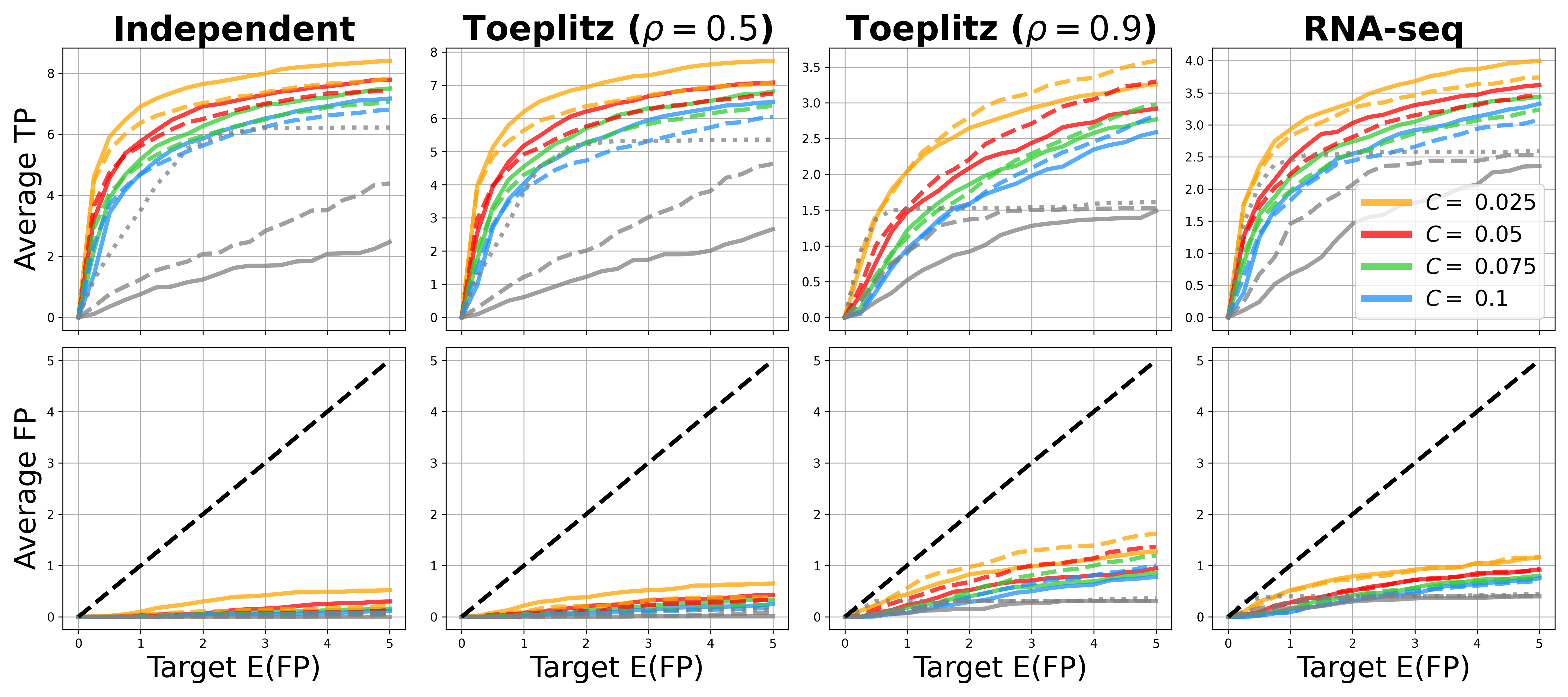}%
\caption{\textit{Sensitivity to $C$: SCAD $(p=1000)$}. See \cref{fig:cutoff_linear_reg_200}.}
\label{fig:cutoff_linear_reg_1000_scad}
\end{figure*} 

\begin{figure*}
\includegraphics[width=0.8\textwidth, height=.25\textheight]{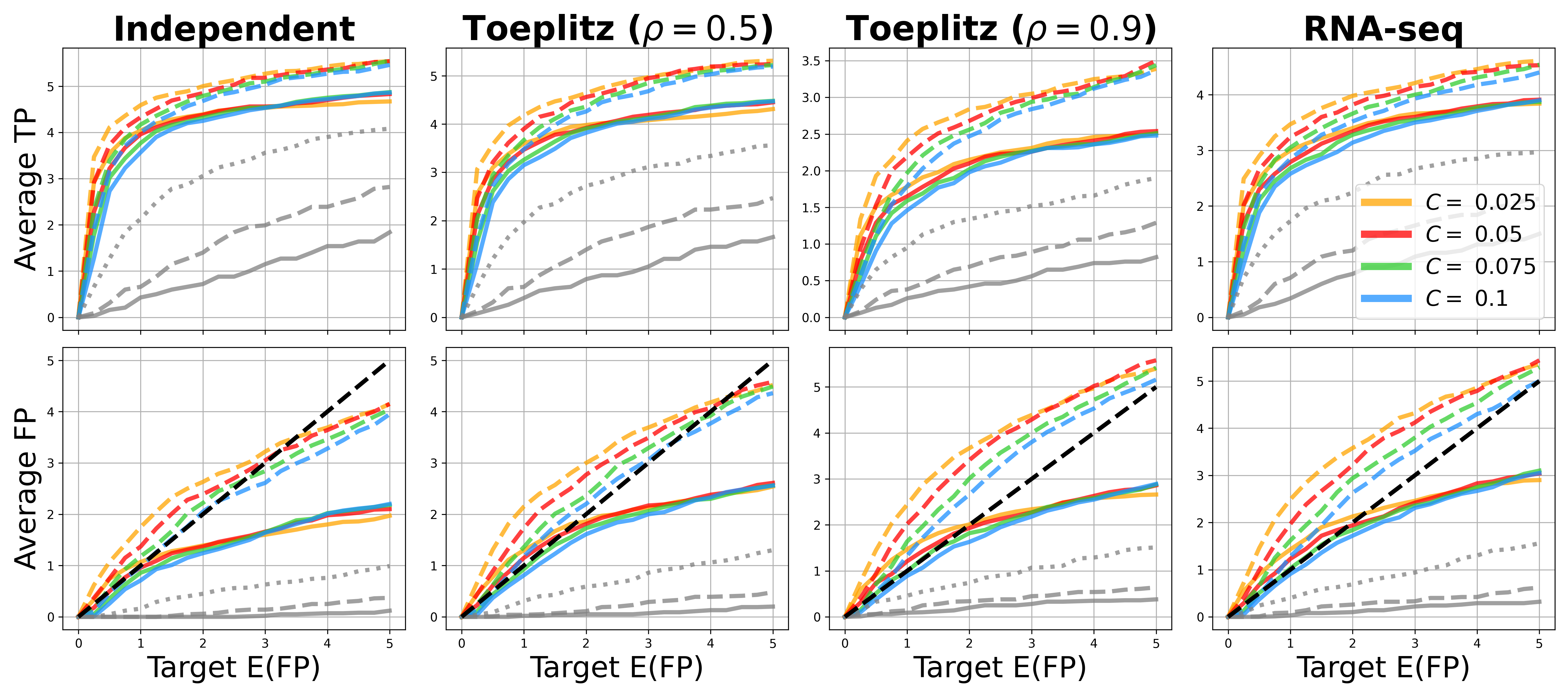}%
\caption{\textit{Sensitivity to $C$: Logistic regression $(p=200)$}. See \cref{fig:cutoff_linear_reg_200}.}
\label{fig:cutoff_linear_class_200}
\end{figure*}

\begin{figure*}
\includegraphics[width=0.8\textwidth, height=.25\textheight]{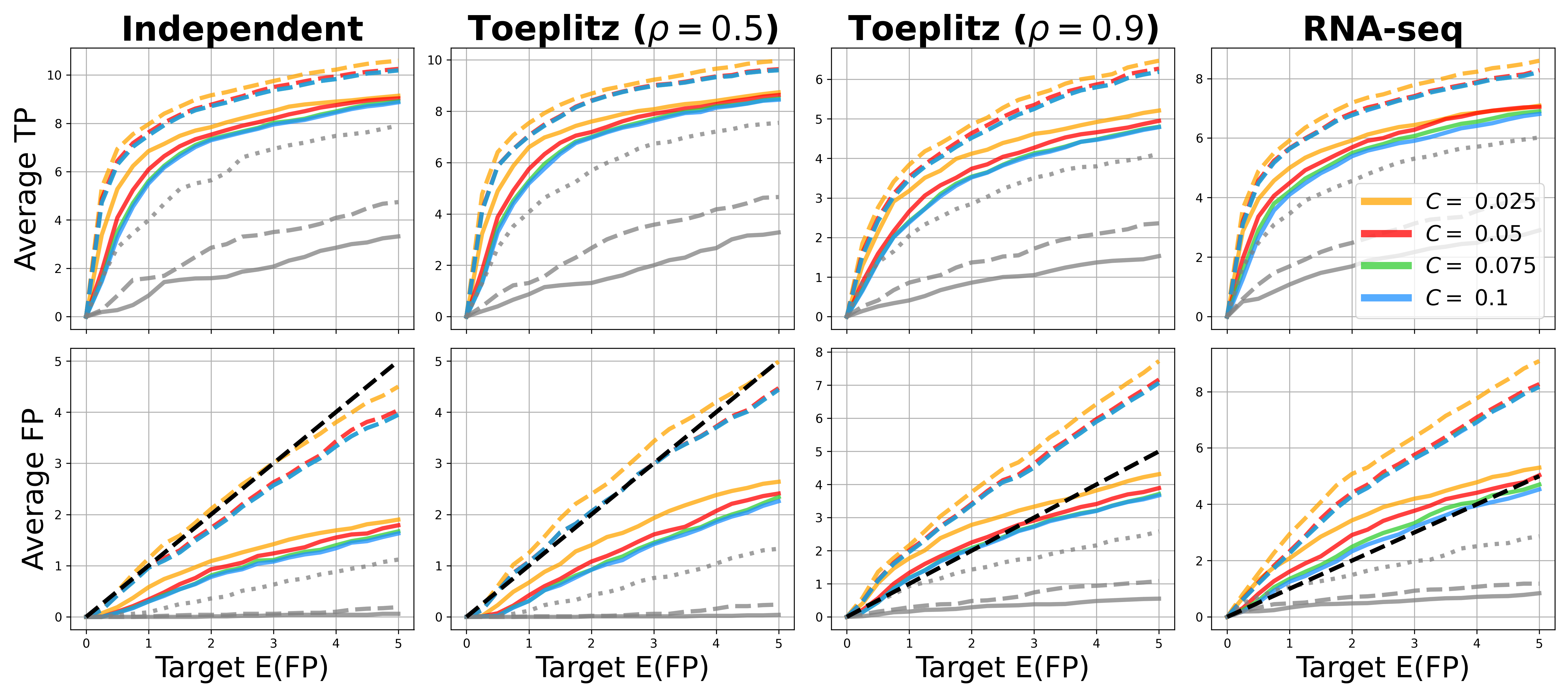}%
\caption{\textit{Sensitivity to $C$: Logistic regression $(p=1000)$}. See \cref{fig:cutoff_linear_reg_200}.}
\label{fig:cutoff_linear_class_1000}
\end{figure*} 
}{}

\ifthenelse{\boolean{showfigures}}{
\begin{figure*}
\includegraphics[width=0.8\textwidth, height=.25\textheight]{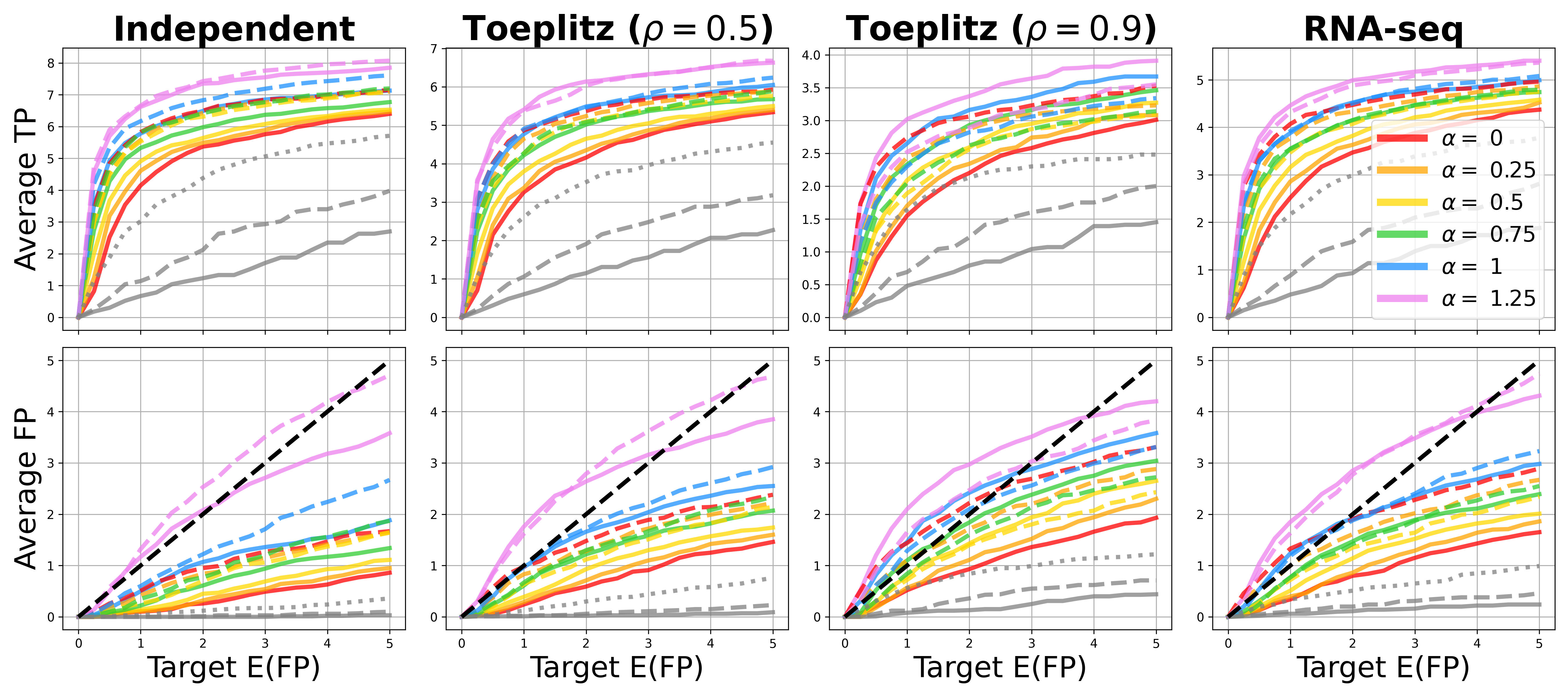}%
\caption{\textit{Sensitivity to $\alpha$: Lasso $(p=200)$}. Colored solid and dashed curves show IPSS(quad) and IPSS(cubic) for different choices of $\alpha$, respectively. Gray solid, dashed, and dotted gray curves show MB, UM, and $r$-concave stability selection with $\tau=0.75$. The dashed black line represents perfect E(FP) control.}
\label{fig:mu_linear_reg_200}
\end{figure*}

\begin{figure*}
\includegraphics[width=0.8\textwidth, height=.25\textheight]{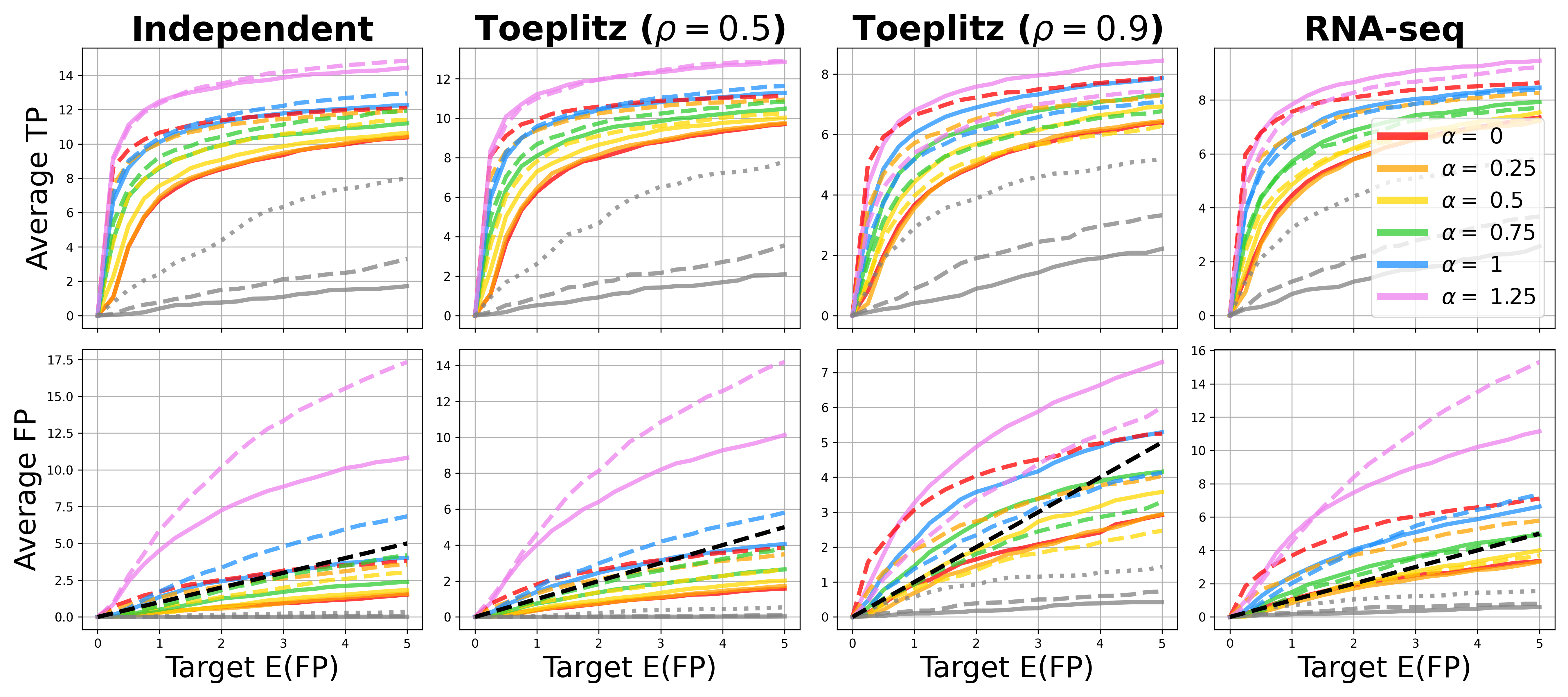}%
\caption{\textit{Sensitivity to $\alpha$: Lasso, $(p=1000)$}. See \cref{fig:mu_linear_reg_200}.}
\label{fig:mu_linear_reg_1000}
\end{figure*} 

\begin{figure*}
\includegraphics[width=0.8\textwidth, height=.25\textheight]{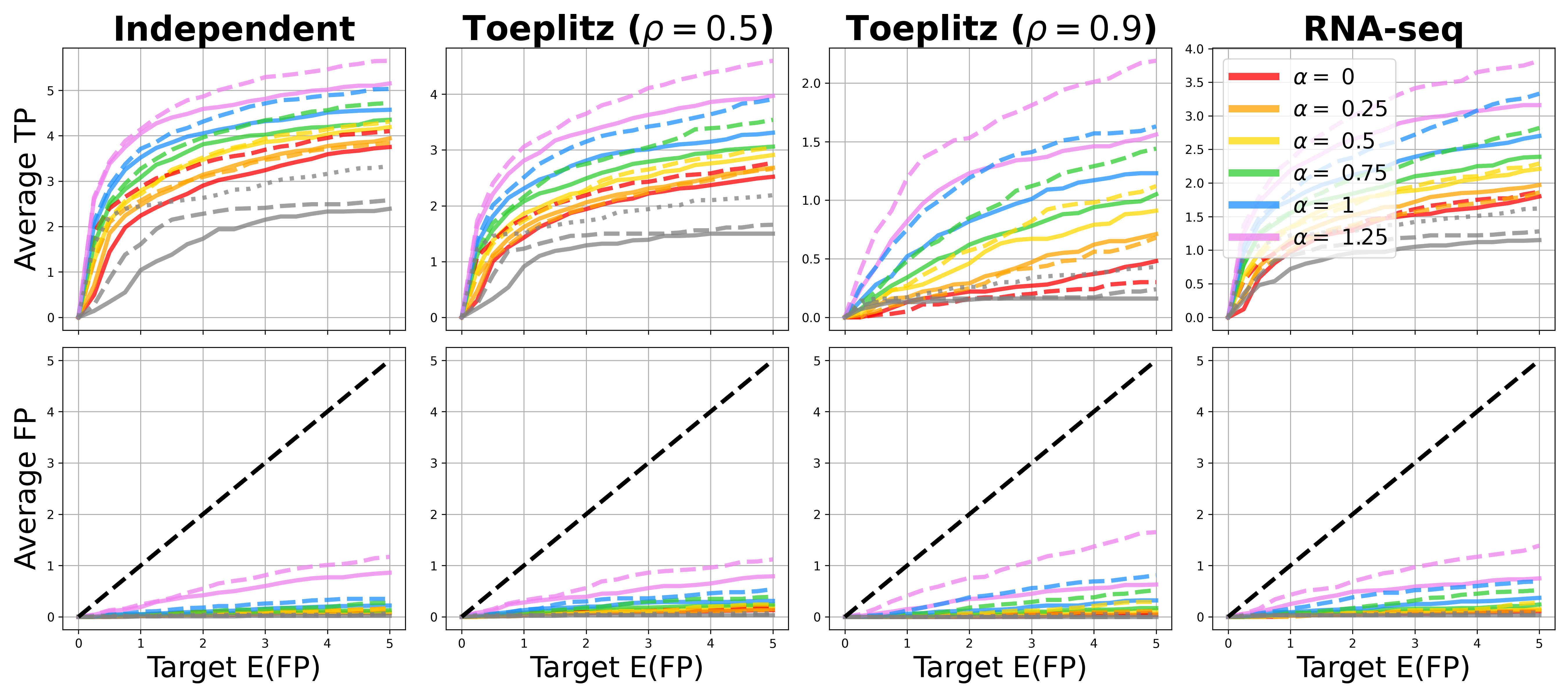}%
\caption{\textit{Sensitivity to $\alpha$: MCP $(p=200)$}. See \cref{fig:mu_linear_reg_200}.}
\label{fig:mu_linear_reg_200_mcp}
\end{figure*}

\begin{figure*}
\includegraphics[width=0.8\textwidth, height=.25\textheight]{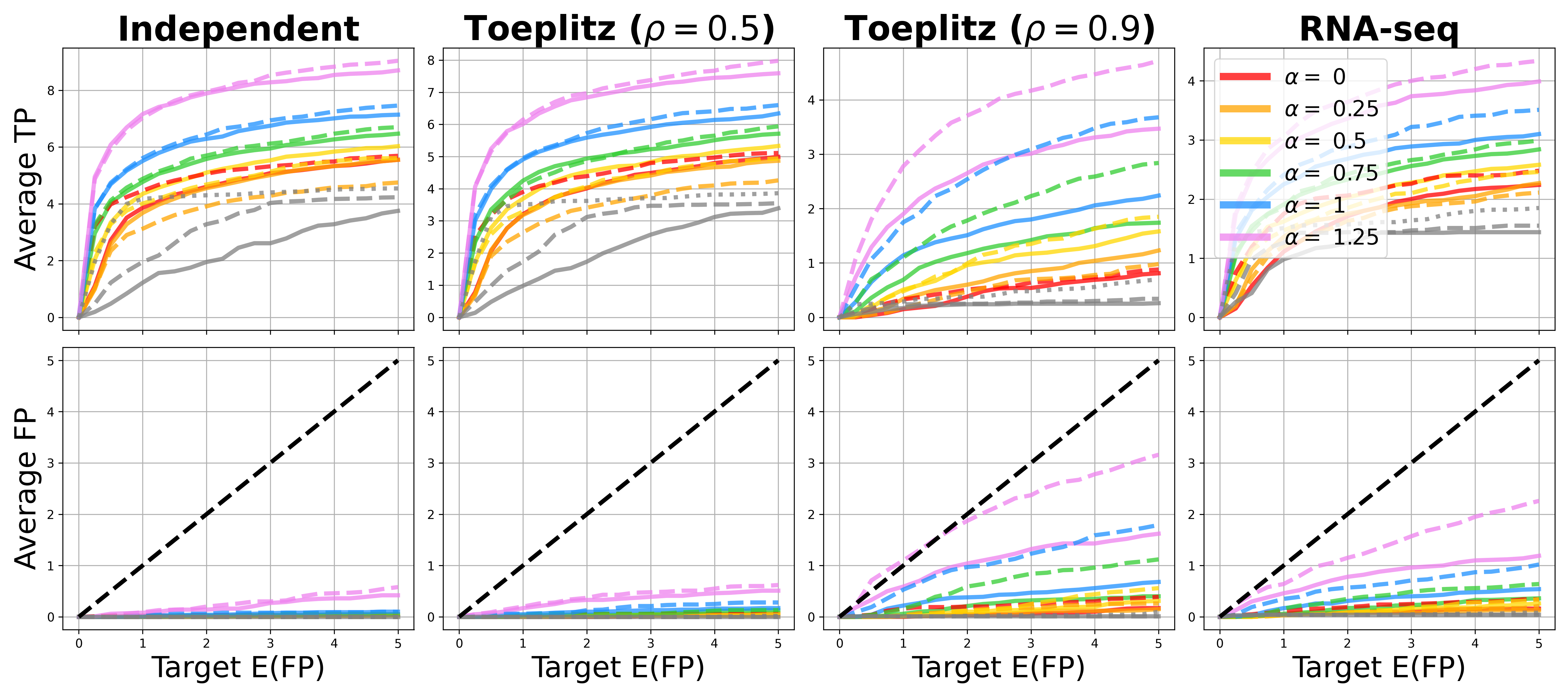}%
\caption{\textit{Sensitivity to $\alpha$: MCP $(p=1000)$}. See \cref{fig:mu_linear_reg_200}.}
\label{fig:mu_linear_reg_1000_mcp}
\end{figure*} 

\begin{figure*}
\includegraphics[width=0.8\textwidth, height=.25\textheight]{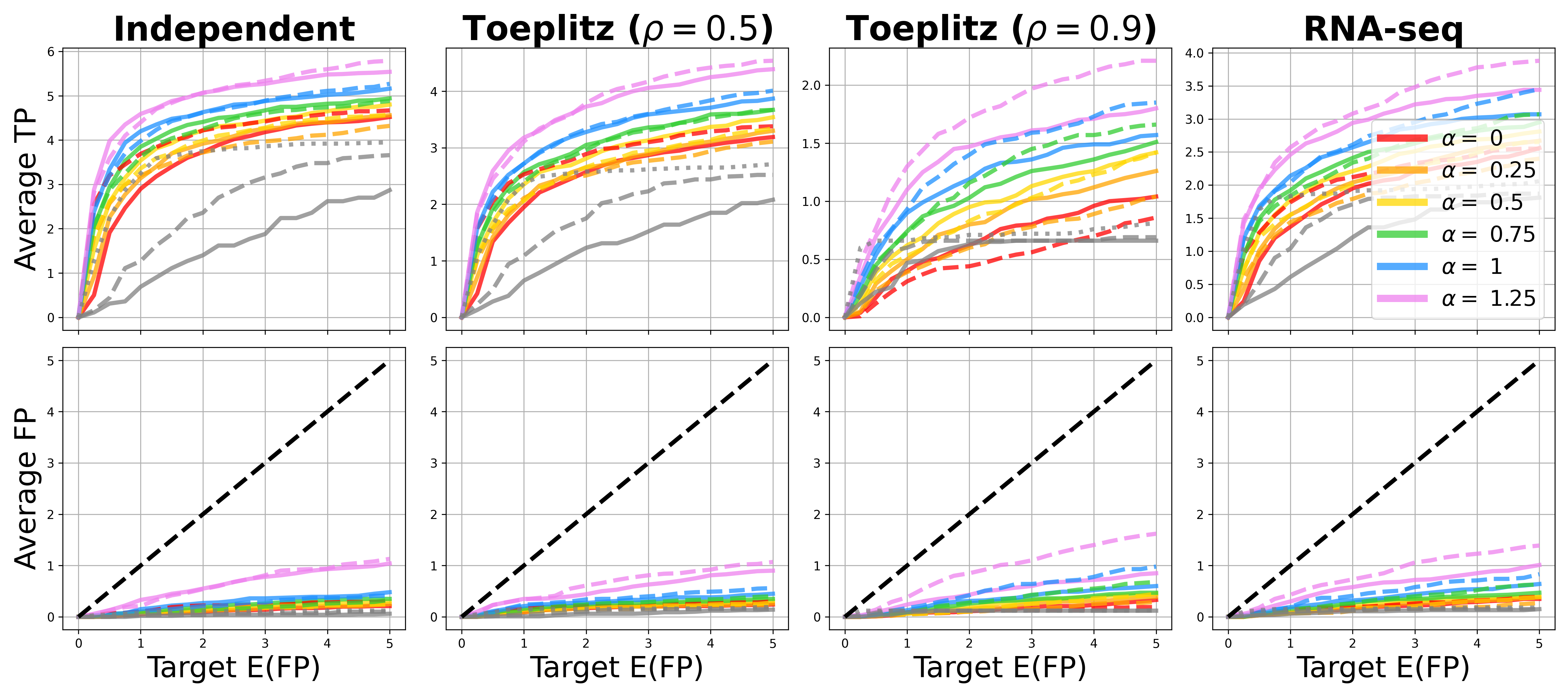}%
\caption{\textit{Sensitivity to $\alpha$: SCAD $(p=200)$}. See \cref{fig:mu_linear_reg_200}.}
\label{fig:mu_linear_reg_200_scad}
\end{figure*}

\begin{figure*}
\includegraphics[width=0.8\textwidth, height=.25\textheight]{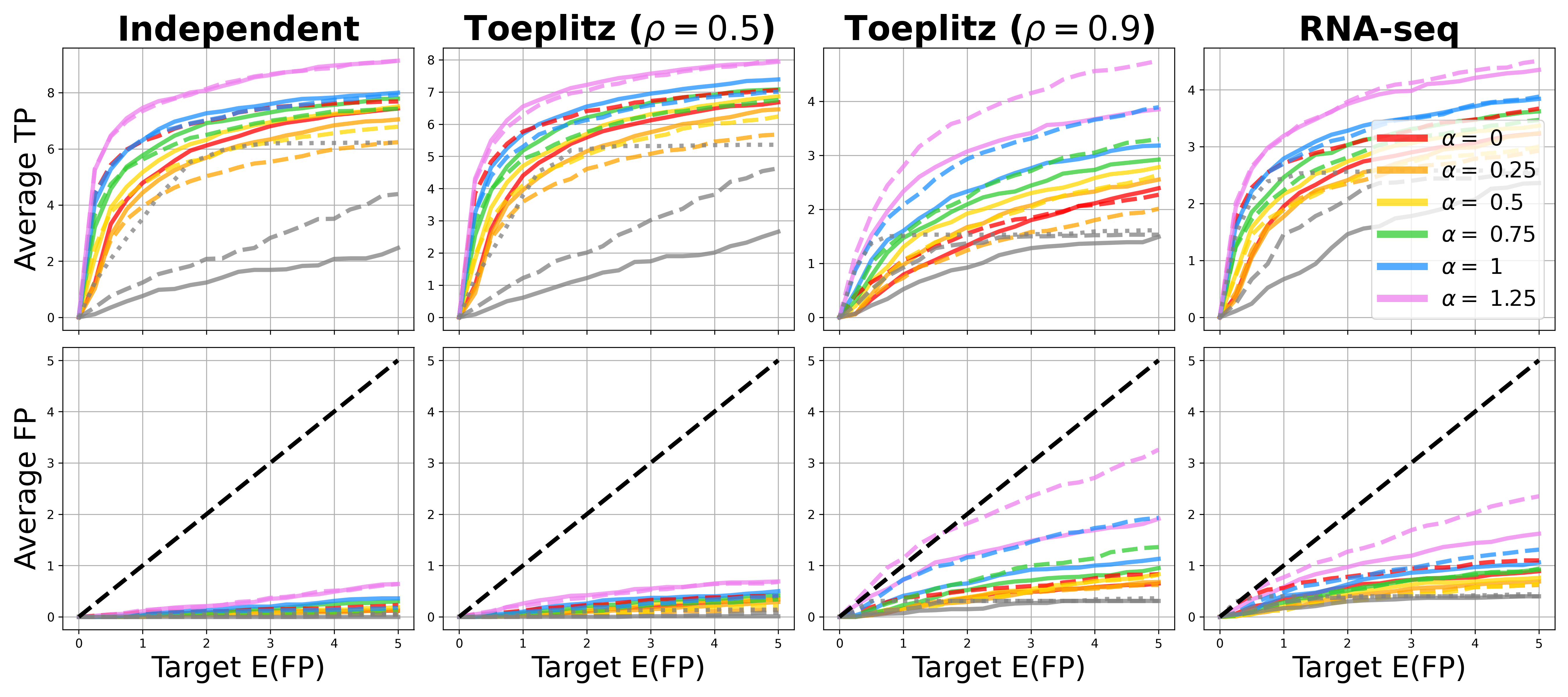}%
\caption{\textit{Sensitivity to $\alpha$: SCAD $(p=1000)$}. See \cref{fig:mu_linear_reg_200}.}
\label{fig:mu_linear_reg_1000_scad}
\end{figure*} 

\begin{figure*}
\includegraphics[width=0.8\textwidth, height=.25\textheight]{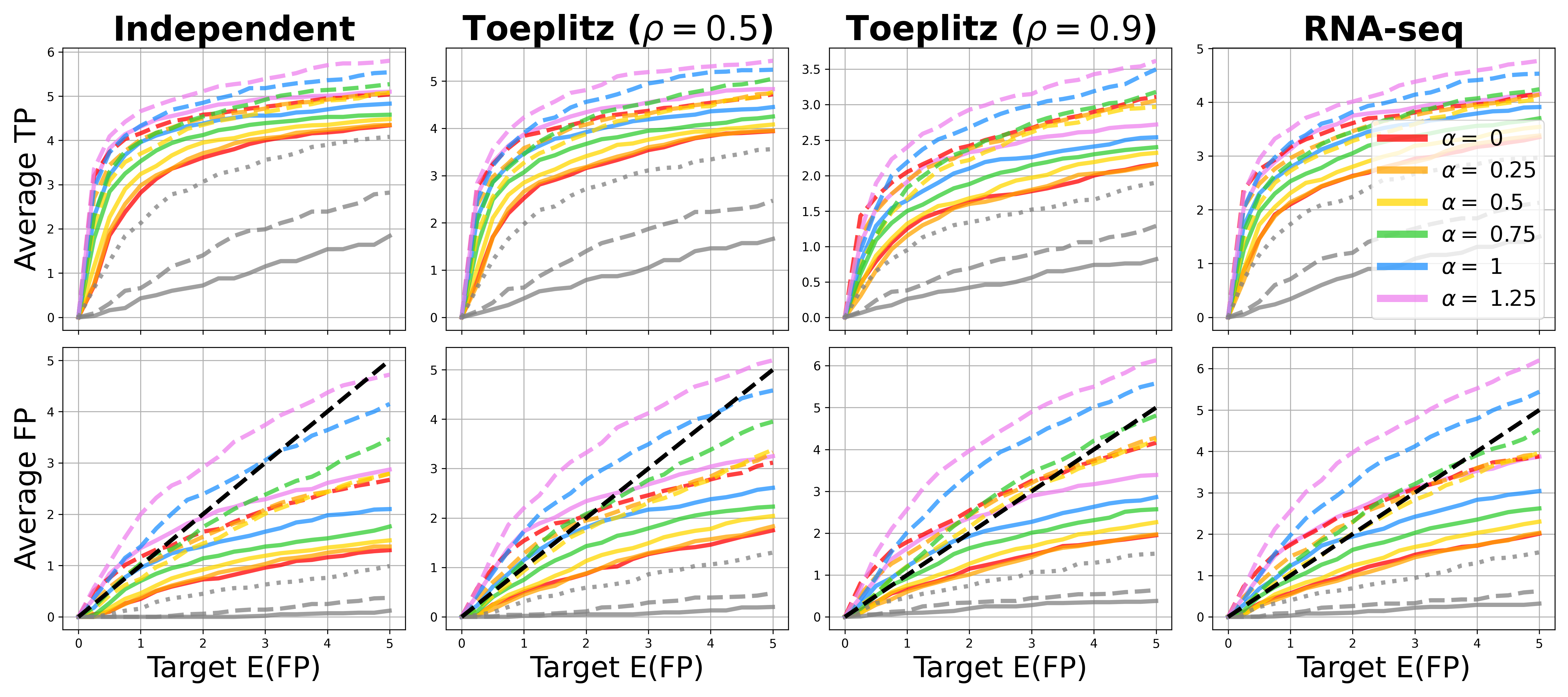}%
\caption{\textit{Sensitivity to $\alpha$: Logistic regression $(p=200)$}. See \cref{fig:mu_linear_reg_200}.}
\label{fig:mu_linear_class_200}
\end{figure*}

\begin{figure*}
\includegraphics[width=0.8\textwidth, height=.25\textheight]{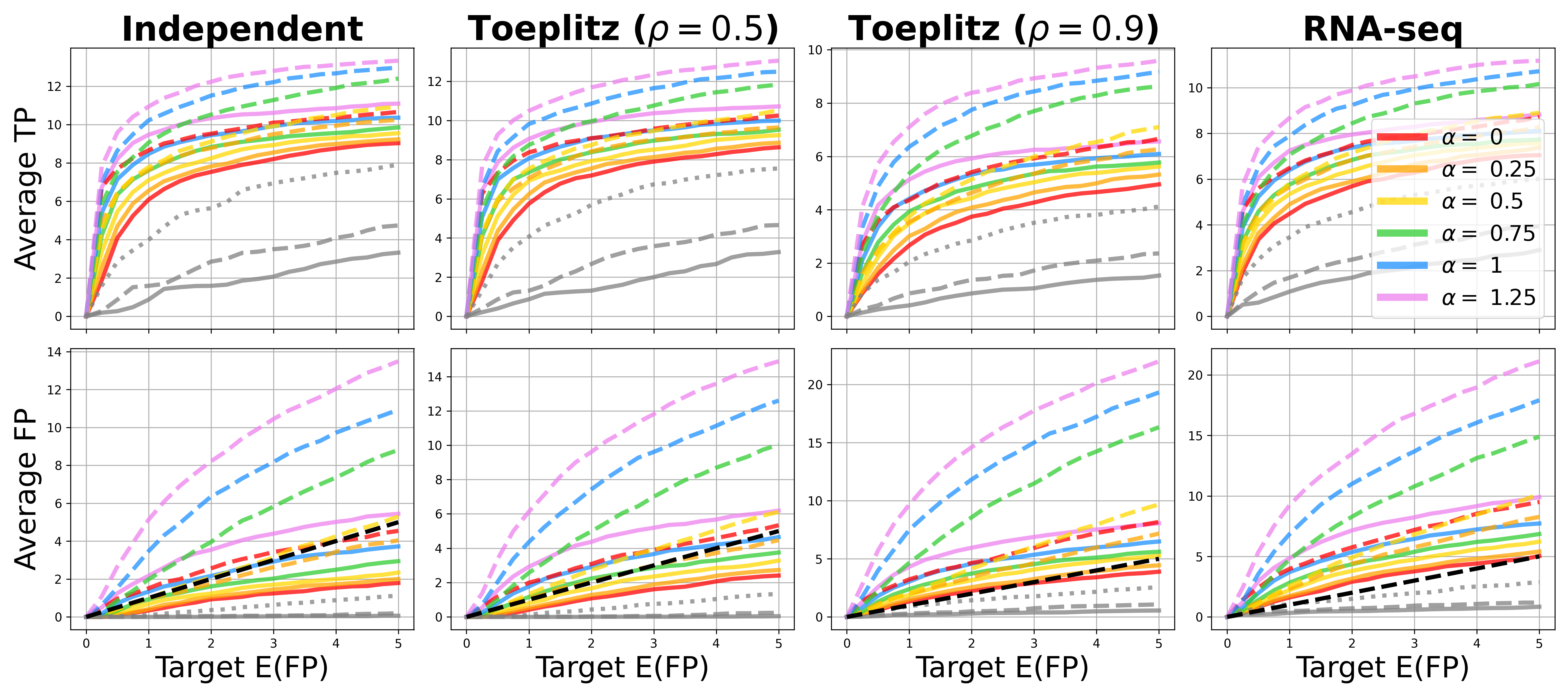}%
\caption{\textit{Sensitivity to $\alpha$: Logistic regression $(p=1000)$}. See \cref{fig:mu_linear_reg_200}.}
\label{fig:mu_linear_class_1000}
\end{figure*} 
}{}

\end{document}